\documentclass{article}

\usepackage{arxiv}

\usepackage[utf8]{inputenc} % allow utf-8 input
\usepackage[T1]{fontenc}    % use 8-bit T1 fonts
\usepackage{hyperref}       % hyperlinks
\usepackage{url}            % simple URL typesetting
\usepackage{booktabs}       % professional-quality 
\usepackage{natbib}
\usepackage{tabularx}
\usepackage{amsfonts}       % blackboard math symbols
\usepackage{nicefrac}       % compact symbols for 1/2, etc.
\usepackage{microtype}      % microtypography
\usepackage{lipsum}
\usepackage{graphicx}
\graphicspath{ {./images/} }

\usepackage{amsmath,amssymb}
\usepackage{tikz}
\usepackage{subcaption}

\usepackage[ruled,vlined,linesnumbered,algo2e]{algorithm2e}

\usepackage[ruled,vlined]{algorithm2e}

\usepackage[utf8]{inputenc} % allow utf-8 input
\usepackage[T1]{fontenc}    % use 8-bit T1 fonts
\usepackage{hyperref}       % hyperlinks
\usepackage{url}            % simple URL typesetting
\usepackage{booktabs}       % professional-quality tables
\usepackage{amsfonts}       % blackboard math symbols
\usepackage{nicefrac}       % compact symbols for 1/2, etc.
\usepackage{microtype}      % microtypography
\usepackage{xcolor}         % colors
\usepackage{mathrsfs}

\newcommand{\R}{\mathbb{R}}

\newcommand{\E}{ \mathbb{E}}
\newcommand{\pto}{\overset{p}{\to}}

\newcommand{\one}{\bold{1}}

\usepackage{amsthm}

\newtheorem{theorem}{Theorem}[section] % Numbered within sections
\newtheorem{lemma}[theorem]{Lemma}     % Numbered along with theorems
\newtheorem{corollary}[theorem]{Corollary} % Numbered along with theorems
\newtheorem{proposition}[theorem]{Proposition} % Numbered along with theorems
\newtheorem{remark}[theorem]{Remark} % Numbered along with theorems
\newtheorem{assumption}{Assumption}
\newtheorem{definition}{Definition}
\usepackage{tikz}
\usetikzlibrary{positioning,shapes,arrows.meta}

\newcommand{\var}{\mathrm{Var}}

\newcommand{\ts}{\mathrm{ts}}
\newcommand{\avar}{\mathrm{aVar}}

\newcommand{\op}{\mathrm{op}}

\newcommand{\rank}{\mathrm{rank}}
\newcommand{\Bern}{\mathrm{Bern}}

\title{Auditing the Auditors: 
Does Community-based Moderation Get It Right?}

\author{
Yeganeh Alimohammadi\textsuperscript{1}\thanks{Marshall School of Business, University of Southern California. \texttt{yalimoha@usc.edu}}\\
University of Southern California
\And
Karissa Huang\textsuperscript{1}\thanks{Department of Statistics, UC Berkeley. \texttt{krhuang@berkeley.edu}}\\
UC Berkeley
\And
Christian Borgs\thanks{Bakar Institute of Digital Materials for the Planet and Department of Electrical Engineering and Computer Sciences, UC Berkeley. \texttt{borgs@berkeley.edu}}\\
UC Berkeley
\And
Jennifer Chayes\thanks{Department of Statistics, Department of Mathematics, School of Information, Department of Electrical Engineering and Computer Sciences, and Bakar Institute of Digital Materials for the Planet, UC Berkeley. \texttt{jchayes@berkeley.edu}}\\
UC Berkeley
}
\begin{document}
\maketitle
\footnotetext[1]{Equal contribution.}

\begin{abstract}
Online social platforms increasingly rely on crowd-sourced systems to label misleading content at scale,  but these systems must both aggregate users' evaluations and decide whose evaluations to trust. 
To address the latter, many platforms audit users  by rewarding agreement with the final aggregate outcome, a design we term consensus-based auditing. 
We analyze the consequences of this design in X’s Community Notes, which in September 2022 adopted consensus-based auditing that ties  users' eligibility for participation to agreement with the eventual platform outcome.
We find evidence of strategic conformity: minority contributors’ evaluations drift toward the majority and their participation share falls on controversial topics, where independent signals matter most.
We formalize this mechanism in a behavioral model in which contributors trade off private beliefs against anticipated penalties for disagreement.  Motivated by these findings, we propose a two-stage auditing and aggregation algorithm that weights contributors by the stability of their past residuals rather than by agreement with the majority. 
The method first accounts for differences across content and contributors, and then measures how predictable each contributor’s evaluations are relative to the latent-factor model. Contributors whose evaluations are consistently informative receive greater influence in aggregation, even when they disagree with the prevailing consensus. 
In the Community Notes data, this approach improves out-of-sample predictive performance while avoiding penalization of disagreement.
\end{abstract}

In the face of rapidly increasing online misinformation and harmful content \cite{vicario2016spreading}, online platforms face a fundamental design question: 
\textbf{how can unreliable information be identified and flagged at scale?} Many platforms have turned this question back to their users, creating crowdsourced systems of content moderation that seek to leverage a diverse user base. Notably, X (formerly Twitter)--and, at a more experimental stage, Meta, TikTok, and Bluesky--invite users to evaluate content and add context to potentially unreliable posts 
\cite{xcommunitynotes, meta2025communitynotes, perez2024bluesky, tiktok2025footnotes}.

A central challenge in these systems is that users vary widely in the reliability of their evaluations \cite{amatriain2009ilikeit, bhuiyan2020investigating}.
As a result, platforms confront a second problem: \textbf{how to audit the auditors
themselves?}
In practice, platforms aggregate user content evaluations into an inferred platform outcome, which we refer to as the \emph{consensus}, that determines whether content is downranked or given additional context \cite{CommunityNotesRankingNotes}. The same consensus is then used to evaluate users' reliability as well: users whose historical input to the system aligns with the consensus gain influence, whereas those who diverge are downweighted or lose eligibility to participate in future evaluations %community evaluation 
\cite{CommunityNotesWritingAbility, CommunityNotesWritingAndRatingImpact}. We refer to this design choice as \emph{consensus-based auditing}.

Although consensus-based auditing appears operationally efficient, it implicitly assumes that consensus is a reliable proxy for truth and that disagreement signals low credibility. This coupling of aggregation and auditing can create a feedback effect and alter reporting incentives \cite{horvitz2012incentives, perdomo2020performative}. When agreement with the consensus is rewarded, users have incentives to anticipate the outcome rather than provide independent input \cite{muchnik2013social, lorenz2011social}.
As a result,
disagreement becomes less visible, and minority perspectives may be  withdrawn
before they are ever aggregated \cite{zhao2025spiral,noelle1974spiral}.
 In this paper, we quantify these effects and show, both theoretically and empirically, that consensus-based auditing systematically distorts
contributor behavior and reduces the representation of informative minority viewpoints.

X’s Community Notes provides a concrete setting for studying these design choices.  Launched in January 2021 (initially as \emph{Birdwatch}), Community Notes is a crowdsourced content evaluation system where participating platform users (\emph{contributors}) add short ``notes'' that provide context to help readers assess a post's claims \cite{xcommunitynotes}. Other participating users rate the helpfulness of these notes, and an aggregation algorithm selects which notes are displayed publicly as annotations on the original post, based on their predicted helpfulness across diverse viewpoints.

% The platform launched \textbf{Community Notes} in January 2021 (initially as {Birdwatch}) as a  crowd-sourced content evaluation system
% : contributors write notes that \emph{evaluate} a given post by adding  information intended to help readers assess its claims; other contributors rate the \emph{helpfulness} of those notes; and an aggregation algorithm (e.g., matrix‑factorization) selects notes to display as public annotations on the original post based on their predicted helpfulness across diverse viewpoints.
In September 2022, Community Notes introduced \emph{Rating Impact} and \emph{Writing Impact}, which enforce a consensus-based auditing rule: a user's ability to continue rating and writing content depends on their historical alignment with  the platform's aggregated outcome (the consensus)  \cite{CommunityNotesWritingAbility, CommunityNotesWritingAndRatingImpact, Perez2022Birdwatch}. This policy change offers a natural setting for examining how consensus‑based auditing shapes activity in a crowd-sourced moderation system.

We identify several systematic shifts indicative of reduced minority visibility. First, over time minority contributors increasingly align their ratings with the majority, suggesting strategic anticipation of the platform outcome.
Further, we examine topic‑level participation and find that posts on controversial topics (e.g., politics, international conflict, etc.) receive fewer notes than noncontroversial posts following the adoption of consensus‑based auditing, even though these are exactly the topics where misinformation risk is highest and where an effective policy should encourage more evaluator engagement \cite{west2021misinformation, vosoughi2018spread, van2022misinformation}.

% Third, we demonstrate that the out‑of‑sample predictive performance of the platform’s aggregation algorithm declines after consensus-based auditing is implemented.

To understand these empirical observations, we develop a simple behavioral model in which users choose their ratings to balance their private belief about content quality against a penalty for deviating from the anticipated consensus.  By analyzing the model’s equilibrium, we formally prove that consensus-based penalties amplify conformity, and disproportionately suppress minority contributors.

Finally, we propose an alternative auditing algorithm to address  some of the key shortcomings of the current system. 
The algorithm proceeds in two stages. In the first stage, we estimate content-level and user-level effects, and systematic user–content alignment from the observed ratings, and then compute residuals as the difference between observed and predicted ratings. The resulting residuals isolate the idiosyncratic component of each evaluation.
In the second stage, we estimate contributor reliability from the variance of these first-stage residuals and aggregate current evaluations using inverse-variance weights. Crucially, “consistency” here refers to \emph{stability of residuals conditional on content and user's baseline}, not agreement with the majority.
As a result, a contributor may consistently disagree with the prevailing consensus and still retain influence, provided their evaluations are stable and informative relative to the modeled structure.

This two-stage method is motivated by classical results on weighted least squares under heteroskedasticity.  After removing intrinsic content-level effects and user-specific average effects, residual evaluations can be modeled as conditionally unbiased signals with user-specific variance; in this setting, inverse-variance weighting yields the minimum-variance unbiased aggregation  of signals \cite{Aitken1936least, carroll1982robust, jackson2021finding, resnick2007influence}.
Empirically, we show that our algorithm improves out‑of‑sample predictive performance relative to the deployed algorithm.

% Our method is inspired by the weighted least squares approach in linear regression, which is a classical technique used for obtaining efficient estimators in the presence of heteroskedasticity \cite{Aitken1936least}. Related approaches have been well-documented to prevent strategic manipulation and provide consistent estimates  \cite{carroll1982robust, jackson2021finding, resnick2007influence}. 

Better predictive performance yields harm reduction at scale. Previous work shows that the effectiveness of community annotations is highly time-sensitive. Notes attached earlier yield substantially larger reductions in engagement and diffusion than notes attached later; posts receive about 50\% fewer reposts when notes are attached within 12 hours, compared to less than 10\% reductions when attached after about 48 hours \cite{slaughter2025reduce}. A more predictive auditing rule increases the signal-to-noise of early aggregates, reducing the number of ratings required to reach a confidence helpfulness decision and shortening time-to-attachment.

% \emph{Why predictive performance matters.} Previous work shows that the effectiveness of community annotations is highly time‑sensitive  \cite{slaughter2025reduce, gao2024can, drolsbach2023diffusion, brashier2021timing}: notes attached earlier yield substantially larger reductions in engagement and diffusion than notes attached later (e.g., posts receive about 50\% fewer reposts when notes are attached within 12 hours, compared to less than 10\% reductions when attached after about 48 hours \cite{slaughter2025reduce}). A more predictive auditing rule increases the signal‑to‑noise of early aggregates and reduces the number of ratings required to reach a confident helpfulness decision, thereby shortening time‑to‑attachment and improving harm reduction at scale.

% Predictive performance matters operationally because the impact of community annotations is highly time‑sensitive: community annotations that are attached earlier yield substantially larger reductions in engagement and diffusion than late attached annotations (e.g., posts receive about 50\% fewer reposts when notes are attached within 12 hours, compared to less than 10\% reductions when attached after 48 hours \cite{slaughter2025communitynotes,?}). A more predictive auditing rule increases the signal‑to‑noise of early aggregates and reduces the number of ratings required to reach a confident aggregation decision, thereby shortening time‑to‑attachment and improving harm reduction at scale.

In sum, our work makes three key contributions: (1) We provide empirical evidence that consensus-based auditing systematically alters minority behavior and reduces engagement with controversial topics. (2) We introduce a behavioral model that explains these shifts.   (3) We design and evaluate an alternative auditing and aggregation algorithm that improves predictive performance while preserving participation of minorities.

\subsection*{Related Work}
% A growing set of platforms use crowd-sourced moderation and annotation systems to add context to potentially misleading content. X’s Community Notes is the most prominent deployed example.

The promise of crowd-sourced evaluation is often motivated by the ``wisdom of crowds'': when individual judgments are diverse and independent, their aggregation can outperform individual experts \cite{galton1907vox, surowiecki2005wisdom}. However, when individuals are exposed to others’ opinions, social influence can undermine independence and induce herding \cite{banerjee1992simple, acemoglu2024model, eyster2010naive, lorenz2011social}. Studies show that once early public feedback is observed, later contributors may follow the emerging trend, causing correlated errors and conformity \cite{acemoglu2022fast, dai2018aggregation}. 
Consensus-based
auditing creates an additional channel for such effects:
when contributors are rewarded for matching the eventual
consensus rather than for providing informative signals, they
are encouraged to conform \cite{horvitz2012incentives, miller2005eliciting, kong2016putting, waggoner2014output, prelec2004bayesian}. 
This is particularly problematic in crowd-sourced moderation systems like Community Notes, where viewpoint diversity
can mitigate biases in content labeling  \cite{thebault2023diverse}. Yet empirical
evidence on how consensus-based auditing rules shape contributor behavior in deployed crowd-sourced moderation systems
remains limited. 

To date, 
empirical work on Community Notes has primarily focused on downstream effects of note attachment on user behavior and information diffusion. Studies show that posts with attached notes have lower engagement -- measured in terms of likes, replies, views, and reposts -- with especially strong effects when notes are attached earlier in a post’s lifecycle  \cite{slaughter2025reduce, gao2024can, drolsbach2023diffusion, brashier2021timing,allen2024quantifying}.  Other studies find that notes receiving broad contributor endorsement are perceived as more trustworthy by users than platform-issued misinformation labels \cite{drolsbach2024community}. Partisan asymmetries in content moderation participation have also been documented:  contributors are more likely to
challenge counter-partisan content and to rate co-partisans’ notes as more helpful \cite{allen2022birds, renault2025republicans, kangur2024checks}. Together, this literature establishes that Community Notes can affect engagement, trust, and partisan dynamics, while leaving open how platform incentive and eligibility rules shape contributor participation and coverage across topics.

One difficulty that crowd-sourced moderation systems face is lack of a ground truth for estimated quantities like note helpfulness and user latent factors. There is a large literature that addresses the problem of how to infer contributor reliability in such settings. For example, in the crowd-sourcing and labeling literature, worker--task models estimate worker accuracy and
task difficulty from patterns of agreement and disagreement \cite{dawid1979maximum, raykar2010learning, karger2011iterative, kashima2024trustworthy}. Relatedly, there is a body of literature on designing mechanisms to elicit truthful information without verified labels, including proper scoring rules and peer-prediction methods \cite{Gneiting2007strictly, miller2005eliciting, prelec2004bayesian, witkowski2012peer, shnayder2016informed}. A common lesson from these works is that participants should be evaluated against targets not determined by themselves, reducing incentives to coordinate on a focal consensus \cite{miller2005eliciting, prelec2004bayesian, Gneiting2007strictly, liu2017machine, faltings2022game}. These approaches formalize the idea that reliability should be learned from error structure, not simply from matching a majority vote. 

The Community Notes algorithm is a bridging-based algorithm that aims to surface content rated highly by users who generally hold opposing viewpoints. Methodologically, its implementation is closely related to latent-factor models and collaborative filtering in recommender systems, where one seeks to infer user and item attributes from sparse, noisy ratings \cite{van2010manipulation,amatriain2009ilikeit, karger2011iterative}. Other approaches to bridging have been proposed in \cite{blair2026structure}. A parallel line of work studies the design of aggregation rules that limit influence from noisy or adversarial users \cite{resnick2007influence}. In statistics, classical results on weighted least squares show that inverse-variance weighting yields efficiency gains under heteroskedastic noise \cite{Aitken1936least, carroll1982robust}. 

Our two-stage algorithm combines these principles and ideas from the literature by drawing on the logic of scoring contributors against targets not defined purely by raw agreement with consensus, and using these scores to construct weights in a second-stage aggregation model. This combination links practical auditing in crowd-sourced moderation to established ideas in statistical efficiency and incentive-compatible information elicitation.

\section*{X Community Notes}
In the Community Notes program, users can write short annotations (called notes) that provide context for potentially misleading or disputed content on the platform, and rate notes written by other users as \textit{Helpful}, \textit{Somewhat Helpful}, or \textit{Not Helpful}. These ratings are aggregated using a matrix factorization algorithm that determines which notes are displayed \emph{publicly}
under the corresponding posts, while all remaining notes are kept hidden \cite{CommunityNotesRankingNotes}. As a result, aggregation outcomes directly shape which information is surfaced and which contributors shape public discourse.

\subsection*{Aggregation via Matrix Factorization}
\label{sub_sec:mf}
For each user–note pair $(u,n)$, let $r_{un}\in\{0, 0.5, 1\}$ denote the observed rating, where \textit{Helpful}, \textit{Somewhat Helpful}, and \textit{Not Helpful} responses are mapped to $1$, $0.5$ and $0$, respectively \cite{CommunityNotesRankingNotes}.
The platform assumes that ratings are modeled as
\begin{equation}\label{eq:r_un}
    {r}_{un} = \mu + h_u + i_n + f_u \cdot g_n,
\end{equation}
where $\mu$ is a global intercept, $h_u$ is a {rater intercept} capturing user $u$'s baseline agreeability (the tendency to mark notes as helpful rather than unhelpful, regardless of content), $i_n$ is a {note intercept} capturing the perceived overall helpfulness of note $n$, and $f_u, g_n \in \mathbb{R}$ are latent {rater} and {note factors} whose product represents the ideological alignment between user and note.

The platform estimates these parameters by solving the regularized least-squares problem
\begin{align}\label{eq:matrix_factorization_eq}
    \hat h_u, \hat i_n, \hat g_n,\hat f_u,\hat \mu
    &= \arg\min_{\mu, h_u,i_n,f_u,g_n} 
       \sum_{(u,n) \text{ observed}} (r_{un} - \hat r_{un})^2 \nonumber \\
    &\quad + \lambda_u \sum_u \left(\|h_u\|^2 + \|f_u\|^2\right) \nonumber \\
    &\quad + \lambda_n \sum_n \left(\|i_n\|^2 + \|g_n\|^2\right).%\footnote{This loss function was at the time of writing this paper, it has since changed.}
\end{align}
where $\lambda_u,\lambda_n$ are regularization parameters.

In this formulation, the note intercept $i_n$ is the primary quantity of interest, as it
captures how broadly helpful a note is across the user base and directly determines its
eligibility for public display.
Notes with estimated intercept $\hat i_n\ge 0.4$ are
classified as \emph{Helpful} and shown publicly beneath the corresponding posts, while
notes with $\hat i_n<0.4$ are withheld from public display. Among these, notes with
$\hat i_n<-0.05$ are classified as \emph{Not Helpful} and may generate negative feedback for both the note’s author and raters who marked the note as helpful \cite{CommunityNotesRankingAlgorithm}.\footnote{The
production scoring system has evolved over time; the formulation here reflects the core
structure and thresholds relevant for our analysis. Our empirical replication uses the
corresponding open-source implementation released by platform X, though operational
details such as confidence adjustments or per-topic factorizations may differ.}

\subsection*{Rating Impact}
In September 2022, Community Notes introduced the {Rating Impact} feature, which
links rating aggregation to contributors’ continued participation.
Rating Impact is a user-level score computed based on whether a contributor’s ratings
align with a note’s eventual \emph{Helpful}/\emph{Not Helpful} classification; agreement
increases the score, while disagreement decreases it
\cite{Perez2022BirdwatchVisible, Perez2022Birdwatch}. A user's Rating Impact score determines their ability to write Community Notes; new contributors must reach a minimum Rating Impact threshold before they can write notes. Writers have a separate Writing Impact score, which governs their ability to write notes on the platform; writers lose the ability to submit new notes if at least $3$ of their $5$ most recently written notes have been labeled \emph{Not Helpful} after aggregation \cite{CommunityNotesWritingAbility}. 
% \begin{enumerate}
%     \item \textbf{Writing eligibility.} New contributors must reach a minimum Rating Impact threshold before they can write new notes \cite{CommunityNotesWritingAbility}.
%     \item \textbf{Ongoing writing access.} Writers lose the ability to submit new notes if at least $3$ of their $5$ most recent written notes have been labeled \emph{Not Helpful} after aggregation. 
% \end{enumerate}

While the stated goal of this system is to surface notes from contributors with a track record of accuracy, it may create conformity incentives. To gain or retain writing privileges, contributors may
feel pressure to align both their ratings and their note content with anticipated
majority views, as repeated disagreement risks reduced influence and loss of access.
Therefore, the September 2022 policy introduction, together with our use of October 1, 2022 as a conservative operational cutoff, provides a natural policy discontinuity for our study. 

% A single user may act as both a writer and a rater, but not for the same note.

% \emph{Roadmap.} We first describe the public dataset used in our analyses, then summarize the matrix-factorization model that drives display decisions, and finally define the Rating Impact policy and our replication of latent variables.

\section*{Data}
Our analysis uses the publicly released Community Notes dataset, which contains the complete history of notes and ratings since Community Notes' launch in 2021 \cite{CommunityNotesDownloadData}.\footnote{The initial name for Community Notes was Birdwatch.} Our primary analysis window spans June~1,~2022 through May~31,~2023, covering several
months on both sides of the rollout of Rating Impact. For each rating event, the dataset records the note identifier $n$, rater identifier $u$,
timestamp, selected rating, a summary of the note content, and metadata about the
associated post. The dataset does not include estimates of the latent parameters $(h_u,i_n,f_u,g_n)$, as these quantities are re-estimated by the platform as new data
arrive.

To study how the Rating Impact policy affects contributor behavior, we reconstruct the
latent factors by running the platform’s matrix factorization algorithm on the open-source
ratings data.
Specifically, we recover weekly estimates of $(h_u,i_n,f_u,g_n)$ by applying the
Community Notes aggregation algorithm to cumulative ratings data up to each week.
We implement the publicly released December~31,~2022 version of X’s open-source matrix factorization code \cite{CommunityNotesRankingNotes, commnotes2022code}.\footnote{The code is publicly available on the platform’s GitHub repository.} As a robustness check, we additionally incorporate an improved platform's algorithm from the May~2025 release (see Supplementary Information, Section~2 for details). 
% We implement two publicly released versions of X’s open-source matrix factorization code: the December~31,~2022 release and the May~2025 release \cite{CommunityNotesRankingNotes, commnotes2022code}.\footnote{Both versions are publicly available on the platform’s GitHub repository.}
While both implementations share the same underlying matrix factorization framework described earlier, they differ modestly in operational details such as regularization choices, handling of ties, and per-topic factorizations. We use both versions as a sensitivity check to verify that our results are robust to reasonable variation in the estimation procedure.

 % Section~\ref{sub_sec:mf}
% \begin{itemize}
% \item the note identifier $n$ and the rater identifier $u$,
%     \item the timestamp of the rating event,
% \item
% the rater’s selected category (Helpful, Somewhat Helpful, Not Helpful),
% \item 
% a summary of the note’s content, and
% \item
% metadata about the associated post.
% \end{itemize}

% \subsection*{Deriving Latent Variables for Our Analysis} [should data come before this or after?]
% To study the effects of the Rating Impact policy, we run the platform’s algorithm on their open source rating data to recover weekly estimates of the the latent parameters $(f_n, f_u, i_n, i_u)$ from the rating matrix. We run two versions of the open-source X implementation of matrix factorization: the December 31, 2022 release and the May 2025 release.\footnote{Both versions are publicly available on the platform’s GitHub repository.} The backbone of both versions is the matrix factorization described in the Matrix Factorization section, but implementation details, such as handling of ties, regularization choices, and per-topic splits, differ slightly. We use both versions as a sensitivity check to ensure that our findings are robust to variations in the estimation procedure.

\section*{Empirical Results}

We document three empirical patterns in platform activity around the September 2022 introduction of the Rating Impact system,  using October~1,~2022 as the analysis cutoff date\footnote{October~1 provides a conservative operational
cutoff following the period when the relevant eligibility rules began to take effect in
our dataset.}; we refer to dates before October~1 as {pre-rollout} and dates after as {post-rollout}.

First, we observe that minority-aligned raters shift their evaluations toward the majority. Second, controversial content receives relatively lower engagement post-rollout, potentially contributing to fewer visible annotations. Finally, the platform’s latent-factor model exhibits reduced out-of-sample predictive performance post-rollout. 
% \begin{enumerate}
%     \item minority-aligned raters shift their evaluations toward the majority
%     \item annotation outcomes diverge by controversy, with relatively lower engagement on controversial content, 
%     \item 
%     the platform’s latent-factor model
% exhibits reduced out-of-sample predictive performance post-rollout.
% \end{enumerate}

These observations align with the hypothesis that raters behaved strategically in response to the Rating Impact policy. The following subsections present each pattern in turn.

\subsection*{Minority Behavior Shift}

First, we examine whether the introduction of the Rating Impact system altered how minority raters evaluate notes. Recall that each user and note are assigned a latent factor, $f_u$ and $g_n$, respectively, that represent their ideology\footnote{These factors can theoretically be vectors but in the context of X's algorithm factors are scalar.}. Then, for a pair of users and notes $(u, n)$, $f_u\cdot g_n$ represents their the rater-note alignment, and its magnitude measures the strength of that alignment. 
% As a proxy for behavior shift, we study distribution shifts in rater and note factor and changes in user-note alignment over time in our primary analysis window around the Rating Impact roll-out. We find that for both rater and note factors, there is \textbf{distribution shift toward the majority} post-Rating Impact rollout relative to pre-rollout. We also find that post-rollout, \textbf{alignment is less explanatory of note helpfulness} compared to pre-rollout. 
As proxies for behavioral change, (i) we study shifts in the distributions of rater and note factors, and (ii) changes in the predictive role of user--note alignment for \emph{Helpful}
ratings. 
% We a post-rollout increase in mass concentration near the majority mode for both note and rater factor distribution, and that rater-post alignment becomes substantially less predictive of whether a rating is \emph{Helpful}.

% We find that for both rater and note factors, there is \textbf{distribution shift toward the majority} post-Rating Impact rollout relative to pre-rollout. We also find that post-rollout, \textbf{alignment is less explanatory of note helpfulness} compared to pre-rollout. 

% We document these findings below. 

\subsubsection*{(i) Latent Factor Distribution Shift}

Figure \ref{fig:rater_factor_shift} plots the distribution of rater factors in the pre-rollout and post-rollout periods. Relative to the pre-rollout period, the share of raters in the majority-aligned mode increases from 58.4\% to 63.1\%, while the share in the minority-aligned mode decreases from 41.6\% to 36.2\%. Within the minority group, the distribution also shifts toward zero: the mean absolute factor $|f_u|$ for minority-aligned raters declines from 0.522 (CI: ± 0.177) pre-rollout to 0.413 (CI: ± 0.215) post-rollout. This change reflects both a subset of minority-aligned raters moving closer to the center and a subset crossing over to the opposite alignment.

A parallel pattern is visible in the note factor distribution. For note factors the mean absolute note factor $|g_n|$ changes from 0.451 (CI: ±  0.25) pre-rollout to 0.408 (CI: ± 0.233) post-rollout. This suggests that the notes produced by minority-aligned writers are, on average, positioned closer to the center of the latent-factor spectrum in the later period\footnote{Note factors partly reflect writer behavior but also depend on endogenous note entry and which notes receive sufficient evaluation to be estimated.}.

\begin{figure}[h]
    \centering

    \begin{minipage}{.5\linewidth}
        \centering
        \includegraphics[width=\linewidth]{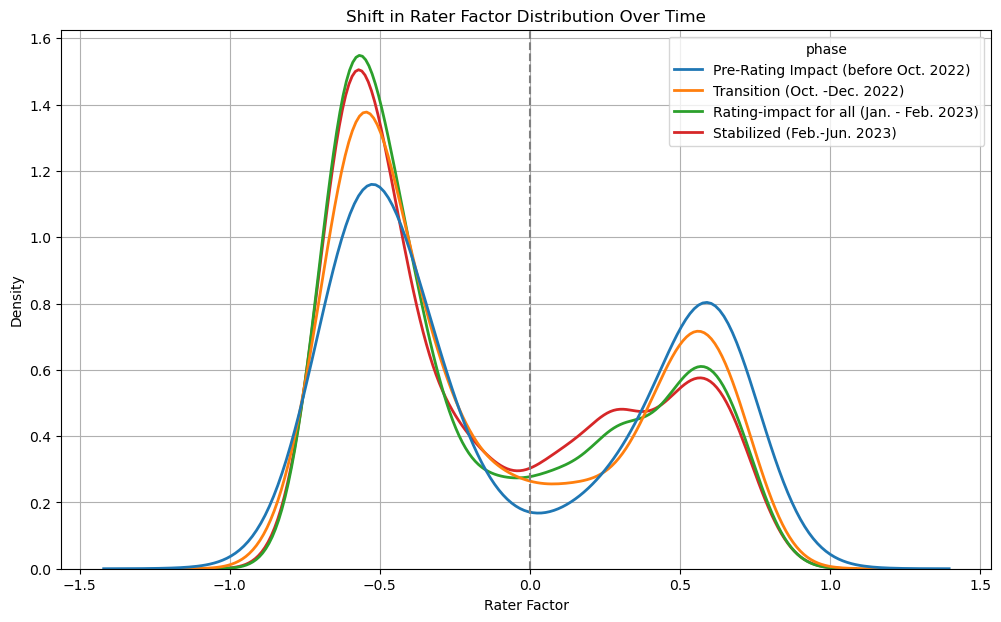}
        \subcaption{Rater factor change}
        \label{fig:rater_factor_shift}
    \end{minipage}

    \vspace{0.75em}

    \begin{minipage}{.5\linewidth}
        \centering
        \includegraphics[width=\linewidth]{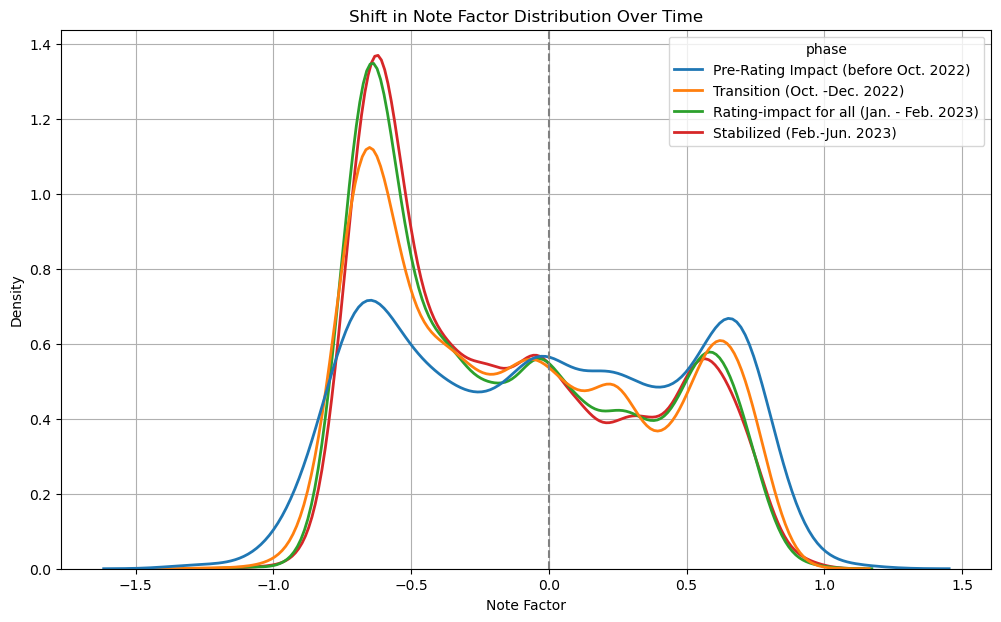}
        \subcaption{Note factor change}
        \label{fig:note_factor_shift}
    \end{minipage}

    \caption{Visualization of rater and note factor distribution shift over time.}
    \label{fig:factor_changes}
\end{figure}

Note that user and note factors are computed relative to all user-note interactions on the platform. 
To distinguish behavioral adaptation from compositional change due to entry of new raters,
we compare factor shifts for two groups.
We call one group ``Early Users"; this is the set of users who have been active raters on the platform since before Oct. 1, 2022. The second group of users are ``New Users"; this is the set of users who joined as new raters on X Community Notes between Oct. 1, 2022 and Jan. 1, 2023. We compute the factor shift for Early Users and New Users, taking the difference between their first latent factor after Jan. 1, 2023 and first latent factor after Oct. 1, 2022, for users where both factors exist. We then run a permutation test on the mean factor shift (10,000 permutations). We find that early users' shifts were more negative than new users ($\Delta = -0.053, p < 0.001$). This is another indicator that the Rating Impact policy influences user behavior, causing users who experienced the policy change to become more strategic in rating. 

% \begin{figure}[!htbp]
%     \centering
%     \includegraphics[width=0.9\linewidth]{commnotes_figures/kde_ci_panels.png}
%     \caption{This plot shows the distribution shift for early users vs. new users during three time periods: Oct. 2022 to Jan. 2023, Jan. 2023 to Apr. 2023, and Apr 2023 to Jun. 2023. We see that in the early user group, the average rater factor shifts from $-0.165$ during the first time period to $-0.251$ during the last time period. In the new users group, the average rater factor shifts from $-0.121$ to $-0.131$, remaining relatively stable over the same periods of time.}
%     \label{fig:note-factor-shift-kde-plot}
% \end{figure}

% \begin{figure}[!htbp]
%     \centering
%     \includegraphics[width=0.9\linewidth]{commnotes_figures/violin_plot.png}
%     \caption{Violin plot showing the difference in distribution between rater factors before and after the rollout of Rating Impact for early users, who joined before the rollout of Rating Impact, versus new users joining after the rollout of Rating Impact. Observe that while the distribution of new users remains relatively stable, the distribution of early users shifts toward the majority, suggesting that early users impacted by the Rating Impact rollout become more strategic in their interactions.}
%     \label{fig:note-factor-shift-violin-plot}
% \end{figure}

\subsubsection*{(ii) User-Note Alignment} To assess whether the predictive role of user–note alignment changed at the Rating Impact rollout, we use Spearman's correlation coefficient as a proxy. 

We use ratings data from Aug. 1, 2022 to Jan. 1, 2023, focusing on a fixed cohort of $1,202$ users who were active on the platform both before and after the October 1, 2022 cutoff. For each period, we compute the Spearman correlation between the rater-note factor dot product and helpfulness ratings, taking the difference (post minus pre) as our test statistic. The correlation declined from $0.792$ before the rollout to $0.525$, yielding an observed difference of $-0.267$. To assess whether this decline is statistically significant, we conduct a permutation test with $1,000$ iterations, randomly reassigning ratings to the pre/post groups while preserving the original group sizes, {which gave a $p$-value of $0.004$}.

% \begin{figure}[H]
%     \centering
%     \includegraphics[width=\linewidth]{commnotes_figures/permutation_spearman.png}
%     \caption{Permutation test test statistic distribution for the difference in Spearman correlation (post minus pre intervention) between rater-note factor dot product and helpfulness ratings among early users. The red dashed line marks the observed difference of $-0.267$. Very few of the $1,000$ permuted differences fall near the observed value, yielding $p = 0.004$, indicating that the decline in predictiveness after October 1, 2022 is highly unlikely to have occurred by chance.}
%     \label{fig:spearman-histogram}
% \end{figure}

\begin{figure}[h]
    \centering
    \includegraphics[width=.7\linewidth]{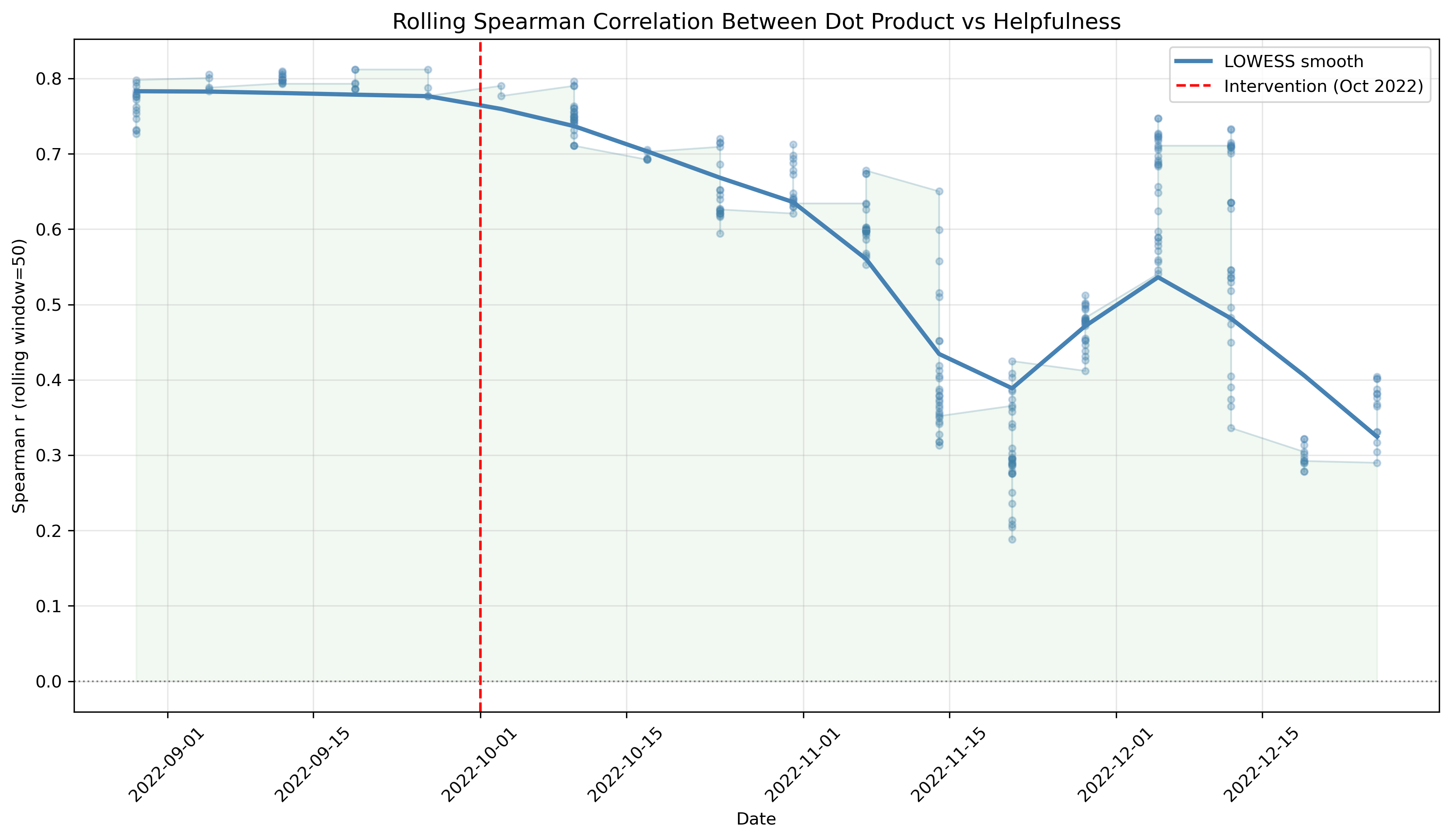}
    \caption{Rolling Spearman correlation between rater-note factor dot product alignment and helpfulness ratings for the cohort of $1,202$ early users, computed over a sliding window of 50 ratings sorted by date. The red dashed line marks the algorithmic change on October 1, 2022. The bold line shows a LOWESS smooth of the rolling correlation. Prior to the intervention, the correlation is relatively stable, whereas after the intervention, the correlation declines steadily. This is an observational indicator that the rollout of Rating Impact weakened the relationship between rater-note alignment and helpfulness ratings among the group of users who were active before the change.}
    \label{fig:spearman-visualization}
\end{figure}

In the Appendix, we conduct sensitivity test these behavioral shifts with various additional metrics. Taken together, these patterns are consistent with the hypothesis that after the Rating Impact rollout, minority–aligned raters adjusted their evaluations toward the anticipated consensus, reducing observable disagreement.

\subsection*{Annotations on Controversial Topics}
Another indication of changing rater behavior is the extent to which they engage with {controversial} vs. {non-controversial} notes; agreement with the majority is far more likely on non-controversial notes, thereby increasing the potential to boost a rater’s Rating Impact score. To study this, we compare annotation patterns for controversial and non-controversial notes before and after rollout of the Rating Impact system. We define controversy at the note level using two complementary approaches: a topic-based classification and a factor-based classification.

\textbf{Topic-based classification.} 
For topic-based classification, we assign each note to a content topic based on its summary text, using retrained version of X’s public topic-assignment pipeline (details in Appendix \ref{sec:empiric}). We then flag as controversial those topics that, in platform discourse and prior literature, are known to be highly polarized (e.g., national politics, public health, geopolitical conflicts). This classification captures domain-level controversy, independent of individual note characteristics. We additionally use Large Language Model (LLM)-based topic assignments as alternative classification procedures (details in Appendix \ref{sec:empiric}).

\textbf{Factor-based classification.}
For factor-based classification, we use the absolute value of the estimated note factor, $|g_n|$, as a continuous measure of ideological alignment. In the Community Notes latent-factor model, notes with factors far from zero are those that receive systematically different evaluations from different groups of raters. We therefore classify a note as controversial if $|g_n|$ lies in the top quantile of the distribution. This approach captures instance-level controversy even within otherwise non-polarized topics.

\vspace{1cm}

We use both definitions in parallel to ensure robustness. The topic-based approach provides interpretability, while the factor-based approach is model-derived and sensitive to within-topic variation. Our empirical results are qualitatively consistent across both definitions. In this section, we report results using the topic-based classification leaving factor-based results to Appendix \ref{sec:empiric}.

Figure \ref{fig:controversial_notes_plot}
presents the weekly share of tweets receiving a first note that ultimately attains Helpful status before and after the rollout date, separately for controversial and non-controversial topics. For controversial topics, the helpful-share increases\footnote{The definition of "helpfulness" was relaxed around the time of the Rating Impact rollout, so the helpfulness for both controversial and non-controversial notes increases \cite{CommunityNotesRankingNotes}.} from 0.061 with Wilson CI [0.044, 0.084] pre-rollout to 0.126 with Wilson CI [0.109, 0.146] post-rollout, a change of 6.5 percentage points. In contrast, non-controversial topics see an increase from 0.062 with Wilson CI [0.027, 0.138] to 0.199 with Wilson CI [0.151, 0.258], a change of 13.7 percentage points. We use a difference-in-differences (DiD) design to estimate the effect of the rollout of the Rating Impact system on note helpfulness for controversial vs. non-controversial notes. Using a symmetric 12 week band around Oct. 1, 2022, we see that the probability a non-controversial note is rated helpful increased by around 9 percentage points (95\% CI [0.015, 0.161])  more than the probability a controversial note is rated Helpful. These shifts are also reflected in the counts during this 12 week band before and after the cutoff date: the proportion of tweets receiving at least one new note in controversial topics decreases ($-3.4 pp$), while the proportion of tweets receiving at least one new note in non-controversial topics rises ($+3.4 pp$).

This trend of reduced engagement on controversial topics is especially present in the minority group. In Appendix \ref{sec:empiric}, we run an additional DiD study on the proportion of ratings assigned to controversial notes for individuals in the minority group vs. non-minority group. We find that, in a 12 week band around the rollout date, the proportion of controversial notes rated by minority users decreases 7.53 percentage points relative to the change observed for non-minority users ($p = 0.0284$).

We emphasize that these comparisons are descriptive; they show a post-rollout divergence in annotation outcomes between controversial and non-controversial topics, but they do not by themselves identify the mechanism. However, the consistent pattern across both classification schemes is in line with the hypothesis that, under the Rating Impact system, content on controversial topics receives fewer helpful ratings and correspondingly fewer surfaced annotations, while non-controversial topics receive more.
% Recall that on Community Notes, notes are rated as "Helpful" or "Unhelpful" based on when their estimated note intercept $\hat i_n$ reaches a threshold, larger than $0.4$ or smaller $0.05 - 0.8\cdot |f_n|$, respectively. 

\begin{figure}[h]
    \centering
    \includegraphics[width=0.5\linewidth]{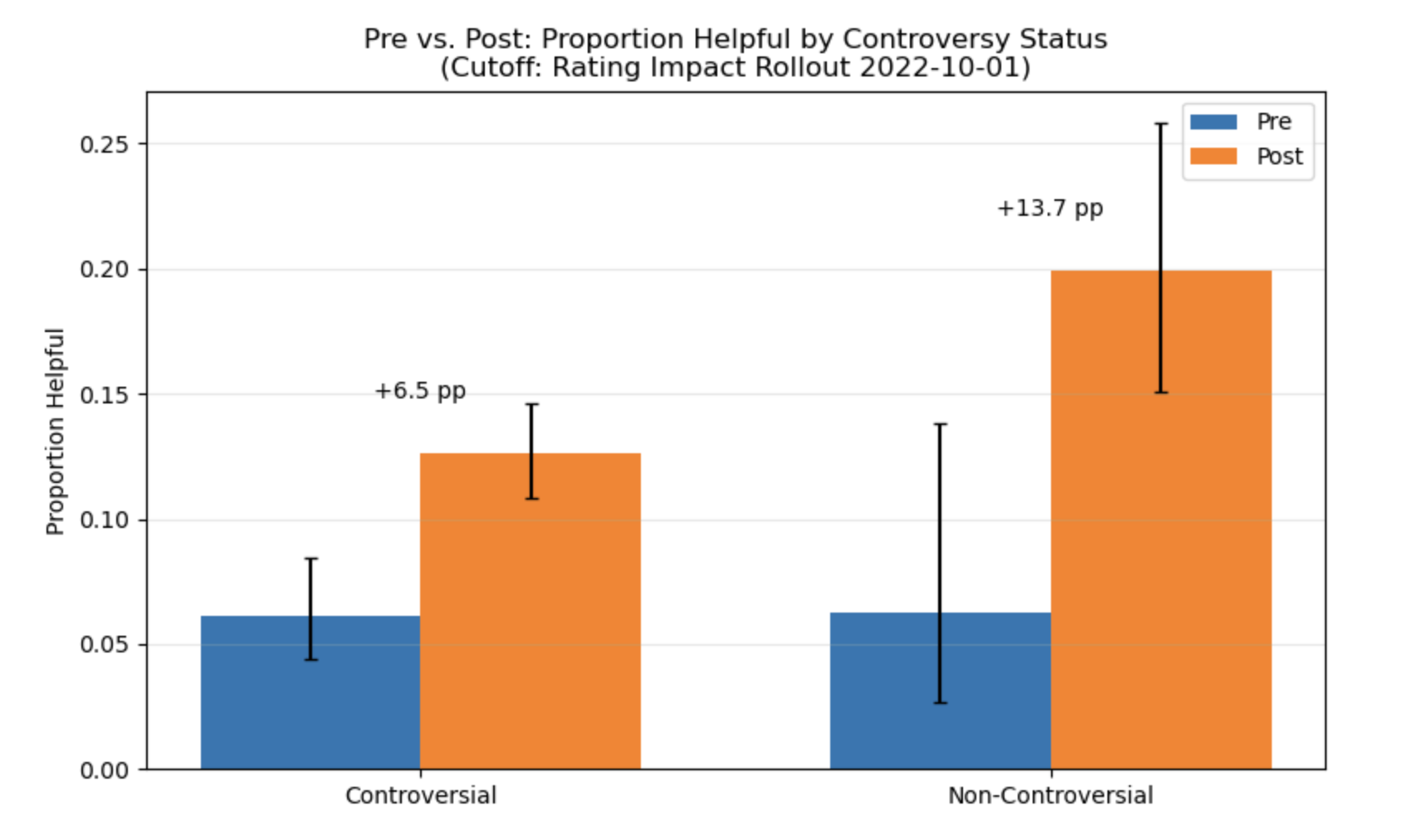}
    \caption{Pre–post change in the share of notes with final status Helpful by controversy category around the Rating Impact rollout (cutoff: 2022-10-01). Bars show the mean proportion in the approximately 20 weeks before (Pre, blue) and after (Post, orange) the cutoff; error bars are 95\% CIs. Text above bars reports the Post–Pre difference in percentage points. The increase is larger for non-controversial notes (+13.7 pp) than for controversial notes (+6.5 pp).}
    \label{fig:controversial_notes_plot}
\end{figure}

\begin{table}[h]
\centering
\begin{tabular}{lcccc}
\toprule
Weeks (Pre, Post) & Estimate & 95\% CI & $p$-value \\
\midrule
(12, 0-12)  & 0.0882 & [0.0152, 0.1612] & 0.0179 \\
(12, 13-26) & 0.1185 & [0.0311, 0.2059] & 0.0079 \\
\bottomrule
\end{tabular}
\captionsetup{justification=centering}
\caption{DiD Estimates. \\ Windowed difference-in-differences estimates of the change in the weekly helpfulness rate difference between non-controversial and controversial content after the October~1,~2022 rollout. The outcome is the difference in weekly proportion of tweets receiving a note that is ultimately rated \textit{Helpful} between tweets with controversial vs. non-controversial notes. The gap increases by 8.8 pp in the first 12 weeks post-rollout and by 11.9 pp in weeks 13--26. Both effects are positive and statistically significant.}
\label{tab:did_results}
\end{table}

\subsection*{Predictive Accuracy}

We study how the predictive performance of the platform’s latent-factor model changes around the rollout of Rating Impact. For each calendar week $t$, we fit the X Community Notes matrix factorization algorithm on cumulative data up to week $t$, and then evaluate two types of prediction error. {The {in-sample} error measures the error for the model predictions on ratings in week $t$. The {one-week-ahead} error measures the error for model predictions on ratings in week $t+1$, restricting to rating  pairs $(u,n)$ where both the rater and note are observed in week $t$ (see Appendix \ref{sec:twostage imp} for implementation details).}

Figure~\ref{fig:predictve_accuracy_over_time} plots the one-week-ahead mean squared error (MSE) over time, with the rollout week marked and Table \ref{tab:pred-accuracy-results} aggregates by period. The post-rollout period shows a higher and more volatile error than the pre-rollout period, with in-sample MSE increasing over $116\%$ and one-week-ahead MSE increasing over $76\%$ between pre- and post- rollout. 

% The reduction in predictive accuracy is consistent with the behavioral changes documented earlier (Section~\ref{4.1}). Both reduce the variation in observed ratings and shift the underlying relationship between raters and notes, making a model fit on earlier data less transferable to later data.

\begin{figure}[!htbp]
    \centering
    \includegraphics[width=0.6\linewidth]{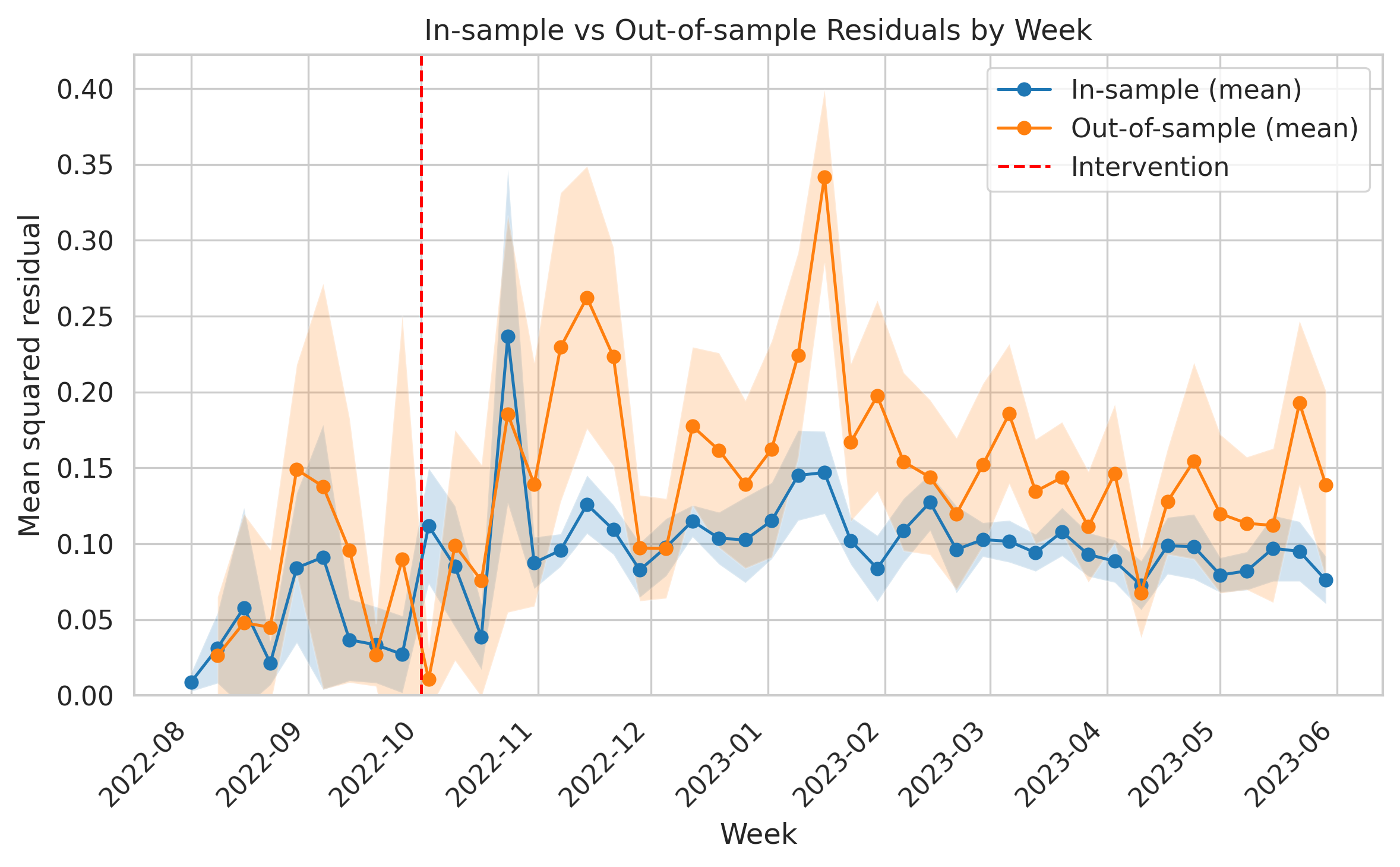}
    \caption{This figure shows the weekly mean squared error (MSE) for in-sample vs. out-of-sample predictions from the matrix factorization model. {The MSE is computed as the squared difference between the observed rating outcomes and the model’s predicted rating outcomes (pre-discretization)}. In-sample errors reflect fit to the same week’s ratings, while out-of-sample errors use factors estimated from week $t$ to predict ratings in week $t+1$. The vertical dashed line marks the Rating Impact analysis date we use.}
    \label{fig:predictve_accuracy_over_time}
\end{figure}

\begin{table}[!htbp]
\centering
\begin{tabular}{cccc}
\hline
 & Pre-Rollout & Post-Rollout & \% Change \\
\hline
In-sample &
\begin{tabular}{@{}c@{}} 0.0488 \\ {\scriptsize (0.0327, 0.0648)} \end{tabular} &
\begin{tabular}{@{}c@{}} 0.1057 \\ {\scriptsize (0.1005, 0.1109)} \end{tabular} &
+116.60\% \\
One-week-ahead &
\begin{tabular}{@{}c@{}} 0.0927 \\ {\scriptsize (0.0600, 0.1254)} \end{tabular} &
\begin{tabular}{@{}c@{}} 0.1634 \\ {\scriptsize (0.1429, 0.1839)} \end{tabular} &
+76.27\% \\
\hline
\end{tabular}
\captionsetup{justification=centering}
\caption{Prediction Accuracy of Matrix Factorization. \\
This table shows the average increase in in-sample MSE and one-week-ahead MSE of the matrix factorization estimates, 
computed as averages over three months pre- and post-rollout. 
Numbers in parentheses indicate 95\% confidence intervals.}
\label{tab:pred-accuracy-results}
\end{table}

\section*{Auditing by Predictive Stability: Two-Stage Weighted Matrix Factorization}
The empirical patterns suggest that the platform’s current auditing may create incentives for conformity that reduce diversity in ratings and coverage of controversial topics.

We therefore propose an alternative rule that separates auditing from agreement.
Following well-established precedent in the literature, we propose an alternative weighting method, which we refer to as weighted matrix factorization, that targets \emph{rater reliability} directly rather than agreement with the final note-status. 
The method consists of the following steps:\\
    \textbf{First stage:} Compute matrix factorization estimates $\hat \mu, \hat h_u, \hat i_n, \hat f_u, \hat g_n$ as in \eqref{eq:matrix_factorization_eq}\\
    
  \noindent  \textbf{Second stage:} 
    \begin{enumerate}
        \item \textbf{Compute residuals:} For each observed user-note pair $(u, n)$, compute the first stage residual:
        \[
            e^{(1)}_{un}
            \;=\;
            r_{un} \;-\;\hat{\mu}\;-\;\hat{h}_u\;-\;\hat{i}_n \;-\;\hat f_u \cdot \hat{g}_n.
        \]
        \item \textbf{Estimate variance:} For each user $u$, estimate the empirical variance of their ratings as 
        \[
            \hat{\sigma}^2_{u}
            \;=\;
            \frac{1}{|N(u)|}\sum_{n\in N(u)} \bigl(e^{(1)}_{un}\bigr)^2
        \]
        where $N(u)$ is the set of notes that user $u$ has rated.
        \item \textbf{Refit intercepts \& factors:} Run a final weighted regression:
        \[
            \arg \min_{\tilde \mu, \tilde h_u, \tilde i_n, \tilde f_u, \tilde g_n}
            \sum_{\substack{(u,n)\\ \text{observed}}}
            \frac{1}{\hat{\sigma}^2_u}
            \left(
            r_{un}-\tilde{\mu}-\tilde{h}_u -\tilde{i}_n-\tilde f_u \cdot \tilde g_n
            \right)^2.
        \]
        This step adjusts the intercepts after incorporating the user variance terms.  In principle, one could also recalculate the $\hat{\sigma}^2_{u}$ and iterate.
    \end{enumerate}

The inverse-variance weight $1 / \hat{\sigma}^2_u$ measures how predictable a rater’s behavior is, given their latent position. Because weights are based on internal consistency rather than agreement with other raters, consistent minority raters keep their influence even when they diverge from the majority. We expect this approach to:

\begin{itemize}
    \item \textbf{Mitigate conformity incentives:} Weights depend on consistency, not consensus alignment.
    \item \textbf{Preserve minority contributions:} Consistent raters from minority viewpoints are not penalized for disagreement.
    \item \textbf{Improve predictive performance:} As in WLS, accounting for heteroskedasticity should reduce mean squared error in prediction.
\end{itemize}

% We present the change in predictive performance in the following subsection.  
% For the other two outcomes—minority rater alignment and coverage of controversial topics—we do not observe the counterfactual in which users do not abstain from rating.  
% To analyze these outcomes, we turn to a user behavior model, which allows us to study how conformity incentives could generate the observed patterns in the absence of direct counterfactual data.

% As noted in the Introduction, our approach follows well-established precedent on using reweighting to achieve consistent estimates while preventing strategic manipulation. 

% In econometrics, \emph{Weighted Least Squares} (WLS) achieves efficiency under heteroskedasticity by weighting each observation inversely to its error variance \cite{Greene2010}. In recommender systems, ``confidence-weighted'' and heteroskedastic matrix factorization models explicitly account for differences in user-level rating noise \citep{Wang2018ConfidenceAware}. In crowdsourcing and truth-inference models, annotator reliability is commonly estimated from residual variation and used to reweight labels \citep{}. Our proposal adapts these ideas to the Community Notes context, replacing consensus alignment with a measure of individual rating consistency.

% We propose an alternate method, based on re-weighting proportional to user variance, to evaluate the intrinsic helpfulness of notes. This method is inspired by the two-stage method proposed \kh{Jackson et. al. \& other sources} for linear regression. 

\subsection*{Empirical Predictive Performance}
To test the performance of our weighted matrix factorization algorithm, we run it on the Community Notes dataset modifying the Community Notes matrix factorization algorithm. We use ratings data from Jan. 1, 2023 to June 1, 2024. {In particular, the ratings data from Jan. 1, 2023-July 1, 2023 serves as a warm start for the matrix factorization algorithm, and we evaluate our method on the ratings data from July 1, 2023-June 1, 2024.} Detailed implementation description is given in the Appendix.

% We run the open source Community Notes matrix factorization algorithm to get out-of-sample, one-week-ahead predictions. Then, we run the two-stage matrix factorization approach, using the residuals from the Community Notes matrix factorization algorithm, to get out-of-sample two-stage estimates.  

In Figure \ref{fig:mf-ts-weekly-mean} and \ref{fig:mf-ts-weekly-median} we compare the mean absolute residual and the median absolute residual on the one-week-ahead predictions from the matrix factorization approach vs. the two-stage approach. Using our proposed two-stage approach, the mean absolute residual is 5.73\% lower on average, and the median absolute residual is 27.99\% lower on average.

% (standard deviation 4.75\%)
% (standard deviation 6.78\%)

\begin{figure}[!htbp]
  \centering

  \begin{subfigure}{0.5\linewidth}
    \centering
    \includegraphics[width=\linewidth]{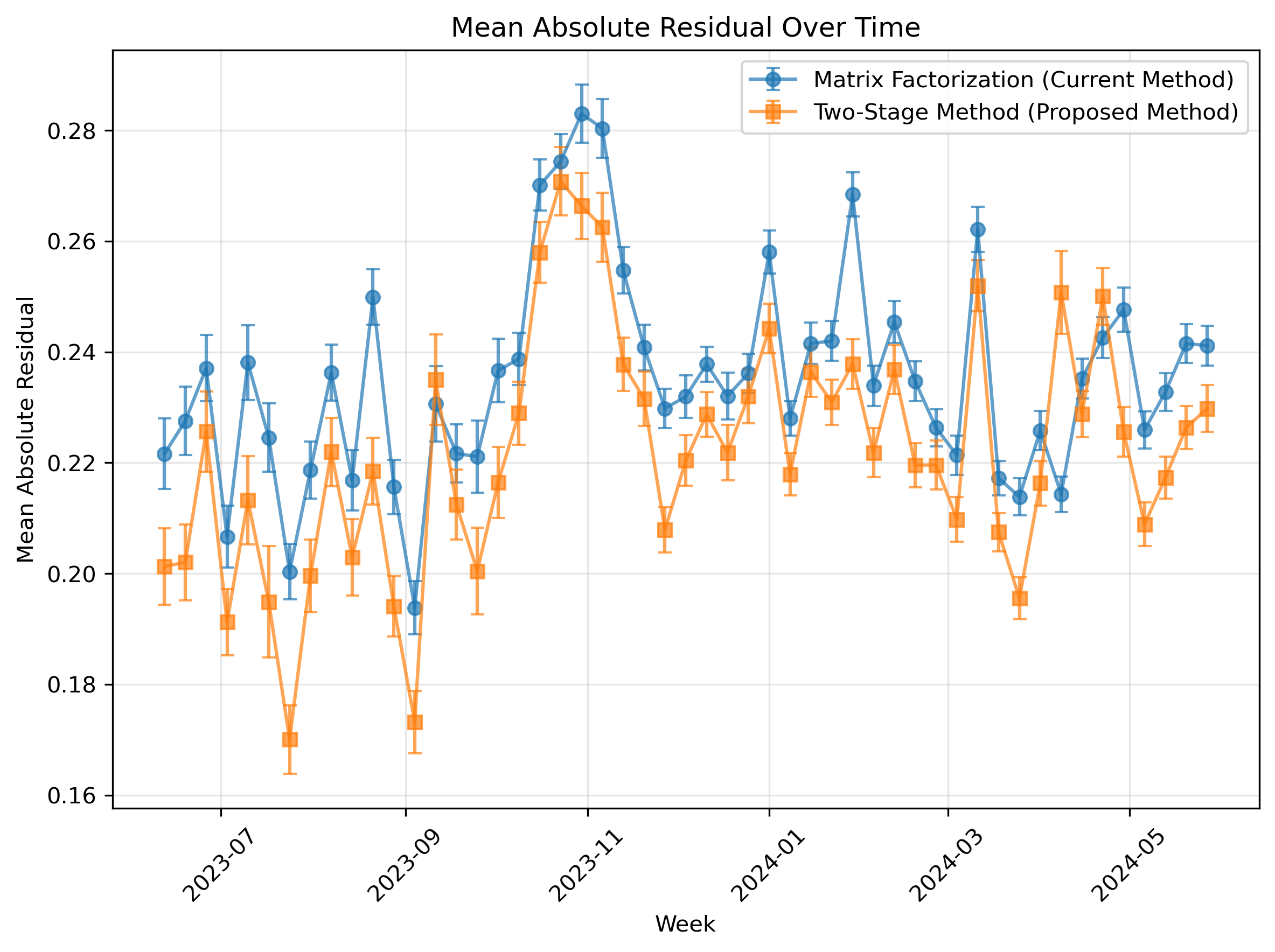}
    \caption{Mean absolute residuals.}
    \label{fig:mf-ts-weekly-mean}
  \end{subfigure}

  \vspace{0.5em}

  \begin{subfigure}{0.5\linewidth}
    \centering
    \includegraphics[width=\linewidth]{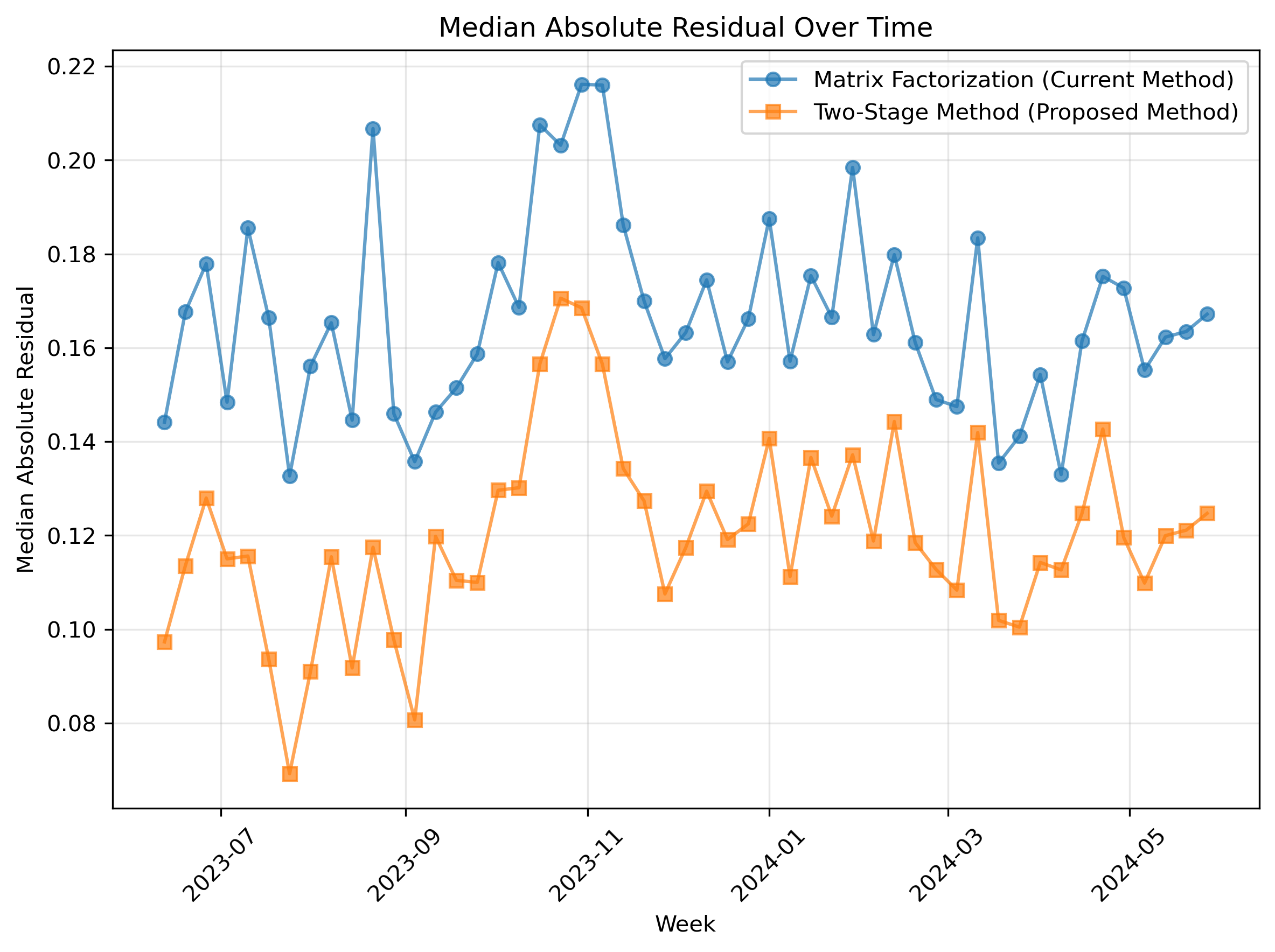}
    \caption{Median absolute residuals.}
    \label{fig:mf-ts-weekly-median}
  \end{subfigure}

  \caption{Weekly out-of-sample (one-week-ahead) predictions for residuals estimated using matrix factorization vs. two-stage approach. Figure \ref{fig:mf-ts-weekly-mean} shows the mean absolute residuals, with error bars, for each weeks' predictions. Figure \ref{fig:mf-ts-weekly-median} shows the median absolute residuals for each weeks' predictions.}
  \label{fig:mf-ts-weekly-in-oos}
\end{figure}

\section*{A Behavioral Model of Strategic Conformity}
What counts as “truth” in content moderation is not straightforward. Whether a note is \textit{helpful} can depend on context, values, and viewpoint, and in many settings there is no exogenous ground truth that the platform can audit against. Yet deployed systems necessarily act as if there is a stable  truth to infer. In particular, both X’s Community Notes model and related designs (including Meta’s) are built around an implicit premise: contributors receive noisy private signals about a latent underlying evaluation (as in Equation~\ref{eq:r_un}), and aggregation can recover that latent helpfulness of a note when signals are diverse and independent.

In this section, we \textbf{take that premise at face value}, and assume that contributors truly do observe latent signals that fit the platform’s modeling assumptions. We then ask \textbf{even if this latent-signal model were correct, would a consensus-based auditing system recover the underlying evaluation once contributors anticipate the platform’s eventual consensus?} 

Our analysis shows that it does not. When contributors are rewarded or penalized based on agreement with anticipated consensus, they have incentives to report strategically rather than reporting their private evaluations (signals) about the content. This systematically biases the platform’s inferred quantities, particularly on controversial content where conformity pressure is strongest.
 Finally, we show how shifting auditing from \textit{agreement with the majority} to \textit{out-of-sample stability of residual behavior} can remove the direct incentive to match the anticipated platform outcome and gives an unbiased estimator for helpfulness of a note.

% Since we do not observe the counterfactual for our weighted matrix factorization method, we turn to a theoretical model of user behavior to study how minority rater alignment and the coverage of controversial topics would change under our method. Heuristically, our model assumes that users observe latent signals on notes according to the matrix factorization model. Then, to report a rating on a note, users maximize a utility function, which trades off truth-telling against conformity to the anticipated consensus. The balance between truth-telling and anticipated consensus in the utility function is governed by a function $\rho(\cdot)$, which is a function of the note controversy. The more controversial the note, the higher weight users place on conforming to the majority in their utility function. We show that if $\rho\equiv 1$, users report truthfully and the matrix factorization estimates for the note helpfulness are consistent. However, any deviation $\rho<1$ may induce systematic bias toward consensus, leading to sign flips among minority users, an underestimation of the minority share, and biased estimates of the note helpfulness. We then show that an augmented version of our two-stage matrix factorization approach induces truthful reporting, which gives consistent estimates of the true note helpfulness, and also leads to predictions that have asymptotically lower variance than the current matrix factorization approach.

\subsection*{Model} We consider a system with $U$ users and $N$ notes where users $u$ rate notes $n$ that appear on their feeds. Following the latent-signal interpretation implicit in matrix factorization \cite{commnotes2022code}, suppose each user, suppose that each user $u$ observes a latent signal $s_{un}$ on a note $n$ with additive noise $\epsilon_{un}$
\[
    r_{un}^\star = s_{un} + \epsilon_{un}, \qquad s_{un} = \mu + h_u + i_n + f_u g_n
\]
where $\E[\epsilon_{un}] = 0, \E[\epsilon_{un}^2] := \sigma_u^2 \in (0,  \infty)$. The errors $\epsilon_{un}$ are independent across users $u$ and i.i.d. for the same $u$.  As in \eqref{eq:r_un}, the main quantity of interest is \(i_n\), which captures the perceived overall helpfulness of note \(n\).
Here, \(\mu\) is a global intercept, and \(h_u\) is a rater intercept capturing user \(u\)'s baseline tendency to rate notes as helpful rather than unhelpful, independent of the note's content. The latent variables \(f_u,g_n\in\mathbb{R}\) are user and note factors, respectively, and their product \(f_u g_n\) captures the extent to which a note is viewed as more or less helpful by users with different latent positions.
% We use $\mu^0, h_u^0, i_n^0, f_u^0, g_n^0$ to denote the true parameter values. We assume that $\mu^0, h_u^0, i_n^0$ are fixed but unknown constants, and that $f_u^0, g_n^0$ are drawn from some distribution with assumptions specified later.

The platform does not observe the latent quantity \(r_{un}^\star\) directly. Instead, it observes reported ratings, which we denote by \(a_{un}\), and fits the matrix factorization model in \eqref{eq:matrix_factorization_eq} to these reports using ridge-regularized least squares. Specifically, the platform computes
\begin{equation}
\label{eq:matrix_factorization_ridge}
\begin{aligned}
    \arg\min_{\mu,h,i,f,g} \;
    &\sum_{(u,n)\in\Omega}
    \bigl(a_{un}-\mu-h_u-i_n-f_u g_n\bigr)^2 \\
    &\quad + \lambda_h \|h\|_2^2
    + \lambda_i \|i\|_2^2
    + \lambda_f \|f\|_2^2
    + \lambda_g \|g\|_2^2 ,
\end{aligned}
\end{equation}
where \(\Omega\) denotes the set of observed ratings.
Let \(\hat\mu,\hat h_u,\hat i_n,\hat f_u,\hat g_n\) denote the resulting estimates. 
% Since this decomposition is not unique without normalization, we impose the standard identifiability constraints
% \begin{equation*}
%     \sum_{u \text{ observed}} h_u = 0 , \quad \sum_{n \text{ observed}} i_n = 0, \quad \frac{1}{U}\sum_{u \text{ observed}} f_u^2 = 1
% \end{equation*}
% \begin{equation*}
%    \sum_{n \text{ observed}} g_n = 0, \quad g_1 = 1,
% \end{equation*}
% which allow us to uniquely identify the intercept and factor terms. 

% $\mu$ is a global intercept, $h_u$ is a {rater intercept} capturing user $u$'s baseline agreeability (the tendency to mark notes as helpful rather than unhelpful, regardless of content),  and $f_u, g_n \in \mathbb{R}$ are latent {rater} and {note factors} whose product represents the ideological alignment between user and note.

\subsubsection*{Modeling user's behavior}

% We capture this by positing that each contributor’s perceived consensus is a noisy signal around a systematic note-level benchmark. 
% Importantly, our results do not
% require all contributors to observe the same benchmark exactly: heterogeneity
% in perceived consensus generally adds noise and can contaminate factor estimates,
% but the key distortion remains---ratings are pulled away from $r^\star_{un}$
% toward an anticipated platform outcome, so the platform no longer observes the
% latent signal even in large samples.

% To formalize the functional form of this note-level benchmark, we use the
% platform’s own low-dimensional aggregation structure. Whether a contributor uses
% a simple heuristic (treating consensus as an average of past ratings) or a more
% sophisticated forecast (approximating the score implied by rerunning the
% platform’s rank-1 matrix factorization on historical data), their forecast
% necessarily averages over other users. Under rank-1 structure, this averaging
% collapses user-side heterogeneity into (i) an intercept shift and (ii) a single
% loading on the note’s latent position. We therefore model the anticipated
% consensus target as an affine function along the note axis, with a loading that
% may vary with controversy.

Consensus-based auditing ties a contributor’s future standing (eligibility or
influence) to whether their ratings agree with the platform’s eventual aggregate
outcome.  At the time of rating, however, that outcome has not yet been realized. Contributors therefore form expectations about it using information that is broadly observable on the platform, including past notes and ratings, visible patterns in prior outcomes, and the aggregation rule itself.

These expectations may differ across contributors, reflecting differences in attention or inference. Still, because they are formed from largely shared information, it is natural to model them as noisy forecasts of a common note-level consensus target.
\begin{definition}
    Let $c_n \in [0, 1]$ be a scalar metric that denotes the controversy of the note $n$. For a note $n$, 
let \(m_n\) denote the anticipated platform consensus, that is, the outcome contributors expect the platform eventually to assign to that note.
\end{definition}

We leave \(m_n\) unrestricted. It may reflect a simple heuristic, such as averaging visible prior ratings, or a more sophisticated forecast of the score implied by the platform's aggregation rule. Contributor-specific differences in attention or inference are captured as additional noise around this target. Specifically, contributor \(u\) observes
\[
    \tilde m_{un} = m_n(c_n) + \epsilon_{un}^m
\]
where $\epsilon^m_{un}$ captures heterogeneity in perceptions across users with $\E[\epsilon_{un}^m \mid f_u, g_n, c_n, \rho_n] =0$ and $\var(\epsilon_{un}^m \mid f_u, g_n, c_n, \rho_n) = \sigma_{m, u}^2(c_n)$.  Allowing $\sigma_{m ,u}^2(\cdot)$
to depend on $c_n$ captures the idea that forecasts of the platform outcome may
be noisier on more controversial notes.

We next model how contributors choose ratings.   A contributor $u$ chooses a latent report $a_{un} \in \mathbb{R}$ on note $n$ by balancing two objectives: (i) reporting their private signal $r^\star_{un}$ against (ii) aligning with what they expect the platform to treat as the consensus, which we denote by $\tilde m_{un}$. 
The first term captures truthful reporting. The second captures the incentive 
created by consensus-based auditing: contributors anticipate that future 
standing on the platform (such as eligibility or influence) depends on whether 
their ratings agree with the platform’s eventual aggregate outcome. Deviating 
from the anticipated outcome therefore carries an expected penalty.

We allow this tradeoff to depend on the controversy level $c_n$. On more controversial notes, the incentive to anticipate the platform's outcome is plausibly stronger, so the weight placed on the contributor's own signal is weaker. We capture this expected downstream consequence of disagreement using a smooth quadratic loss. Similar tensions between private information and social conformity arise in models of social learning and information cascades  \cite{banerjee1992simple,Smith2000ObservationalLearning,Acemoglu2011BayesianLearning,acemoglu2022fast}.
\begin{definition}[User's Utility]
    User $u$'s utility for a report $a \in \R$ on note $n$ is defined to be 
    \begin{equation}\label{eq:user_period_utility}
        U_{un}(a \mid r_{un}^\star, c_n)= -\frac{\rho(c_n)}{2}(a - r_{un}^\star)^2 - \frac{1-\rho(c_n)}{2}(a-\tilde m_{un}(c_n))^2 + \zeta_{un}.
    \end{equation}
    Here, $\rho(\cdot) \in [0, 1]$ is conformity weight that is weakly decreasing in $c_n$; it reflects the intuition that as the controversy of a note decreases, users are more inclined to report their true signal, and as the controversy of a note increases, users are more inclined to conform to the majority. $\zeta_{un}$ is an idiosyncratic payoff shock (mean zero and finite variance, i.i.d. across users $u$ and notes $n$) that generates residual noise in choices independent of action $a$\footnote{In X's model, reports are discretized to $\{0,0.5,1\}$, so we interpret user actions $a$ as a latent index; the observed rating maps negative values of $a$ to $0$, positive values to $1$, and $0$ to $0.5$. Our MF is fit to the observed ratings, while the analysis proceeds on the latent index.}.
\end{definition}
We model the expected downstream consequence of disagreement using a smooth quadratic loss, which serves as a reduced-form approximation to the platform’s discrete eligibility and impact rules. Similar tensions between private information and social conformity arise in models of social learning and information cascades

% If $c_n\in\{0,1\}$ we write $\rho_0:=\rho(0)\succeq \rho_1:=\rho(1)$ to emphasize that controversial notes (\(c_n{=}1\)) place less weight on truth and more on the consensus projection. 

Maximizing a user's utility \eqref{eq:user_period_utility} immediately implies that the optimal latent report is the convex combination
\begin{equation}\label{eq:reporting}
a_{un}^\star=\rho(c_n) r^{\star}_{un} +(1-\rho(c_n)) \tilde m_{un}(c_n).
\end{equation}

Thus, reports place weight \(\rho(c_n)\) on the contributor's private signal and weight \(1-\rho(c_n)\) on the anticipated platform consensus. 
When $\rho(c_n)=1$, the contributor reports their private signal exactly; we refer to this case as \textit{truthful reporting}.
Because \(\rho(\cdot)\) is weakly decreasing in controversy, more controversial notes lead contributors to place relatively less weight on their own signals and more weight on the anticipated platform outcome.

\subsection*{Results}
In this section, we present our main theoretical results.  Proofs and technical regularity conditions are given in Appendix \ref{sec:proofs}. 
Throughout, we work in the setting where the number of users and notes are growing at the same asymptotic rate.
In reality, the platform observes only a subset of user–note interactions, which we model as missing-completely-at-random sampling: each rating \(a_{un}\) is observed independently with probability \(p\in(0,1]\).  
% In reality, the platform cannot observe all user-note interactions, so we assume that the platform observes ratings $a_{un}$ at random with constant probability $p \in (0, 1]$ independently of all other latent variables in the system. This condition is called missing completely at random (MCAR).
% We denote the resulting observation matrix by \(\Omega\), with entries \(\omega_{un}\sim \mathrm{Bernoulli}(p)\). 
We also assume that the true $\{(h_u, f_u)\}_u$ and $\{(i_n, g_n)\}_u$ are bounded and i.i.d., with finite variance and that all intercept and factor terms are mutually independent. In addition, we assume that $h_u, i_n, g_n$ are mean zero, but that $\E[f_u] = \mu_f$, for some known positive constant $\mu_f$; this allows us to identify a majority and minority group of users. The boundedness assumption is typically observed in real-life data, since user and item attributes are generally finite or are normalized by construction. Finally, we assume mean-zero sub-Gaussian noise terms  $\epsilon_{un}$ and $\epsilon_{un}^m$ independent of the latent variables. Formal regularity conditions are given in Appendix \ref{sec:proofs} Assumption 1.

% \begin{assumption}\label{asmp:user}
%     Assume that  with $\E[h_u] = 0, \E[f_u] = c$, for some constant $c>0$, which is known to the platform, and $\var(h_u) = \sigma_h^2 > 0, \var(f_u) = \sigma_f^2 > 0$. Furthermore, assume that $h_u$ and $f_u$ are independent and that $h_u, f_u$ are bounded. 
% \end{assumption}
% \begin{assumption}\label{asmp:note}
%     Assume that $\{(i_n, g_n)\}_u$ are i.i.d. with $\E[i_n] = 0, \E[g_n] = 0$, and $\var(i_n) = \sigma_i^2 > 0, \var(g_n) = \sigma_g^2 > 0$. Furthermore, assume that $i_n$ and $g_n$ are independent and that $i_n, g_n$ are bounded, i.e. $|i_n| \leq B_i, |g_n| \leq B_g$. 
% \end{assumption}
% \begin{assumption}\label{asmp:independence}
%      The user-side variables \(\{(h_u,f_u)\}_{u=1}^U\) are independent of the note-side variables \(\{(i_n,g_n)\}_{n=1}^N\).
% \end{assumption}
% \begin{assumption}\label{asmp:sub-guassian}
%     The noise variables \(\epsilon_{un}, \epsilon_{un}^m\) are independent, mean-zero, \(\sigma_u\) sub-Gaussian, and independent of all latent variables (see e.g., \cite{vershynin2012introduction} Definition 5.7).
% \end{assumption}
% \begin{assumption}\label{asmp:mf_centering}
%     Assume that 
%     \[
%         \sum_u h_u^0 = 0 \qquad \text{and} \qquad \sum_n i_n^0 = 0.
%     \]
% \end{assumption}
% \begin{assumption}\label{asmp:observations}
%     Assume that each user and each note have at least one observation. Note that in practice if this is not the case, such users and notes are not considered in factor estimation. 
% \end{assumption}
\subsubsection*{Private-Signal Reporting ($\rho=1$)}
We begin with the benchmark case \(\rho(\cdot)\equiv 1\), so that contributors report their private signals and \(a_{un}^\star=r_{un}^\star\). In this setting the platform observes noisy realizations of the latent-signal model itself, and the matrix factorization recovers note helpfulness.
This result extends the analogous result in the interactive fixed-effects framework in \cite{bai2009panel} to our setting with missing data using techniques from the matrix completion and interactive fixed-effects literature \cite{chen2020noisy, moon2015linear, su2026estimation}.

\begin{theorem}\label{thm:truth-recovery}
   Assume $U,N\to\infty$ and that $\E[f_u] = \mu_f$ where $\mu_f$ is known. Then, the estimate for note helpfulness is consistent. In particular 
    % there exists a constant $c$ such that,
    % \[
    %     \hat f_u \xrightarrow{p} cf_u^0, \qquad \hat g_n \xrightarrow{p} c^{-1} g_n^0,
    % \]
    % and
    \[
        \hat i_n \xrightarrow{p} i_n^0.
    \]
    That is, in the truthful regime, rank-1 MF recovers the true note helpfulness.
\end{theorem}

\subsubsection*{Strategic Conformity ($\rho<1$)} Next, we turn to analyzing the case when $\rho(\cdot) \not\equiv 1$, in which contributors place positive weight on the anticipated platform consensus. In this regime, the following theorem tells us that the estimate of note helpfulness $\hat i_n$ will be biased, and in particular will not converge to the model-implied helpfulness $i_n$. For the next theorem, we assume that the anticipated platform consensus $m_n$ varies across notes and has finite variance; formal regularity conditions are given in Appendix \ref{sec:proofs} Assumption 2. 

\begin{theorem}\label{thm:conformity}
    Suppose $\rho(\cdot) \not\equiv 1$. As $U, N \to\infty$, there exists a random variable $i_n^\infty$ such that
    \[
        \hat i_n \xrightarrow{p} i_n^\infty, 
    \]
    and for at least one $n$, $i_n^\infty \neq i_n^0.$
\end{theorem}

The source of this bias comes from the conformity incentives. When \(\rho(c_n)<1\), the platform no longer observes contributors' latent signals directly. Instead, it observes reports that combine the private signal \(r^\star_{un}\) with the anticipated consensus target \(m_n\). Thus, matrix factorization is applied to conformity-distorted signals rather than to the latent signal matrix itself. 

To see where the distortion enters, write
\[
\delta_n := m_n-(\mu+i_n).
\]
The quantity \(\delta_n\) measures how far the anticipated platform consensus deviates from the note-side latent component of the signal.   Under users’ strategic reporting, the observed report matrix differs from the latent signal matrix by a note-side perturbation proportional to \((1-\rho(c_n))\delta_n\). The platform nevertheless fits the same matrix factorization model. Consequently, the recovered parameters correspond to the best rank-1 approximation of this distorted matrix. In particular, the bias is governed by the projection of the conformity term \((1-\rho(c_n))\delta_n\) onto the note-factor direction \(g_n\). Whenever this projection is nonzero on a nontrivial set of notes, the resulting factorization converges to a parameter \(i_n^\ast\) which is different from the true note helpfulness \(i_n^0\).

% This distortion also appears in the inferred user factors. 
Furthermore, we can characterize the extent to which user factors will shift in this regime due to strategic behavior. Note that the latent factors $f_u$ and $g_n$ are unique up to a global constant scaling factor. Multiplying all note factors by $c$ and user factors by $c^{-1}$ still leads to a valid solution for solving the matrix factorization optimization problem given by \eqref{eq:matrix_factorization_eq}. Thus, we use the following normalization for identification. Define the (estimated) residualized ratings to be
\[
    y_{un} \;:=\; a_{un}^\star - \hat\mu - \hat h_u - \hat i_n
\]
where $\hat \mu, \hat h_u, \hat i_n$ are given by solving the least squares problem \eqref{eq:matrix_factorization_eq}. The latent factor normal equations are given by 
\begin{equation}
\hat f_u \;=\; \frac{\sum_n\omega_{un} y_{un}\,g_n}{\sum_n \omega_{un}g_n^2},
\qquad
\hat g_n \;=\; \frac{\sum_u \omega_{un} y_{un}\,f_u}{\sum_u \omega_{un}f_u^2}.
\label{eq:normal-eq}
\end{equation}
The identifiability conditions on the matrix factorization estimates allow us to determine the sign and scale of the factor estimates.

Our next result states how the estimate of the user factor behaves as a function of $\rho(c_n), m(c_n),$ and $g_n$ in expectation. 

\begin{theorem}
\label{lem:user-normal}
{Let $\rho \not \equiv 1$. Consider the setting of Theorem \ref{thm:conformity}, and suppose that the true note factors $\{g_n\}_n$ are known.
Let $\hat \mu, \hat h_u, \hat i_n, \hat f_u$ denote the solution to \eqref{eq:matrix_factorization_eq}. Then
\[
    \E[\hat f_u \mid f_u, g_n, c_n] = w_1 f_u + c(1-w_1) + o_p(1),
\]
where 
\[
    \qquad w_1 = \frac{\sum_n \rho_n g_n^2}{\sum_n g_n^2}.
\]}
\end{theorem}

Theorem \ref{lem:user-normal} tells us that $\mathbb{E}[\hat f_u\mid f_u]$ is an affine transformation of the user's truth $f_u$. {As the probability of encountering controversial notes increases (more notes with $c_n=1$),  $w_1$ \emph{decreases}.} Thus, the estimated user factor places less weight on the individual’s true position and more weight on the aggregate conformity distortion.  In particular, the bias is driven by the dependence of ratings on anticipated consensus, rather than by estimation error in the note factors. In practice, the ${g_n}_n$ are themselves estimated, so additional distortion may arise from estimation error, potentially amplifying the effect described above.

Theorem \ref{lem:user-normal} also implies the following proposition, telling us which members of the minority are more susceptible to measurement errors in their latent factor estimations. Let the true minority be $\{u:f_u<0\}$ with share $\pi^-_{\text{true}}:=\Pr(f_u<0)\in(0,\tfrac12)$, since we assume that $f_u$ has positive expected value. Because the latent-factor sign is globally arbitrary, the theoretical section uses a normalization opposite to the empirical sections. Empirically, we flip signs so that the majority is negative and the minority positive; all sign-based claims are invariant under the global transformation $(f, g) \to (-f, -g)$.

\begin{proposition}
\label{prop:flip}
Consider the setting of Theorem \ref{lem:user-normal}. Then, 
\[
    \E[\hat f_u^* \mid f_u, g_n, c_n] > 0 \qquad \text{if and only if} \qquad f_u > -\frac{c(1-w_1)}{w_1}.
\]
In particular, the minority slice $(-\frac{c(1-w_1)}{w_1}, 0)$ is mapped to positive estimates.
\end{proposition}

\begin{proposition}
\label{prop:shrink}
Consider the setting of Theorem \ref{lem:user-normal}. Let $F$ be the CDF of $f_u$. Assume that $F$ is continuous with no atom at $0$. Then, the estimated minority share
\[
    \pi^-_{\text{est}} := \Pr(\hat f_u<0 \mid f_u, g_n, c_n)
\]
satisfies
\[
    \pi^-_{\text{est}} = F\left(-\frac{c(1-w_1)}{w_1}\right) < F(0) +o(1) =\pi^-_{\text{true}} +o(1),
\]
and $\pi^-_{\text{est}}$ is (weakly) decreasing in $w_2$ and (weakly) increasing in $w_1$.
\end{proposition}

% \begin{remark}
%     \kh{Our model reveals that user-factor distortion is larger when more weight is placed on controversial notes. To see this mathematically, note that if $\rho(c)$ is non-increasing in $c$, then $w_1$ decreases and $1-w_1$ increases. Then, by Propositions~\ref{prop:flip}-\ref{prop:shrink}, the flip moves left and $\pi^-_{\text{est}}$ decreases as $\bar c$ rises.}
% \end{remark}

\begin{remark}
    One can carry a similar computation for the note factors $\E[\hat g_n\mid f_u, g_n, c_n]$ to get
    \[
        \mathbb{E}[\hat g_n \mid f_u,g_n,c_n] \approx g_n\rho_n \left(1 - c\frac{\sum_u f_u}{\sum_u f_u^2}\right).
    \]
    If user opinions are balanced ($\mathbb{E}(f_u)=0$), then the expected estimate $\mathbb{E}[\hat g_n \mid f_u, g_n, c_n]$ is proportional to $g_n$ with a coefficient that depends on $\rho_n(c_n)$. In particular, higher value of $c_n$ (controversial notes) reduce this coefficient, shrinking the note factor toward zero.
\end{remark}

\subsubsection*{Statistical guarantee for the two-stage estimator}
We now turn to the alternative auditing rule studied in the empirical section. Recall  \emph{two-stage} algorithm: the platform first fits the standard unweighted regularized matrix factorization model, and then uses the resulting residuals to estimate contributor-specific noise levels. It then refits the same model using inverse residual variance weights. This is the matrix-factorization analogue of feasible generalized least squares and is closely related to weighted low-rank approximation \cite{Aitken1936least, srebro2003weighted, udell2016glrm}.

In the first stage, the platform fits the unweighted regularized model in \eqref{eq:matrix_factorization_ridge} and obtains estimates \((\hat\mu,\hat h,\hat i,\hat f,\hat g)\). For each contributor \(u\), let
 \(S_u := \{n : (u,n)\in\Omega\}\) and \(N_u := |S_u|\), and define the first-stage residual variance estimate
\begin{equation}
\label{eq:first-stage-variance}
\hat\sigma_u^2
:=
\frac{1}{N_u}
\sum_{n\in S_u}
\bigl(a_{un}-\hat\mu-\hat h_u-\hat i_n-\hat f_u\hat g_n\bigr)^2.
\end{equation}

In the second stage, the platform sets
\[
\hat w_u := \frac{1}{\hat\sigma_u^2}
\]
and refits the same regularized rank-1 model using contributor-specific weights. More generally, for any bounded positive weights \(w=\{w_u\}\), define the weighted regularized matrix factorization problem
\begin{equation}\label{eq:weighted-matrix-factorization}
    \arg \min_{\tilde \mu, \tilde h_u, \tilde i_n, \tilde f_u, \tilde g_n} \sum_{\substack{(u,n)\\ \text{observed}}} w_u\left(r_{un}-\tilde{\mu}-\tilde{h}_u -\tilde{i}_n-\tilde f_u \cdot \tilde g_n \right)^2.
\end{equation}
We write \(\tilde i_n^{\mathrm{ts}}\) for the note-helpfulness estimate produced by \eqref{eq:weighted-matrix-factorization} with weights \(\hat w_u = 1/\hat\sigma_u^2\).

Our next theorem shows that, among estimators obtained in this way, the estimator with weights $w_u = 1/\hat\sigma_u^2$ is consistent and has the lowest asymptotic variance.

\begin{theorem}
Assume that $\rho \equiv 1$ and that $\mu_f$ is known. Then the solution to \eqref{eq:weighted-matrix-factorization} with weights $1/\hat\sigma_u^2$, $\tilde i_n^\ts$, recovers consistent estimates of $i_n$, i.e., as $U, N\to\infty$
    \[
        \tilde i_n^\ts \xrightarrow{p} i_n^0.
    \]
    Moreover, among all other solutions of \eqref{eq:weighted-matrix-factorization} $\tilde i_n$ with positive, finite weights $w_u \in (0, \infty)$,
    \[
        \avar(\tilde i_n^\ts) \leq \avar(\tilde i_n).
    \]
    Here, for a scalar estimator $X_m$, 
    \[
        \avar(X_m) = \lim_{m\to\infty} m\cdot \var(X_m).
    \]
\end{theorem}

This theorem gives a statistical interpretation of contributor impact under the redesigned rule. In the two-stage estimator, contributors are weighted by the inverse of their residual variance, so contributors whose evaluations are more stable relative to the fitted latent structure receive greater influence in the second stage. In this sense, the redesign audits contributors by \emph{residual stability} rather than by agreement with the platform's eventual consensus.

This is the key contrast with consensus-based auditing. Under the current rule implemented in Community Notes, contributor influence is tied to whether ratings align with the platform's final aggregate outcome. Under the two-stage rule, influence is instead tied to the statistical precision of a contributor's evaluations within the latent-factor model. The theorem shows that, under private-signal reporting, this weighting rule yields a consistent estimator and attains the lowest asymptotic variance among weighted matrix factorization estimators in this class.

Whether the same rule also changes contributors' strategic incentives is a separate behavioral question.
The result here suggests an alternative notion of contributor \emph{rating impact} based on \emph{residual stability} rather than agreement with the platform's eventual consensus.

\section*{Conclusion}
Crowdsourced moderation systems are often motivated by the idea that diverse, independent evaluations can be aggregated into reliable judgments. Our findings show that this promise depends not only on how ratings are aggregated, but also on how contributors are audited. In X's Community Notes, auditing contributors by whether they agree with the platform's eventual consensus creates incentives to anticipate that consensus rather than to provide independent evaluations. Empirically, this is associated with strategic conformity by minority contributors, reduced engagement on controversial content, and lower predictive performance of the platform's latent-factor model.

Our theoretical analysis clarifies why this occurs. Even if one grants the platform's latent-signal model, consensus-based auditing alters the object being measured: once contributors partially conform to the anticipated platform outcome, matrix factorization no longer aggregates independent evaluations alone. Instead, it recovers a conformity-distorted projection of those evaluations. These observations also suggest a different design principle. Rather than rewarding contributors for matching the eventual majority outcome, platforms can evaluate them using targets that do not mechanically favor conformity. Motivated by this idea, we study a two-stage procedure that weights contributors by the stability of their residual behavior rather than by agreement with the final consensus. In the Community Notes data, this approach improves out-of-sample predictive performance while allowing informative disagreement to retain influence.

More broadly, our results suggest that crowdsourced moderation should be designed to preserve independence, especially on controversial content where finding misinformation is most valuable. Systems that reward agreement with the final aggregate may appear to improve reliability, but can instead suppress the disagreement needed for accurate aggregation. In environments without externally verifiable ground truth, the design of the auditing rule is therefore not a peripheral implementation detail; it is part of the core of the system.

\section*{Acknowledgments}
KH was partially supported by the National Science Foundation under grant DGE 2146752.

\newpage
% Bibliography
\bibliographystyle{plain} % Replace "plain" with your desired style
\bibliography{ref}

@misc{CommunityNotesDownloadData,
  author       = {{X (formerly Twitter) Community Notes Guide}},
  title        = {Downloading Data},
  howpublished = {\url{https://communitynotes.x.com/guide/en/under-the-hood/download-data}},
  note         = {Accessed: 2025-08-30},
  year = {n.d.}
}

@misc{CommunityNotesRankingNotes,
  author       = {{X Community Notes Guide}},
  title        = {Ranking Notes},
  howpublished = {\url{https://communitynotes.x.com/guide/en/under-the-hood/ranking-notes}},
  note         = {Accessed: 2025-08-30},
  year = {n.d.}
}

@misc{CommunityNotesWritingAbility,
  author       = {{Community Notes Guide – X}},
  title        = {Locking and unlocking the ability to write notes},
  howpublished = {\url{https://communitynotes.x.com/guide/en/contributing/writing-ability}},
  note         = {Accessed: 2025-08-30},
  year = {n.d.}
}

@misc{CommunityNotesWritingAndRatingImpact,
  author       = {{Community Notes Guide – X}},
  title        = {Rating and Writing Impact},
  howpublished = {\url{https://communitynotes.x.com/guide/en/contributing/writing-and-rating-impact}},
  note         = {Accessed: 2025-08-30},
  year = {n.d.}
}

@article{acemoglu2022fast,
  title={Learning from reviews: The selection effect and the speed of learning},
  author={Acemoglu, Daron and Makhdoumi, Ali and Malekian, Azarakhsh and Ozdaglar, Asuman},
  journal={Econometrica},
  volume={90},
  number={6},
  pages={2857--2899},
  year={2022},
  publisher={Wiley Online Library}
}

@article{acemoglu2024model,
  title={A model of online misinformation},
  author={Acemoglu, Daron and Ozdaglar, Asuman and Siderius, James},
  journal={Review of Economic Studies},
  volume={91},
  number={6},
  pages={3117--3150},
  year={2024},
  publisher={Oxford University Press UK}
}

@article{renault2025republicans,
  title={Republicans are flagged more often than Democrats for sharing misinformation on X’s Community Notes},
  author={Renault, Thomas and Mosleh, Mohsen and Rand, David G},
  journal={Proceedings of the National Academy of Sciences},
  volume={122},
  number={25},
  pages={e2502053122},
  year={2025},
  publisher={National Academy of Sciences}
}

@article{miller2005eliciting,
  title={Eliciting informative feedback: The peer-prediction method},
  author={Miller, Nolan and Resnick, Paul and Zeckhauser, Richard},
  journal={Management Science},
  volume={51},
  number={9},
  pages={1359--1373},
  year={2005}
}

@inproceedings{witkowski2012peer,
  title={Peer prediction without a common prior},
  author={Witkowski, Jens and Parkes, David C},
  booktitle={Proceedings of the 13th ACM Conference on Electronic Commerce},
  pages={964--981},
  year={2012}
}

@inproceedings{shnayder2016informed,
  title={Informed truthfulness in multi-task peer prediction},
  author={Shnayder, Victor and Agarwal, Arpit and Frongillo, Rafael and Parkes, David C},
  booktitle={Proceedings of the 2016 ACM Conference on Economics and Computation},
  pages={179--196},
  year={2016}
}

@article{raykar2010learning,
  title={Learning from crowds.},
  author={Raykar, Vikas C and Yu, Shipeng and Zhao, Linda H and Valadez, Gerardo Hermosillo and Florin, Charles and Bogoni, Luca and Moy, Linda},
  journal={Journal of machine learning research},
  volume={11},
  number={4},
  year={2010}
}

@article{karger2011iterative,
  title={Iterative learning for reliable crowdsourcing systems},
  author={Karger, David and Oh, Sewoong and Shah, Devavrat},
  journal={Advances in neural information processing systems},
  volume={24},
  year={2011}
}

@article{dawid1979maximum,
  title={Maximum likelihood estimation of observer error-rates using the EM algorithm},
  author={Dawid, Alexander Philip and Skene, Allan M},
  journal={Journal of the Royal Statistical Society: Series C (Applied Statistics)},
  volume={28},
  number={1},
  pages={20--28},
  year={1979},
  publisher={Wiley Online Library}
}

@article{eyster2010naive,
  title={Naive herding in rich-information settings},
  author={Eyster, Erik and Rabin, Matthew},
  journal={American economic journal: microeconomics},
  volume={2},
  number={4},
  pages={221--243},
  year={2010},
  publisher={American Economic Association}
}

@article{prelec2004bayesian,
  title={A Bayesian truth serum for subjective data},
  author={Prelec, Drazen},
  journal={Science},
  volume={306},
  number={5695},
  pages={462--466},
  year={2004}
}

@article{van2010manipulation,
  title={Manipulation robustness of collaborative filtering},
  author={Van Roy, Benjamin and Yan, Xiang},
  journal={Management Science},
  volume={56},
  number={11},
  pages={1911--1929},
  year={2010}
}

@article{muchnik2013social,
  title={Social influence bias: A randomized experiment},
  author={Muchnik, Lev and Aral, Sinan and Taylor, Sean J},
  journal={Science},
  volume={341},
  number={6146},
  pages={647--651},
  year={2013}
}

@article{noelle1974spiral,
  title={The spiral of silence a theory of public opinion},
  author={Noelle-Neumann, Elisabeth},
  journal={Journal of communication},
  volume={24},
  number={2},
  pages={43--51},
  year={1974},
  publisher={Oxford University Press}
}

@article{zhao2025spiral,
  title   = {Mapping the Spiral of Silence: Surveying Unspoken Opinions in Online Communities},
  author  = {Zhao, Dora and Yang, Diyi and Bernstein, Michael S.},
  journal = {arXiv preprint arXiv:2502.00952},
  year    = {2025},
  doi     = {10.48550/arXiv.2502.00952}
}

@article{allen2024quantifying,
  title={Quantifying the impact of misinformation and vaccine-skeptical content on Facebook},
  author={Allen, Jennifer and Watts, Duncan J and Rand, David G},
  journal={Science},
  volume={384},
  number={6699},
  pages={eadk3451},
  year={2024},
  publisher={American Association for the Advancement of Science}
}

@techreport{horvitz2012incentives,
  title={Incentives and truthful reporting in consensus-centric crowdsourcing},
  author={Horvitz, Eric},
  institution={Microsoft Research},
  year={2012},
  url={https://www.microsoft.com/en-us/research/wp-content/uploads/2012/02/Incentives-MSR-TR.pdf}
}

@article{brashier2021timing,
  title={Timing matters when correcting fake news},
  author={Brashier, Nadia M and Pennycook, Gordon and Berinsky, Adam J and Rand, David G},
  journal={Proceedings of the National Academy of Sciences},
  volume={118},
  number={5},
  pages={e2020043118},
  year={2021},
  publisher={National Academy of Sciences}
}

@article{jackson2021finding,
  title={Finding the wise and the wisdom in a crowd: Estimating underlying qualities of reviewers and items},
  author={Carayol, Nicolas and Jackson, Matthew O},
  journal={The Economic Journal},
  volume={134},
  number={663},
  pages={2712--2745},
  year={2024},
  publisher={Oxford University Press}
}

@online{Perez2022Birdwatch,
  author    = {Perez, Sarah},
  title     = {Twitter expands its crowdsourced fact-checking program {Birdwatch} ahead of {US} midterms},
  year      = {2022},
  month     = {September},
  day       = {7},
  url       = {https://techcrunch.com/2022/09/07/twitter-expands-its-crowdsourced-fact-checking-program-birdwatch-ahead-of-u-s-midterms/},
  journal   = {TechCrunch},
  note      = {Accessed: 2025-09-16}
}

@online{Perez2022BirdwatchVisible,
  author    = {Perez, Sarah},
  title     = {Twitter is making its crowdsourced fact-checks visible to all {U.S.} users with {Birdwatch} expansion},
  year      = {2022},
  month     = {October},
  day       = {6},
  url       = {https://techcrunch.com/2022/10/06/twitter-is-making-its-crowdsourced-fact-checks-visible-to-all-u-s-users-with-birdwatch-expansion/},
  journal   = {TechCrunch},
  note      = {Accessed: 2025-09-16}
}

@article{
vicario2016spreading,
author = {Michela Del Vicario  and Alessandro Bessi  and Fabiana Zollo  and Fabio Petroni  and Antonio Scala  and Guido Caldarelli  and H. Eugene Stanley  and Walter Quattrociocchi },
title = {The spreading of misinformation online},
journal = {Proceedings of the National Academy of Sciences},
volume = {113},
number = {3},
pages = {554-559},
year = {2016},
doi = {10.1073/pnas.1517441113},
URL = {https://www.pnas.org/doi/abs/10.1073/pnas.1517441113},
eprint = {https://www.pnas.org/doi/pdf/10.1073/pnas.1517441113}}

@misc{xcommunitynotes,
  author       = {{X Corp.}},
  title        = {About Community Notes on X},
  howpublished = {\url{https://help.x.com/en/using-x/community-notes}},
  note         = {Accessed: 2025-11-17},
  year         = {2025}
}

@misc{meta2025communitynotes,
  author       = {{Meta Platforms, Inc.}},
  title        = {Introducing Community Notes -- Adding Context to Posts},
  howpublished = {\url{https://www.meta.com/technologies/community-notes}},
  year         = {2025}
}

@article{perez2024bluesky,
  author       = {Sarah Perez},
  title        = {Bluesky adds ‘anti-toxicity’ tools and aims to integrate ‘a Community Notes-like’ feature in the future},
  journal      = {TechCrunch},
  date         = {2024-08-28},
  url          = {https://techcrunch.com/2024/08/28/bluesky-adds-anti-toxicity-tools-and-aims-to-integrate-a-community-notes-like-feature-in-the-future/#:~:text=On%20X%2C%20Community%20Notes%20serve%20as%20a,algorithm%20that%20works%20to%20find%20consensus%20among},
  note         = {Accessed: 2025-11-17},
  year = {2024}
}

@misc{tiktok2025footnotes,
  author       = {{TikTok Pte. Ltd.}},
  title        = {Rolling out TikTok Footnotes in the U.S.},
  howpublished = {\url{https://newsroom.tiktok.com/rolling-out-tiktok-footnotes-in-the-us?lang=en}},
  note         = {Accessed: 2025-11-17},
  year         = {2025}
}

@article{bhuiyan2020investigating,
  title={Investigating differences in crowdsourced news credibility assessment: Raters, tasks, and expert criteria},
  author={Bhuiyan, Md Momen and Zhang, Amy X and Sehat, Connie Moon and Mitra, Tanushree},
  journal={Proceedings of the ACM on Human-Computer Interaction},
  volume={4},
  number={CSCW2},
  pages={1--26},
  year={2020},
  publisher={ACM New York, NY, USA}
}

@misc{amatriain2009ilikeit,
author = {Amatriain, Xavier and Pujol, Josep and Oliver, Nuria},
year = {2009},
month = {06},
pages = {247-258},
title = {I Like It... I Like It Not: Evaluating User Ratings Noise in Recommender Systems},
volume = {5535},
isbn = {978-3-642-02246-3},
doi = {10.1007/978-3-642-02247-0_24}
}

@misc{galton1907vox,
  title={Vox populi},
  author={Galton, Francis},
  year={1907},
  publisher={Nature Publishing Group UK London}
}

@book{surowiecki2005wisdom,
  title={The wisdom of crowds},
  author={Surowiecki, James},
  year={2005},
  publisher={Vintage}
}

@article{banerjee1992simple,
  title={A simple model of herd behavior},
  author={Banerjee, Abhijit V},
  journal={The quarterly journal of economics},
  volume={107},
  number={3},
  pages={797--817},
  year={1992},
  publisher={MIT Press}
}

@article{Smith2000ObservationalLearning,
  author  = {Smith, Lones and S{\o}rensen, Peter},
  title   = {Pathological Outcomes of Observational Learning},
  journal = {Econometrica},
  year    = {2000},
  volume  = {68},
  number  = {2},
  pages   = {371--398},
  doi     = {10.1111/1468-0262.00113}
}

@article{Acemoglu2011BayesianLearning,
  author  = {Acemoglu, Daron and Dahleh, Munther A. and Lobel, Ilan and Ozdaglar, Asuman},
  title   = {Bayesian Learning in Social Networks},
  journal = {The Review of Economic Studies},
  year    = {2011},
  volume  = {78},
  number  = {4},
  pages   = {1201--1236},
  doi     = {10.1093/restud/rdr004}
}

@article{lorenz2011social,
  title={How social influence can undermine the wisdom of crowd effect},
  author={Lorenz, Jan and Rauhut, Heiko and Schweitzer, Frank and Helbing, Dirk},
  journal={Proceedings of the national academy of sciences},
  volume={108},
  number={22},
  pages={9020--9025},
  year={2011},
  publisher={National Academy of Sciences}
}

@article{
slaughter2025reduce,
author = {Isaac Slaughter  and Axel Peytavin  and Johan Ugander  and Martin Saveski },
title = {Community notes reduce engagement with and diffusion of false information online},
journal = {Proceedings of the National Academy of Sciences},
volume = {122},
number = {38},
pages = {e2503413122},
year = {2025},
doi = {10.1073/pnas.2503413122},
URL = {https://www.pnas.org/doi/abs/10.1073/pnas.2503413122},
eprint = {https://www.pnas.org/doi/pdf/10.1073/pnas.2503413122}}

@article{gao2024can,
  title   = {Can Crowdchecking Curb Misinformation? Evidence from Community Notes},
  author  = {Gao, Yang and Zhang, Maggie Mengqing and Rui, Huaxia},
  journal = {Information Systems Research},
  year    = {2025},
  doi     = {10.1287/isre.2024.1609}
}

@article{drolsbach2024community,
  title={Community notes increase trust in fact-checking on social media},
  author={Drolsbach, Chiara Patricia and Solovev, Kirill and Pr{\"o}llochs, Nicolas},
  journal={PNAS nexus},
  volume={3},
  number={7},
  pages={pgae217},
  year={2024},
  publisher={Oxford University Press US}
}

@article{drolsbach2023diffusion,
  title={Diffusion of community fact-checked misinformation on twitter},
  author={Drolsbach, Chiara Patricia and Pr{\"o}llochs, Nicolas},
  journal={Proceedings of the ACM on Human-Computer Interaction},
  volume={7},
  number={CSCW2},
  pages={1--22},
  year={2023},
  publisher={ACM New York, NY, USA}
}

@inproceedings{allen2022birds,
author = {Allen, Jennifer and Martel, Cameron and Rand, David G},
title = {Birds of a feather don’t fact-check each other: Partisanship and the evaluation of news in Twitter’s Birdwatch crowdsourced fact-checking program},
year = {2022},
isbn = {9781450391573},
publisher = {Association for Computing Machinery},
address = {New York, NY, USA},
url = {https://doi.org/10.1145/3491102.3502040},
doi = {10.1145/3491102.3502040},
booktitle = {Proceedings of the 2022 CHI Conference on Human Factors in Computing Systems},
articleno = {245},
numpages = {19},
location = {New Orleans, LA, USA},
series = {CHI '22}
}

@article{Aitken1936least, 
    title={{IV.} {O}n {L}east {S}quares and {L}inear {C}ombination of {O}bservations}, 
    volume={55}, 
    DOI={10.1017/S0370164600014346}, 
    journal={Proceedings of the Royal Society of Edinburgh}, 
    author={Aitken, A. C.}, 
    year={1936}, 
    pages={42–48}
}

@inproceedings{resnick2007influence,
author = {Resnick, Paul and Sami, Rahul},
title = {The influence limiter: provably manipulation-resistant recommender systems},
year = {2007},
isbn = {9781595937308},
publisher = {Association for Computing Machinery},
address = {New York, NY, USA},
url = {https://doi-org.libproxy.berkeley.edu/10.1145/1297231.1297236},
doi = {10.1145/1297231.1297236},
booktitle = {Proceedings of the 2007 ACM Conference on Recommender Systems},
pages = {25–32},
numpages = {8},
location = {Minneapolis, MN, USA},
series = {RecSys '07}
}

@article{carroll1982robust,
  author       = {Carroll, Raymond J. and Ruppert, David},
  title        = {Robust Estimation in Heteroscedastic Linear Models},
  journal      = {The Annals of Statistics},
  volume       = {10},
  number       = {2},
  pages        = {429--441},
  year         = {1982},
  publisher    = {Institute of Mathematical Statistics},
  doi          = {10.1214/aos/1176345784}
}

@article{dai2018aggregation,
  title={Aggregation of consumer ratings: an application to yelp. com},
  author={Dai, Weijia and Jin, Ginger and Lee, Jungmin and Luca, Michael},
  journal={Quantitative Marketing and Economics},
  volume={16},
  number={3},
  pages={289--339},
  year={2018},
  publisher={Springer}
}

@misc{commnotes2022code,
  author       = {{Twitter, Inc.}},
  title        = {Community Notes: Documentation and source code powering Community Notes},
  year         = {2022},
  howpublished = {\url{https://github.com/twitter/communitynotes}}
}

@article{Gneiting2007strictly,
author = {Tilmann Gneiting and Adrian E Raftery},
title = {Strictly Proper Scoring Rules, Prediction, and Estimation},
journal = {Journal of the American Statistical Association},
volume = {102},
number = {477},
pages = {359--378},
year = {2007},
publisher = {ASA Website},
doi = {10.1198/016214506000001437},
URL = { 
    
        https://doi.org/10.1198/016214506000001437
},
eprint = { 
        https://doi.org/10.1198/016214506000001437
}
}

@article{
west2021misinformation,
author = {Jevin D. West  and Carl T. Bergstrom },
title = {Misinformation in and about science},
journal = {Proceedings of the National Academy of Sciences},
volume = {118},
number = {15},
pages = {e1912444117},
year = {2021},
doi = {10.1073/pnas.1912444117},
URL = {https://www.pnas.org/doi/abs/10.1073/pnas.1912444117},
eprint = {https://www.pnas.org/doi/pdf/10.1073/pnas.1912444117}}

@article{
vosoughi2018spread,
author = {Soroush Vosoughi  and Deb Roy  and Sinan Aral },
title = {The spread of true and false news online},
journal = {Science},
volume = {359},
number = {6380},
pages = {1146-1151},
year = {2018},
doi = {10.1126/science.aap9559},
URL = {https://www.science.org/doi/abs/10.1126/science.aap9559},
eprint = {https://www.science.org/doi/pdf/10.1126/science.aap9559}}

@article{van2022misinformation,
  title={Misinformation: susceptibility, spread, and interventions to immunize the public},
  author={Van Der Linden, Sander},
  journal={Nature medicine},
  volume={28},
  number={3},
  pages={460--467},
  year={2022},
  publisher={Nature Publishing Group US New York}
}

@article{kangur2024checks,
  title={Who checks the checkers? exploring source credibility in twitter’s community notes},
  author={Kangur, Uku and Chakraborty, Roshni and Sharma, Rajesh},
  journal={Journal of Computational Social Science},
  volume={9},
  number={1},
  pages={24},
  year={2026},
  publisher={Springer}
}

@inproceedings{thebault2023diverse,
  title={Diverse perspectives can mitigate political bias in crowdsourced content moderation},
  author={Thebault-Spieker, Jacob and Venkatagiri, Sukrit and Mine, Naomi and Luther, Kurt},
  booktitle={Proceedings of the 2023 ACM Conference on Fairness, Accountability, and Transparency},
  pages={1280--1291},
  year={2023}
}

@online{CommunityNotesRankingAlgorithm,
  title        = {Note ranking algorithm},
  author       = {{X Corp. / Community Notes Guide}},
  url          = {https://communitynotes.x.com/guide/en/under-the-hood/ranking-notes},
  note         = {Accessed: 2025-08-05},
  organization = {Community Notes},
  howpublished = {\url{https://communitynotes.x.com/guide/en/under-the-hood/ranking-notes}},
  year = {n.d.}
}

@article{bai2009panel,
  title={Panel data models with interactive fixed effects},
  author={Bai, Jushan},
  journal={Econometrica},
  volume={77},
  number={4},
  pages={1229--1279},
  year={2009},
  publisher={Wiley Online Library}
}

@inproceedings{perdomo2020performative,
  title={Performative prediction},
  author={Perdomo, Juan and Zrnic, Tijana and Mendler-D{\"u}nner, Celestine and Hardt, Moritz},
  booktitle={International Conference on Machine Learning},
  pages={7599--7609},
  year={2020},
  organization={PMLR}
}

@inproceedings{kong2016putting,
  title={Putting peer prediction under the micro (economic) scope and making truth-telling focal},
  author={Kong, Yuqing and Ligett, Katrina and Schoenebeck, Grant},
  booktitle={International Conference on Web and Internet Economics},
  pages={251--264},
  year={2016},
  organization={Springer}
}

@inproceedings{waggoner2014output,
  title={Output agreement mechanisms and common knowledge},
  author={Waggoner, Bo and Chen, Yiling},
  booktitle={Proceedings of the AAAI Conference on Human Computation and Crowdsourcing},
  volume={2},
  pages={220--226},
  year={2014}
}

@article{kashima2024trustworthy,
  title={Trustworthy human computation: a survey},
  author={Kashima, Hisashi and Oyama, Satoshi and Arai, Hiromi and Mori, Junichiro},
  journal={Artificial Intelligence Review},
  volume={57},
  number={12},
  pages={322},
  year={2024},
  publisher={Springer}
}

@inproceedings{liu2017machine,
author = {Liu, Yang and Chen, Yiling},
title = {Machine-Learning Aided Peer Prediction},
year = {2017},
isbn = {9781450345279},
publisher = {Association for Computing Machinery},
address = {New York, NY, USA},
url = {https://doi-org.libproxy.berkeley.edu/10.1145/3033274.3085126},
doi = {10.1145/3033274.3085126},
booktitle = {Proceedings of the 2017 ACM Conference on Economics and Computation},
pages = {63–80},
numpages = {18},
location = {Cambridge, Massachusetts, USA},
series = {EC '17}
}

@book{faltings2022game,
  title={Game theory for data science: Eliciting truthful information},
  author={Faltings, Boi and Radanovic, Goran},
  year={2022},
  publisher={Springer Nature}
}

@misc{vershynin2012introduction,
  title={Introduction to the non-asymptotic analysis of random matrices.},
  author={Vershynin, Roman},
  year={2012}
}

@inproceedings{srebro2003weighted,
  title={Weighted low-rank approximations},
  author={Srebro, Nathan and Jaakkola, Tommi},
  booktitle={Proceedings of the 20th international conference on machine learning (ICML-03)},
  pages={720--727},
  year={2003}
}

@article{udell2016glrm,
  title={Generalized low rank models},
  author={Udell, Madeleine and Horn, Corinne and Zadeh, Reza and Boyd, Stephen},
  journal={Foundations and Trends in Machine Learning},
  volume={9},
  number={1},
  pages={1--118},
  year={2016},
  publisher={Emerald Publishing Limited}
}

@article{chen2020noisy,
  title={Noisy matrix completion: Understanding statistical guarantees for convex relaxation via nonconvex optimization},
  author={Chen, Yuxin and Chi, Yuejie and Fan, Jianqing and Ma, Cong and Yan, Yuling},
  journal={SIAM journal on optimization},
  volume={30},
  number={4},
  pages={3098--3121},
  year={2020},
  publisher={SIAM}
}

@inproceedings{farias2022uncertainty,
  title={Uncertainty quantification for low-rank matrix completion with heterogeneous and sub-exponential noise},
  author={Farias, Vivek and Li, Andrew A and Peng, Tianyi},
  booktitle={International Conference on Artificial Intelligence and Statistics},
  pages={1179--1189},
  year={2022},
  organization={PMLR}
}

@article{moon2015linear,
  title={Linear regression for panel with unknown number of factors as interactive fixed effects},
  author={Moon, Hyungsik Roger and Weidner, Martin},
  journal={Econometrica},
  volume={83},
  number={4},
  pages={1543--1579},
  year={2015},
  publisher={Wiley Online Library}
}

@article{su2026estimation,
  title={Estimation and inference for unbalanced panel data models with interactive fixed effects},
  author={Su, Liangjun and Wang, Fa and Wang, Yiren},
  journal={Journal of Econometrics},
  volume={255},
  pages={106222},
  year={2026},
  publisher={Elsevier}
}

@misc{blair2026structure,
    title={The Structure of Bridging},
    author={Blair, Carter and De Raaij, Jacob and Procaccia, Ariel and Rubin-Toles, Maxon and Shah, Nisarg and Si, Michelle and Wang, Serena},
    year={2026}
}
\appendix

\newpage

\section{Guide to the Appendix}

This Appendix has two goals. First, it provides the empirical implementation details and robustness analyses underlying the main-text results. Second, it contains the full technical Appendix for the theoretical results. The organization of the Appendix is designed to mirror the logic of the paper: we first document the empirical reconstruction and additional analyses, and then turn to the stylized model and proofs.

The Appendix is organized as follows. Appendix~\ref{sec:data} describes the data sources, preprocessing steps, and reconstruction of weekly latent factors from the public Community Notes data and open-source code. Appendix~\ref{sec:empiric} presents additional empirical analyses in the order of the main text, including the shift in minority-aligned contributors after the introduction of Rating Impact, the change in participation on controversial content, and additional predictive-performance analyses. Appendix~\ref{sec:twostage imp} describes the implementation of the two-stage weighted matrix factorization procedure and additional empirical details for that estimator. Appendix~\ref{sec:proofs} contains the full theory Appendix: it states the stylized estimation model, lists the regularity conditions used in the proofs, and provides proofs of all theorems and propositions in the main text.

\section{Data, reconstruction, and empirical methodology}\label{sec:data}

% \section{SI Methods and Implementation Details}
\subsection{Data sources}
In this paper, we use the open source code and data from X Community Notes \cite{CommunityNotesDownloadData}. To study the effect of the Rating Impact rollout, we focus on the time frame between June 1, 2022 and May 31, 2023. To evaluate the predictive performance of our two-stage matrix factorization (MF) method, we use ratings data from Jan. 1, 2023 to June 1, 2024. Information about the primary dataframes we use in our analysis and relevant columns are stored in Table \ref{tab:dataframes}. More detailed information about all available data can be found in the Community Notes documentation \cite{CommunityNotesDownloadData}.

\begin{table}[ht]
\centering
\small
\setlength{\tabcolsep}{4pt}
\begin{tabularx}{\linewidth}{l X X}
\hline
\textbf{Dataframe} & \textbf{Relevant Columns} & \textbf{Description} \\
\hline
\texttt{ratings\_df} &
\texttt{noteId, raterParticipantId, createdAtMillis, helpfulnessLevel} &
Record of each (user, note) rating pair and the timestamp \\
\texttt{history\_df} &
\texttt{noteId, createdAtMillis, currentStatus} &
Record of each note and what its most recent note status was (i.e. whether it has reached Helpful or Not Helpful status) \\
\texttt{note\_df} &
\texttt{noteId, raterParticipantId, createdAtMillis, tweetId, summary} &
Contains metadata about each note \\
\texttt{note\_factor\_df} &
\texttt{noteId, week\_dt, noteIntercept, noteFactor1} &
Contains weekly computations for note intercept and note factor (one for 2022 version, one for 2025 version of the code) \\
\texttt{rater\_factor\_df} &
\texttt{raterParticipantId, week\_dt, raterIntercept, raterFactor1} &
Contains weekly computations for rater intercept and rater factor (one for 2022 version, one for 2025 version of the code) \\
\hline
\end{tabularx}
\caption{Dataframes used in our analysis.}
\label{tab:dataframes}
\end{table}

\subsection{Reconstructing Weekly Latent Factors}
The public release does not include the latent parameters used internally by the platform's aggregation system. As a result, all note and rater intercepts and factors used in our analysis are reconstructed from the ratings history rather than observed directly. 
In order to recover weekly estimates for rater and note intercepts and factors, we run X's matrix factorization algorithm. For each week $w$ from June 1, 2022 to May 31, 2023, we run the matrix factorization algorithm on all ratings up to and including ratings from week $w$. Since matrix factorization can only recover the rater and note factors up to a global scaling and sign, the algorithm checks the sign distribution of factors and ensures that the majority always has a negative sign. Thus, the sign meaning stays consistent throughout the weeks. We run both the version from Dec. 2022 and the version from May 2025 \cite{commnotes2022code}. Both implementations largely solve the biased matrix factorization problem presented in the main text using stochastic gradient descent with $L^2$ regularization, and factor normalization to ensure consistent interpretation of factors across weeks.

The 2022 version is a straightforward implementation of the least-squares MF optimization problem with single-round optimization and basic convergence criteria. The 2025 version of the code includes many enhancements (multi-round reputation filtering, harassment detection, uncertainty quantification, and abuse mitigation), however, our implementation utilizes only the stable initialization improvement from the modern codebase. Specifically, we run the \texttt{run\_single\_round\_mf} function, which implements stable initialization using a designated modeling group to prevent factor sign drift across time.

\subsection{Policy Timing, and User Cohorts}

We use Oct.\ 1, 2022 as the analysis cutoff date. This is a conservative operational cutoff following the period when the Rating Impact eligibility rules (announced in September 2022) begin to take effect in the observed data. Dates before Oct.\ 1 are referred to as \emph{pre-rollout}, and dates on or after Oct.\ 1 as \emph{post-rollout}.

Several analyses distinguish behavioral adaptation from compositional change due to entry of new raters. We define \emph{early users} to be raters who were active on Community Notes before Oct.\ 1, 2022, and \emph{new users} to be raters who first became active between Oct.\ 1, 2022 and Jan.\ 1, 2023. 

Finally, note and rater factors should be interpreted as relative positions within a weekly latent scale estimated from all observed user-note interactions on the platform. In particular, note factors are equilibrium objects: they reflect not only how notes are evaluated, but also which notes are written and which notes receive enough ratings to be assigned a factor. For this reason, our empirical comparisons focus on within-pipeline temporal changes, cohort differences, and discontinuities at the rollout boundary, rather than on absolute comparisons across different estimation procedures.

\section{Additional Empirical Results}\label{sec:empiric}

\subsection{Robustness Checks for Minority Behavior Shift}
In this section, we provide several sensitivity tests for the evidence on changes in minority behavior. Recall that the platform normalizes users with negative latent factor to be the majority. The main empirical finding was that, following the introduction of Rating Impact, minority-aligned contributors moved closer to the majority in the platform's latent-factor space, and the predictive role of user-note alignment for Helpful ratings weakened. Here we show that this pattern is robust across alternative visualizations of the factor distributions, a permutation-based comparison of factor shifts for early and new users, a regression discontinuity design for distributional shape, and additional predictive specifications based on the user-note dot product.

\subsubsection{Latent Factor Distribution Shift}
Recall from the main text that we define early users to be the cohort of users were who active on Community Notes before the rollout date of Oct. 1, 2022, and new users to be the cohort of users who became active between Oct. 1, 2022 and Jan. 1, 2023. In Figures \ref{fig:early-new-user-density-plots}, \ref{fig:factor-shift-histogram}, and \ref{fig:factor-density-violin-plot}, we give additional visualizations for the distribution shift for early users compared with new users between Oct. 2022 and Jan. 2023. All figures give observational evidence that users who were affected by the Rating Impact policy change their behavior, with their factors aligning more with the majority over time.

% \subsection*{Latent Factor Distribution Shift}
\subsubsection{RDD Tests for Bimodality}
As an additional robustness test the latent factor distribution shift among note and rater factors, we compute the bimodality coefficient of the distributions over time and run a regression discontinuity design. For each week $t$, we compute the \emph{bimodality coefficient} (BC) of the empirical distribution of latent factors, defined as
\begin{equation}
\mathrm{BC}_t \;=\; \frac{\mathrm{skewness}_t^2 + 1}{\mathrm{kurtosis}_t},
\end{equation}
where skewness and kurtosis are computed from the stimated factors in week $t$.
The bimodality coefficient is scale-invariant and increases with the prominence of multiple modes; for reference, unimodal symmetric distributions (e.g., Gaussian) have $\mathrm{BC} \approx 1/3$, while bimodal distributions yield larger values.

We compute $\mathrm{BC}_t$ separately for (i) rater factors $\{f_u\}$ and (ii) note factors $\{f_n\}$ using weekly snapshots of the matrix factorization estimates reconstructed from the public data. The estimation pipeline, normalization, and regularization are held fixed across time, so temporal changes in $\mathrm{BC}_t$ reflect changes in the empirical distribution rather than rescaling artifacts. 
% In Figures \ref{fig:bc-rater} and \ref{fig:bc-note}, we show the bimodality coefficients for rater and note factors over time.

\paragraph{Regression discontinuity design.}
We estimated a sharp RDD to test whether the October 1, 2022 intervention produced a
discontinuous shift in weekly bimodality coefficients. The running variable $R_t$ is defined as the signed number of days between the Monday of week $t$ and the cutoff date, so $R_t = 0$ corresponds to the first post-intervention week. We use the full available time series as the estimation bandwidth ($9$ weeks pre-, $9$ weeks post-intervention).

We estimated a local linear model on each side of the cutoff:
\begin{equation}
    BC_t = \beta_0 + \beta_1 R_t + \beta_2 \cdot \mathbf{1}[R_t \geq 0]
    + \beta_3 \cdot \left(R_t \times \mathbf{1}[R_t \geq 0]\right) + \varepsilon_t
\end{equation}
where $\beta_2$ identifies the discontinuous jump at the threshold and $\beta_3$ allows the post-intervention slope to differ from the pre-intervention slope. The model was estimated separately for the rater-level and note-level bimodality series via OLS with HC3 heteroskedasticity-robust standard errors.

\paragraph{Results and interpretation.}
For rater factors, we observe a statistically significant negative discontinuity in the bimodality coefficient at the cutoff (Figure~\ref{fig:rdd-bc-rater}), indicating an abrupt shift toward a more unimodal distribution following the introduction of Rating Impact. This pattern is consistent with minority-aligned raters moving closer to the center of the latent spectrum or crossing alignment toward the majority group.

For note factors, we also observe a significant decline in the bimodality coefficient at the cutoff (Figure~\ref{fig:rdd-bc-note}), though the subsequent time trend differs from that of raters.
The regression discontinuity design is used to identify the \emph{local} effect of the Rating Impact rollout on the
distributional shape of latent factors, not to characterize longer-run dynamics. While the post-cutoff time path of note factors differs from that of
rater factors, this divergence is expected and does not affect the interpretation of the discontinuity. Rater factors
represent latent traits of a largely fixed population of users and therefore evolve primarily through behavioral
adaptation. In contrast, note factors are equilibrium objects shaped by endogenous entry:
which notes are written, and which notes receive sufficient evaluations to get assigend a factor all depend
on post-rollout incentives. These selection forces can alter higher-order moments of the note-factor distribution over
time, even when the immediate response to the policy change is a reduction in bimodality. Importantly, our inference
relies on the direction and significance of the discontinuity at the rollout boundary, which is common to both rater
and note factors, and is consistent with strategic conformity reducing the salience of minority-aligned positions.

\subsubsection{Additional Alignment Tests}
The main text used Spearman's correlation between the rater--note dot product \(f_u g_n\) and Helpful ratings as a nonparametric measure of the predictive role of user--note alignment. Here we report two additional robustness checks on the same question.

\paragraph{Logistic Regression}
For each rating between Aug. 1, 2022 and Jan. 1, 2023, we compute the dot product between the user factor and note factor at the time of rating. For each period before and after the Oct. 1, 2022 rollout, we regress helpfulness ratings on the rater-note dot product using logistic regression, and take the difference in coefficients (post minus pre) as our test statistic measuring change in predictiveness. The logistic regression coefficient declined from $15.696$ to $4.427$, a change of $-11.269$. We conduct a permutation test with $1,000$ iterations, randomly reassigning the ratings to pre/post groups, while preserving the original group sizes. The p-value of $0.001$ indicates that the observed decline is statistically significant at the $0.05$ level, consistent with the Spearman correlation results reported in the main text. The test statistic distribution is shown in Figure \ref{fig:permutation-logistic}.

Note that the logistic regression coefficient is sensitive to the scale of the dot product and is therefore less robust than the Spearman correlation as a test statistic; we include it here for completeness.

\paragraph{DiD for Note Helpfulness} 
Second, we estimate a difference-in-differences (DiD) framework on the same dataset to test whether the rollout of Rating Impact affects the predictiveness of rater-note factor for note helpfulness ratings. We regress note helpfulness using the following:
\[
    r_{un} = \alpha + \beta(f_u g_n) + \gamma\text{Post} + \delta(\text{Post}\times f_u g_n) + \epsilon_{un},
\]
where $f_u, g_n$ are the rater and note factors, and $\text{Post}$ is an indicator that is $1$ after Oct. 1, 2022. The coefficient $\beta$ captures the baseline predictiveness of the rater-note factor prior to the intervention, $\gamma$ captures any level shift in helpfulness ratings post-rollout, and $\delta$ is the DiD estimator of interest. Standard errors are heteroskedasticity-robust (HC3). Results are shown below.

\begin{table}[h]
\centering
\begin{tabular}{lrrrrrr}
\hline
 & Parameter & Std.\ Err. & $z$ & $p$-value & Lower CI & Upper CI \\
\hline
Intercept                    &  0.4937 & 0.039 &  12.649 & $<$0.001 &  0.417 &  0.570 \\
$f_u g_n$                    &  1.3821 & 0.068 &  20.323 & $<$0.001 &  1.249 &  1.515 \\
Post                         & $-$0.0767 & 0.047 & $-$1.636 &    0.102 & $-$0.169 &  0.015 \\
Post$\times f_u g_n$         & $-$0.5344 & 0.100 & $-$5.345 & $<$0.001 & $-$0.730 & $-$0.338 \\
\hline
% \multicolumn{7}{l}{$R^2 = 0.350$, Adj.\ $R^2 = 0.344$, $N = 329$} \\
% \hline
\end{tabular}
\caption{DiD estimates for the predictiveness of the rater-note factor on note helpfulness ratings, corresponding to the specification in \eqref{eq:did_share}. The dependent variable is $r_{un}$, is the helpfulness rating. $f_u g_n$ is the rater-note dot product and Post is an indicator equal to 1 after the Rating Impact rollout (October 2022). The baseline coefficient on $f_u g_n$ ($\hat{\beta} = 1.382$, $p < 0.001$) indicates a strong pre-intervention relationship between the rater-note factor and helpfulness. The DiD estimator Post$\times f_u g_n$ ($\hat{\delta} = -0.534$, $p < 0.001$) indicates that this predictive relationship weakened significantly following the rollout. Standard errors are heteroskedasticity-robust (HC3).}
\label{tab:did_dot_product}
\end{table}

Taken together with the Spearman and rolling-correlation analyses in the main text, these additional specifications reinforce the same conclusion: after the introduction of Rating Impact, user--note alignment becomes less predictive of helpfulness ratings among contributors who were active through the policy change.

\begin{figure}[t]
    \centering
    \includegraphics[width=.6\linewidth]{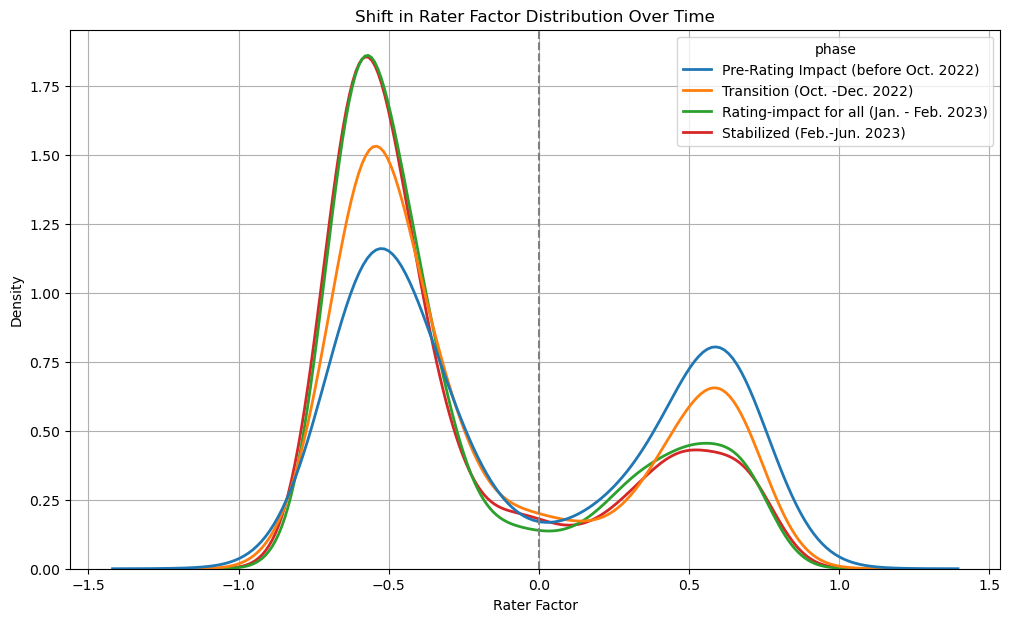}\par\vspace{0.5em}
    \includegraphics[width=.6\linewidth]{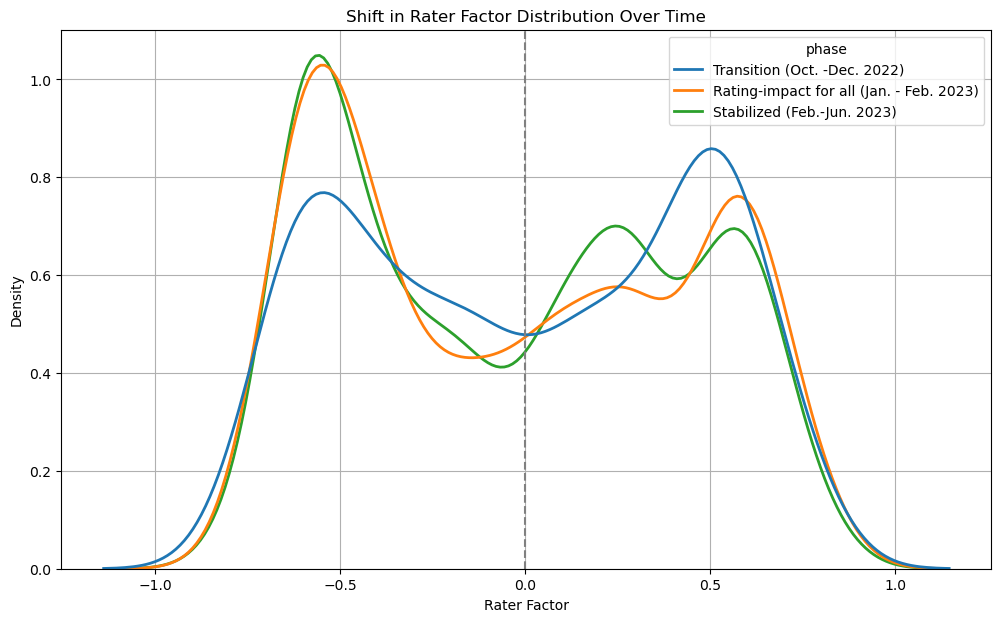}
    \caption{The top figure shows rater factor distribution shift for the early user cohort and the bottom figure shows rater factor distribution shift for the new user cohort using the 2022 version of the matrix factorization code. The minority mode decreases significantly for early users, while remaining more stable for new users. For early users, the proportion of positive factors during the Transition period compared with the Stabilized period decreases from $31.4\%$ to $24.8\%$, a decrease of $6.6$ percentage points. For new users, the proportion of positive factors comparing the same two time periods decreases from $50.3\%$ to $49.5\%$, a decrease of only $0.8$ percentage points.}
    \label{fig:early-new-user-density-plots}
\end{figure}

\begin{figure}[t]
    \centering
    \includegraphics[width=.6\linewidth]{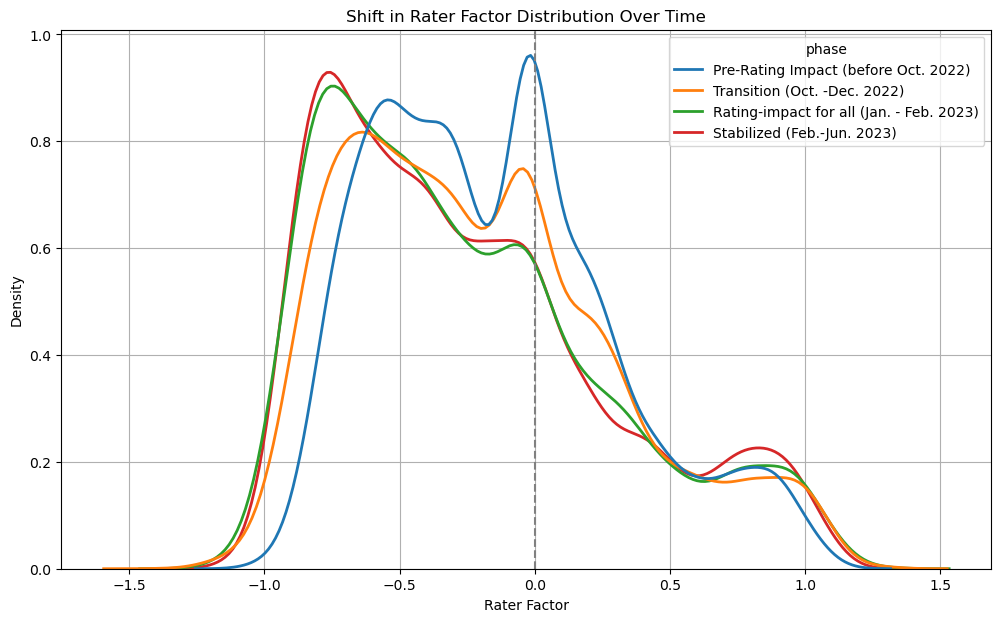}\par\vspace{0.5em}
    \includegraphics[width=.6\linewidth]{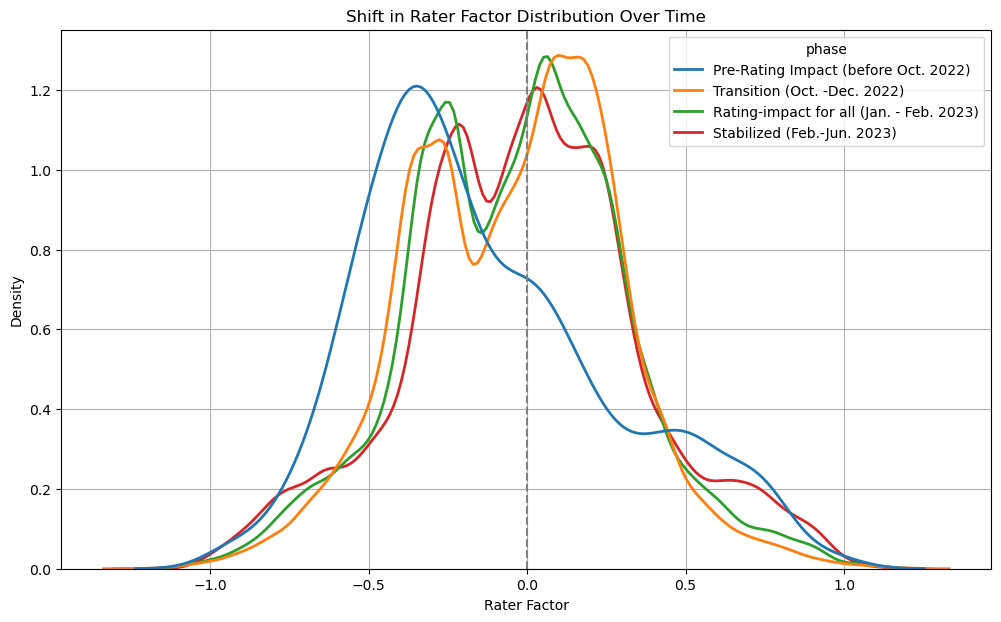}
    \caption{The top figure shows rater factor distribution shift for the early user cohort and the bottom figure shows rater factor distribution shift for the new user cohort using the 2025 version of the matrix factorization code. Here again, the minority mode decreases, while remaining relatively stable for new users. For early users, the proportion of positive factors decreases from $30.8\%$ during the Transition period to $27.8\%$ during the Stabilized period, a decrease of $3.0$ percentage points. For new users the proportion of positive factors actually increases from $48.8\%$ during the Transition period to $49.6\%$ during the Stabilized period, an increase of $0.8$ percentage points.}
    \label{fig:early-new-user-density-plots-2025}
\end{figure}

\begin{figure}
    \centering
    \includegraphics[width=.6\linewidth]{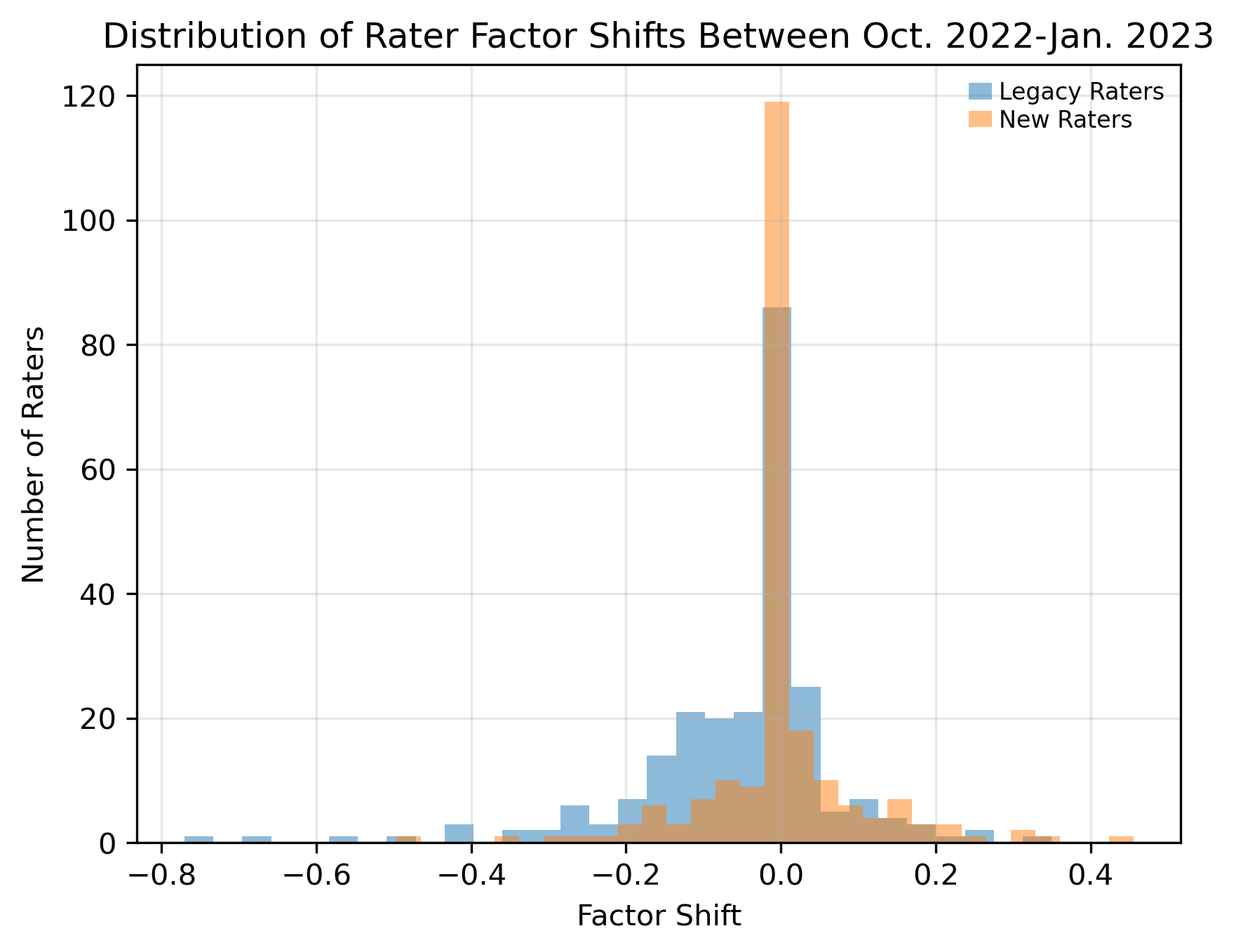}
    \caption{This figure shows a histogram of the difference in rater factor between Oct. 2022 and Jan. 2023 for raters in the early cohort vs. in the new cohort. The mean of the early cohort factor shift is $-0.049$ (95\% bootstrap CI: [-0.067, -0.033]), and the mean of the new cohort factor shift is $0.003$ (95\% bootstrap CI: [-0.010, 0.017]). This shows that early user factors shift more toward the majority than new user factors.}
    \label{fig:factor-shift-histogram}
\end{figure}

\begin{figure}
    \centering
    \includegraphics[width=\linewidth]{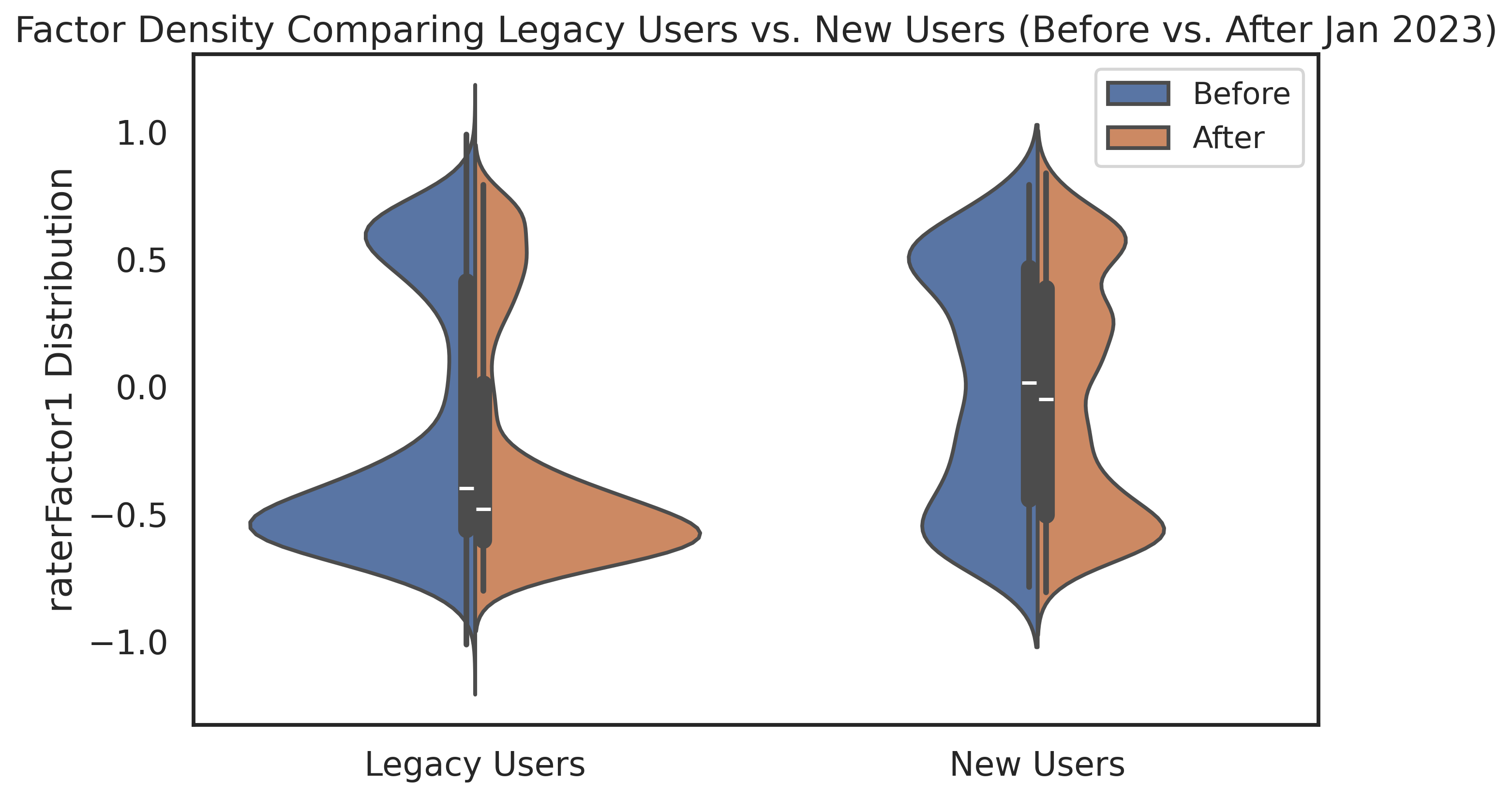}
    \caption{This plot visualizes the difference in rater factor distribution for early users compared to new users before and after Jan. 2023 when the rollout of Rating Impact was complete. Again we see that the distribution of new user factors remains relatively stable before and after, whereas the distribution of early user factors shifts towards the majority.}
    \label{fig:factor-density-violin-plot}
\end{figure}

\begin{figure}
    \centering
    \includegraphics[width=.8\linewidth]{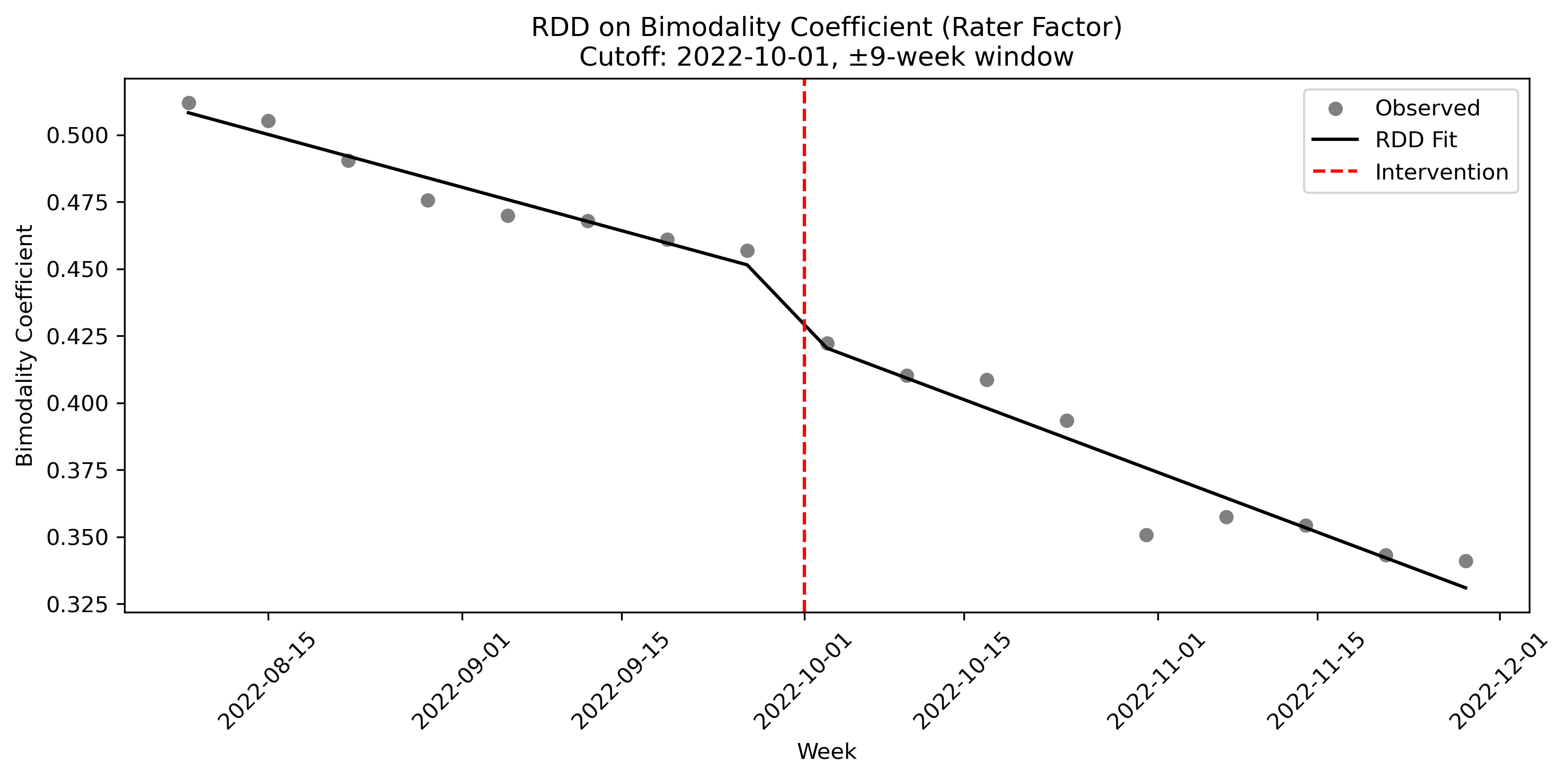}
    \caption{This figure shows the RDD analysis on the bimodality coefficient measured at weekly intervals for rater factors with the cutoff date of Oct. 1, 2022. There is a $0.022$ decline in the bimodality coefficient at the cutoff ($p=0.003$), implying that post-cutoff the distribution becomes more unimodal.}
    \label{fig:rdd-bc-rater}
\end{figure}

\begin{figure}
    \centering
    \includegraphics[width=.8\linewidth]{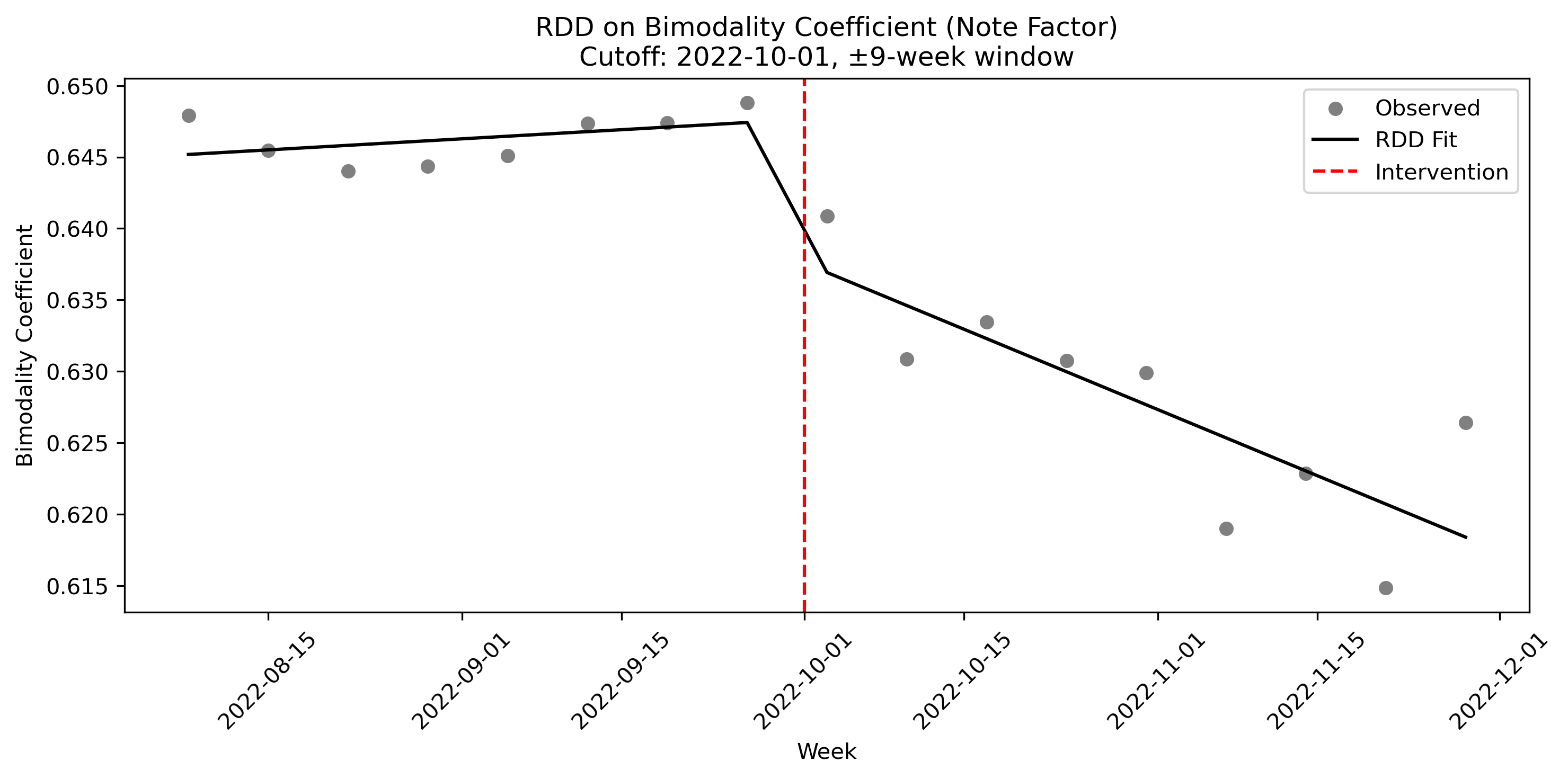}
    \caption{This figure shows the RDD analysis on the bimodality coefficient measured at weekly intervals for note factors with the cutoff date of Oct. 1, 2022. There is a $0.010$ decline in the bimodality coefficient at the cutoff ($p = 0.017$), implying that at the cutoff the distribution jumps to a more unimodal distribution.}
    \label{fig:rdd-bc-note}
\end{figure}

\begin{figure}
    \centering
    \includegraphics[width=.6\linewidth]{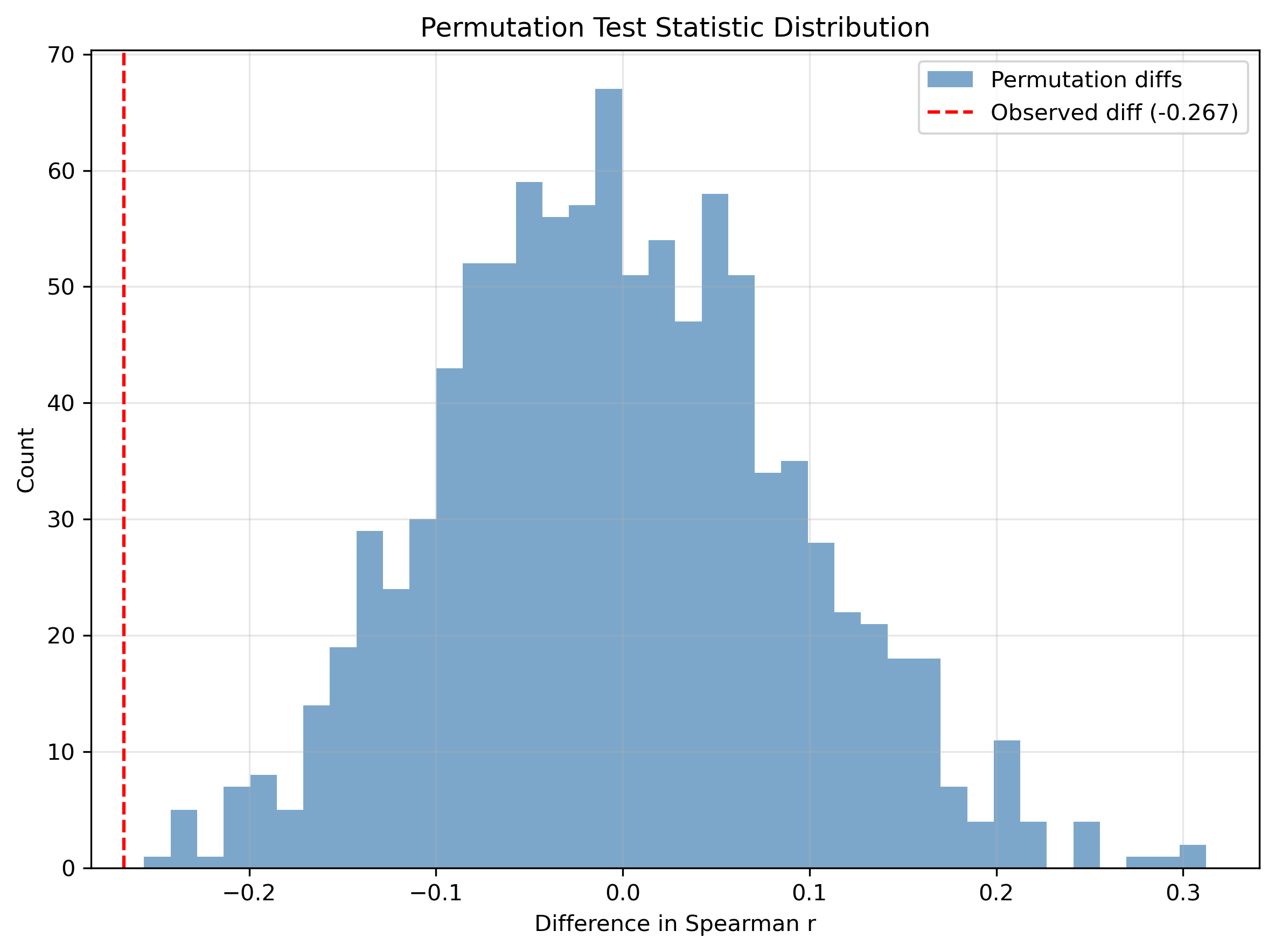}
    \caption{Permutation test test statistic distribution for the difference in Spearman correlation (post minus pre intervention) between rater-note factor dot product and helpfulness ratings among early users. The red dashed line marks the observed difference of $-0.267$. Very few of the $1,000$ permuted differences fall near the observed value, yielding $p = 0.004$, indicating that the decline in predictiveness after October 1, 2022 is highly unlikely to have occurred by chance.}
    \label{fig:permutation-spearman}
\end{figure}

\begin{figure}
    \centering
    \includegraphics[width=.6\linewidth]{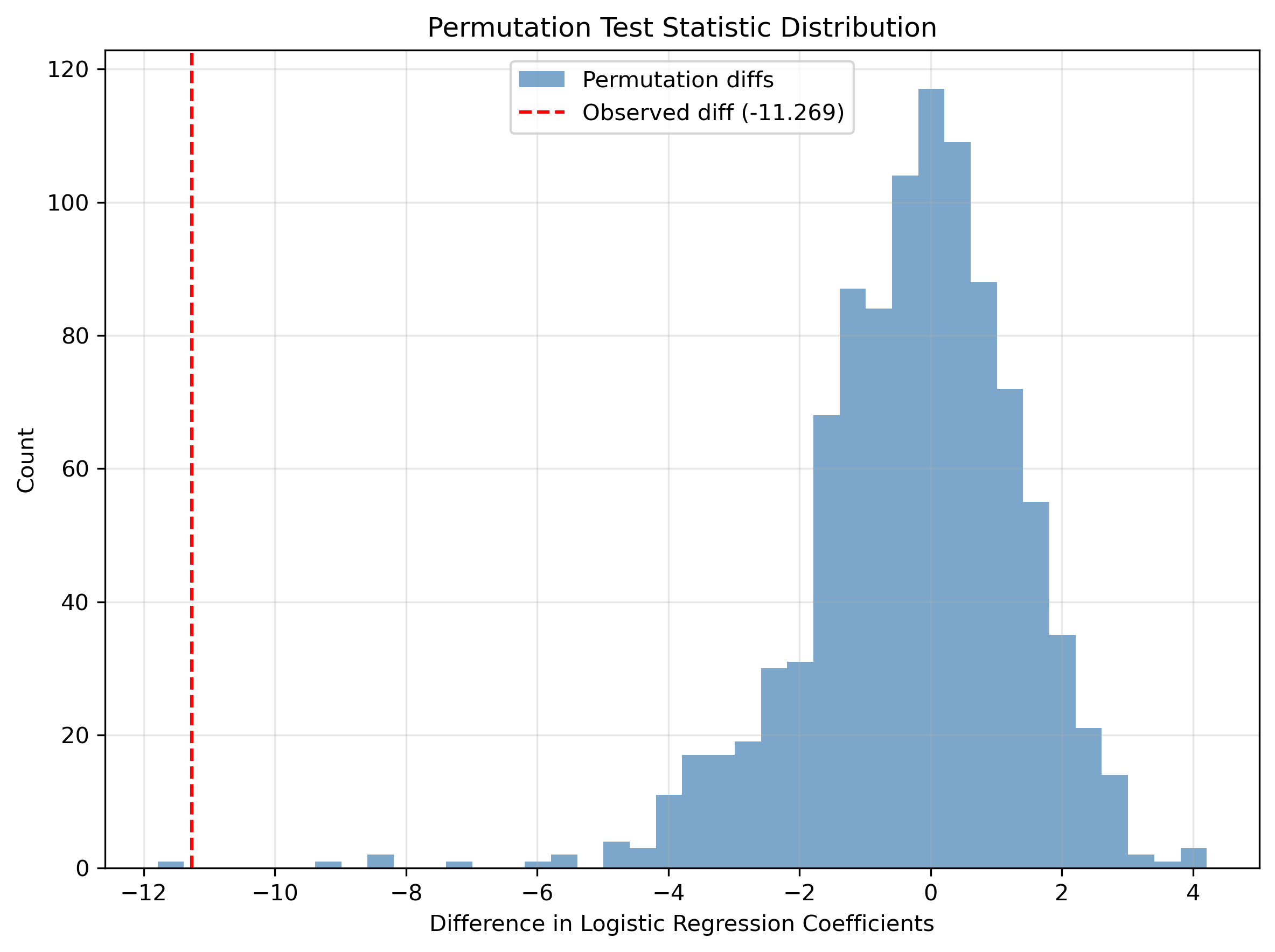}
    \caption{Permutation test test statistic distribution for the difference in logistic regression coefficients (post minus pre intervention). The red dashed line marks the observed difference of $-11.269$. Very few of the $10,000$ permuted differences fall near the observed value, yielding $p = 0.001$, indicating that the decline in predictiveness after October 1, 2022 is highly unlikely to have occurred by chance.}
    \label{fig:permutation-logistic}
\end{figure}

\subsection{Controversial Content and Participation}

This section provides additional detail and sensitivity analyses for the controversial-content result in the main text. The main pattern is that, following the rollout of Rating Impact, notes on controversial content are less likely to attain Helpful status than notes on non-controversial content. We document this pattern using two complementary definitions of controversy. The first is topic-based and uses note summaries to classify notes into broad content areas before labeling those areas as controversial or non-controversial. The second is factor-based and uses the magnitude of the estimated note factor as a model-based measure of polarization. We also examine whether the decline in controversial-note visibility is accompanied by changes in contributor-level engagement with controversial content.

\subsubsection{Topic Assignment and Controversy Definitions}

Next we describe how we assign topic labels to notes and how those labels are used
to classify content as controversial or non-controversial. 
We first define the primary topic
assignment procedure used throughout the main text, which closely follows X’s public
implementation with expanded coverage. We then present alternative large-language-model
(LLM)-based topic assignments used as alternative classification procedures to assess sensitivity to the
topic-classification mechanism.

\paragraph{Primary Topic Assignment (Bag-of-Words Classifier)}
Our primary topic assignment builds on the topic-modeling code used by the
platform implementation, which combines seed-term matching with a supervised
bag-of-words classifier. In the version of X’s code used for this paper,
each topic is defined initially by a small set of seed terms. Preliminary
topic assignment is based on exact and fuzzy matches to these seed terms,
after which a multi-class logistic regression classifier is trained to
expand the set of in-topic notes. 

In the version of X’s code used for this paper, the native topic inventory is limited to
\emph{Ukraine Conflict}, \emph{Gaza Conflict}, \emph{Messi--Ronaldo}, and \emph{Scams}. To
support analyses requiring broader topical coverage, we expand this inventory by introducing
additional candidate topics and associated seed terms. These additional topics and seed terms
are used \emph{only} to augment the training
data for the bag-of-words classifier. The classifier architecture, feature representation,
and regularization follow X’s implementation, with one modification: we lower the minimum
balanced-accuracy threshold for topic inclusion to $0.01$ in order to retain topics with
sparse coverage.
Table~\ref{tab:topic-seed-terms} lists the resulting topic set and seed terms. All topic
labels used in the main-text analyses are produced by this retrained bag-of-words classifier.

Independently of the topic-assignment procedure, we label topics \emph{a priori} as
controversial or non-controversial based on domain knowledge and prior literature.
Table~\ref{tab:topic-controversy} reports the full list of topics and their controversy
classification.

\paragraph{LLM-Based Topic Assignment (Alternative Classification Procedures)}
To assess whether our results depend on the specific topic-assignment mechanism, we conduct
robustness checks using two alternative LLM-based classification procedures. These procedures
are not used in the main analyses and serve only to evaluate sensitivity to the choice of
topic classifier.

Both LLM-based procedures use an identical, fixed inventory of topic labels:
\emph{Ukraine Conflict, Gaza Conflict, Messi Ronaldo, Sports NFL, Sports NBA, Movies TV,
Education, Food Nutrition, Space Astronomy, Health, COVID-19, Climate Environment,
Weather Disasters, Artificial Intelligence, Tech Companies, US Politics, Crime Legal,
Economy Finance, Scams, Other}. In both cases, the input text is the note-level
\texttt{summary} field, and each note is assigned \emph{exactly one} topic label.
Notes with missing summaries are excluded.

\textbf{Approach 1 (Prompted LLM Classification):}
Our first procedure uses a prompted instruction-tuned LLM accessed through Snowflake Cortex’s text completion interface. For each note summary, we supplied an explicit natural-language prompt that enumerates the full label set and enforces a single-label classification objective with a structured output format. Decoding was deterministic (temperature set to zero) to ensure reproducibility across runs. The exact prompt used was:
\begin{quote}
\small\ttfamily
You are a classifier. Choose exactly ONE topic label from: \\
\relax [Ukraine Conflict, Gaza Conflict, Messi Ronaldo, Sports NFL, Sports NBA, Movies TV, \\
Education, Food Nutrition, Space Astronomy, Health, COVID-19, Climate Environment, \\
Weather Disasters, Artificial Intelligence, Tech Companies, US Politics, Crime Legal, \\
Economy Finance, Scams, Other]. \\
Return ONLY valid JSON like \{"topic":"<label>"\}. \\
Summary: <raw>\{\emph{summary}\}</raw>
\end{quote}
The model’s response was parsed to extract the value of the \texttt{topic} field. Outputs that failed to parse as valid JSON or that returned a label outside the admissible set were treated as missing for this procedure (these cases were rare in practice).

\textbf{Approach 2 (Managed LLM Classification):}
Our second procedure uses Snowflake Cortex’s managed multi-class classification function. Given a note summary and the same fixed list of candidate topics, this service returns a ranked list of predicted labels. We assigned the highest-ranked label as the note’s topic, using the following task description to enforce a single-label objective:
\begin{quote}
\small\ttfamily
Classify the topic of a short summary. Choose exactly one topic.
\end{quote}
As with the prompted approach, notes with missing summaries or malformed outputs were excluded for this classification. X's bag-of-words classifier excludes 64.22\% of notes, LLM Approach 1 excludes 28.79\% of notes, and LLM Approach 2 excludes 31.11\% of notes from having a topic. Excluding all notes that are excluded in any one of the models above, the three approaches agree on 10.89\% of the notes. Table \ref{tab:pairwise_agreement_nonmissing} shows pairwise comparisons for the agreement of the different approaches.

\begin{table}[h]
\centering
\renewcommand{\arraystretch}{1.25}
\setlength{\tabcolsep}{10pt}
\begin{tabular}{lccc}
\hline
 & \textbf{Bag-of-Words} & \textbf{Prompted LLM} & \textbf{Managed LLM} \\
\hline
\textbf{Bag-of-Words} & -- & 11.20\% & 10.92\% \\
\textbf{Prompted LLM} & 11.20\% & -- & 74.57\% \\
\textbf{Managed LLM}  & 10.92\% & 74.57\% & -- \\
\hline
\end{tabular}
\caption{Pairwise topic agreement rates between different topic modeling approaches. Percentages are calculated over denominator computed as notes that are included for both approaches being compared.}
\label{tab:pairwise_agreement_nonmissing}
\end{table}

The three topic-assignment procedures differ in coverage and agreement. The bag-of-words classifier excludes 64.22\% of notes outright, reflecting some limitations of X's native topic-modeling implementation, which degrades in precision when extended to our broader topic inventory. This contributes to the lower agreement between the bag-of-words and LLM-based classifiers, relative to the agreement between the LLM-based approaches.

These agreement rates indicate that the bag-of-words and LLM-based procedures are not interchangeable at the topic level. We therefore interpret the LLM-based analyses as sensitivity checks to alternative classification procedures.

\subsubsection{Topic-Based Controversy Results to Alternative Classifiers}
The main-text analyses use topic labels only to stratify notes into \emph{controversial} versus \emph{non-controversial} bins, rather than to make topic-specific claims. To assess sensitivity to the topic-classification procedure, we repeat the analysis using each classifier’s assigned labels independently  (see Figures \ref{fig:topic-prompted}-\ref{fig:topic-managed}). These procedures should be interpreted as alternative operationalizations of controversy rather than interchangeable measurements of the same latent topic labels. Reassuringly, the qualitative pattern is broadly similar across procedures: after the policy change, controversial content is relatively less likely to attain Helpful status than non-controversial content. We therefore interpret the LLM-based results as sensitivity analyses to alternative classification procedures. 

\subsubsection{Topic-Based Difference-in-Differences}
Here we give the details for the DiD in Table 1. of the main text. We construct a weekly panel dataframe for each week $w$ and topic group $g \in \{\text{Controversial, Non-Controversial}\}$. Define $n_{g, w}$ to be the number of unique tweets that receive their first note in week $w$ and $k_{g, w}$ to be the number of unique tweets whose first observed note attains Helpful status in week $w$. We estimate a weekly helpfulness proportion using Jeffreys shrinkage:
\[
    p_{g, w} = \frac{k_{g, w} + 0.5}{n_{g, w} + 1}.
\]
We restrict our analysis to weeks where $n_{g, w} > 5$ for all $g$ (the cutoff 5 is chosen by X).

To summarize the contrast between topic groups, we define the weekly helpfulness gap
\[
    d_w = p_{\text{Controversial}, w} - p_{\text{Non-Controversial}, w}
\]
weighted by $\eta_w = \min_g n_{g, w}$. We estimate short-term and medium-term effects by comparing the mean of $d_w$ for pre-defined time periods before the rollout date (Oct. 1, 2022). The pre weeks are the 12 weeks before Oct. 1, 2022, the post short term weeks are the 12 weeks after Oct. 1, 2022, and the post medium term weeks are weeks 13-26 after Oct. 1, 2022. For each set of post weeks, we estimate
\[
    d_w = \alpha + \beta \cdot \text{Post}_w + \epsilon_w.
\]
We compute HAC standard errors, and report $\hat \beta$, HAC standard error, $95\%$ confidence intervals, and $p$-values for both the short term and medium term post windows.

This specification yields the topic-based difference-in-differences estimates reported in the main text. A positive coefficient \(\beta\) indicates that, after the rollout, non-controversial notes become more likely to attain Helpful status relative to controversial notes. The resulting estimates and confidence intervals are reported in Table~1 of the main text.

\subsubsection{User-Level Engagement Analysis}
The note-level DiD above is descriptive: it shows that controversial notes become relatively less likely to attain Helpful status after the rollout, but it does not by itself distinguish between changes in exposure, changes in rating behavior, and changes in which users choose to rate which content. To address this, we study contributor-level engagement with controversial content.

% \paragraph{Selective engagement with controversial content.}
For each user $u$ and calendar week $w$, let $R_{uw}$ denote the total number of ratings
submitted by user $u$ in week $w$, and let $R^{C}_{uw}$ denote the number of those ratings
assigned to controversial notes. We define
\[
S_{uw} = \frac{R^{C}_{uw}}{R_{uw}},
\]
the share of user $u$’s rating activity in week $w$ devoted to controversial content, and restrict attention to weeks with $R_{uw}>0$.
We estimate the following difference-in-differences specification:
\begin{equation}
\label{eq:did_share}
S_{uw}
=
\alpha_u
+
\gamma_w
+
\beta\,(\text{Minority}_u \times \text{Post}_w)
+
\varepsilon_{uw},
\end{equation}
where $\alpha_u$ are user fixed effects, $\gamma_w$ are week fixed effects,
$\text{Minority}_u$ is an indicator for users classified as minority-aligned based on
pre-rollout behavior, and $\text{Post}_w$ indicates weeks on or after the Rating Impact
rollout.

% \paragraph{Overall engagement.}
To assess whether minority-aligned users reduce rating activity more broadly, we estimate
a parallel difference-in-differences specification using total weekly engagement as the
outcome:
\begin{equation}
\label{eq:did_total}
\log(1+R_{uw})
=
\alpha_u
+
\gamma_w
+
\delta\,(\text{Minority}_u \times \text{Post}_w)
+
\eta_{uw},
\end{equation}
where variables are defined as above.
The log transformation accommodates the skewness of rating counts and allows coefficients
to be interpreted as approximate percentage changes.

Comparing estimates from ~\eqref{eq:did_share} and \eqref{eq:did_total} allows us to
distinguish selective disengagement from overall reductions in activity: a decline in
$S_{uw}$ conditional on $R_{uw}>0$ indicates avoidance of controversial content even among
users who remain active.

% \paragraph{Method.}
For each note, we add a binary indicator for whether the note is controversial or non-controversial based on its topic, as defined by Table \ref{tab:topic-controversy}. We generate topics using LLMs, with the method described above. We consider the cohort of users who were active before the initial roll out of Rating Impact. In this cohort, there are 1052 users. Our cut-off date for treatment is Oct. 1, 2022 when the Rating Impact policy was announced. We take the users' pre-treatment factor to be their most recently computed factor before Oct. 1, 2022. There are 354 users with positive factor and 698 users with negative factor. Those with positive factor ($>0.1$ to avoid borderline cases at the boundary of $0$) are labeled to be in the minority group using a binary indicator variable. We then run the two DiD studies documented by ~\eqref{eq:did_share} and \eqref{eq:did_total}. We consider a 12-week window before and after Oct. 1, 2022. Results are shown in the tables below.

\begin{table}[h]
\centering
\begin{tabular}{lrrrrrr}
\hline
& Parameter & Std.\ Err. & T-stat & P-value & Lower CI & Upper CI \\
\hline
Intercept & 0.6125 & 0.0063 & 96.730 & 0.0000 & 0.6001 & 0.6249 \\
treat     & -0.0753 & 0.0343 & -2.1926 & 0.0284 & -0.1427 & -0.0080 \\
\hline
\end{tabular}
\caption{This table shows the DiD given by \eqref{eq:did_share} for the proportion of controversial notes a user rates in a week. The DiD estimate of -0.0753 indicates that, post Oct. 1, 2022, the proportion of controversial notes rated by minority users decreased by about 7.5 percentage points relative to the change observed for non-minority users. This effect is statistically significant at the 0.05 level (p = 0.0284).}
\end{table}

\begin{table}[h]
\centering
\begin{tabular}{lrrrrrr}
\hline
& Parameter & Std.\ Err. & T-stat & P-value & Lower CI & Upper CI \\
\hline
Intercept & 1.4903 & 0.0129 & 115.37 & 0.0000 & 1.4650 & 1.5157 \\
treat     & 0.1346 & 0.0701 & 1.9216 & 0.0547 & -0.0027 & 0.2720 \\
\hline
\end{tabular}
\caption{This table shows the DiD given by \eqref{eq:did_total} for the log of the total number of notes that a user rates in a week. The DiD estimate of 0.1346 indicates that, post Oct. 1, 2022, the log of the number of notes rated by minority users increases, however this was not statistically significant at the 0.05 level.}
\end{table}

% \subsection*{Factor-Based Classification}

\subsubsection{Factor-Based Controversy Robustness}

As discussed in the main text, we complement the topic-based classification with a factor-based measure of controversy that does not rely on text or topic labels. In the latent-factor model used by Community Notes, the note factor $g_n$ captures systematic differences in how a note is evaluated by different groups of raters. Notes with larger absolute factor values therefore correspond to content that elicits more polarized evaluations across the user base.

As a robustness check, we replicate the analyses from the
\emph{Annotations on Controversial Topics} section of the main text using this factor-based definition of controversy. Specifically, we classify a note as \emph{controversial} if the magnitude of its estimated factor $|g_n|$ lies in the top $20$th percentile of the distribution of note factors, and as \emph{non-controversial} otherwise.\footnote{We choose the $20$th percentile for concreteness; results are qualitatively unchanged when using alternative thresholds, e.g. $10$th percentile.} Because note factors are re-estimated over time, we first average the estimated factor values for each note across the pre-weeks (before the intervention) and then form the distribution over these averaged note factors. This averaging reduces noise from week-to-week re-estimation and ensures a stable ranking of notes by ideological extremity.

Using this classification, we re-compute Wilson confidence intervals and estimate a difference-in-differences (DiD) design to compare changes in note helpfulness around the rollout of Rating Impact for controversial versus non-controversial notes. For controversial notes, the share of tweets receiving at least one note rated helpful decreases from $0.143$ (Wilson CI $[0.071, 0.3267]$) pre-rollout to $0.080$ (Wilson CI $[0.059, 0.107]$) post-rollout, a decline of $6.3$ percentage points. In contrast, for non-controversial notes, the helpful share increases from $0.131$ (Wilson CI $[0.068, 0.238]$) to $0.265$ (Wilson CI $[0.232, 0.302]$), an increase of $13.4$ percentage points. The DiD estimate using a symmetric $14$-week window around October~1,~2022 implies that the probability a non-controversial note is rated helpful increased by about $13.8$ percentage points more than the probability a controversial note is rated helpful ($p=0.031$). Results are shown in Table \ref{tab:did_results} and Figure~\ref{fig:factor-controversy-plot}.

\begin{table}[h]
\centering
\begin{tabular}{lcccc}
\toprule
Weeks (Pre, Post) & Estimate & 95\% CI & $p$-value \\
\midrule
(14, 0-14)  & 0.138 & [0.012, 0.263] & 0.031 \\
(14, 15-28) & 0.239 & [0.153, 0.324] & <1e-5 \\
\bottomrule
\end{tabular}
\captionsetup{justification=centering}
\caption{DiD Estimates. \\ Windowed difference-in-differences estimates of the change in the weekly helpfulness rate difference between non-controversial and controversial content after the October~1,~2022 rollout. The outcome is the difference in weekly proportion of tweets receiving a note that is ultimately rated \textit{Helpful} between tweets with controversial vs. non-controversial notes. The gap increases by 13.8 pp in the first 14 weeks post-rollout and by 23.9 pp in weeks 15--28. Both effects are positive and statistically significant.}
\label{tab:did_results}
\end{table}

To account for the possibility that controversial notes receive more polarized ratings and therefore fail to attain helpful status, we additionally examine engagement at the rating level. Specifically, we compare the proportion of all ratings (regardless of final helpfulness status) assigned to controversial versus non-controversial notes before and after the rollout. We run a Fisher's exact test on a $2 \times 2$ contingency table of controversial and non-controversial ratings in the pre- and post-period, excluding ratings on unclassified topics. The odds ratio quantifies the relative shift in the balance between controversial and non-controversial ratings across the two periods; an odds ratio greater than 1 indicates that controversial ratings were relatively more prevalent than non-controversial ratings in the pre-period compared to the post-period. Results are shown in Table \ref{tab:fisher_results} for the three topic modelers.

Taken together, the topic-based and factor-based analyses point in the same direction. The topic-based approach is easier to interpret substantively, while the factor-based approach captures within-topic heterogeneity and avoids dependence on text classification. The agreement between the two strengthens the conclusion that, after the rollout of Rating Impact, controversial content becomes relatively less likely to receive publicly surfaced annotations.

\begin{table}[h]
\centering
\begin{tabular}{lcccc}
\toprule
Topic Modeler & Pre Controversial Rate & Post Controversial Rate & Odds Ratio & $p$-value \\
\midrule
Bag of Words & 0.897 & 0.853 & 1.505 & $<$0.0001 \\
LLM1         & 0.756 & 0.535 & 2.696 & $<$0.0001 \\
LLM2         & 0.783 & 0.495 & 3.672 & $<$0.0001 \\
\bottomrule
\end{tabular}
\captionsetup{justification=centering}
\caption{Fisher's Exact Test Results by Topic Modeler. \\ Odds ratios from a Fisher's exact test on a $2\times 2$ contingency table of controversial and non-controversial ratings in the pre- and post-period over classified topics. An odds ratio greater than 1 indicates that controversial ratings were relatively more prevalent in the pre-period than in the post-period. All three topic modelers show a significant shift toward non-controversial ratings after the rollout.}
\label{tab:fisher_results}
\end{table}

\newpage

\begin{table}

\centering
\begin{tabularx}{\textwidth}{l l X}
\toprule
\textbf{Category} & \textbf{Topic} & \textbf{Seed terms} \\
\midrule
\textbf{Conflict \& Geopolitics} & Ukraine Conflict & \texttt{ukrain}, \texttt{russia}, \texttt{kiev}, \texttt{kyiv}, \texttt{moscow}, \texttt{zelensky}, \texttt{putin} \\
 & Gaza Conflict & \texttt{israel}, \texttt{palestin}, \texttt{gaza}, \texttt{jerusalem} \\
\addlinespace
\textbf{Sports \& Entertainment} & Messi--Ronaldo & \texttt{messi\textbackslash s}, \texttt{ronaldo} \\
 & Sports (NFL) & \texttt{\textbackslash bnfl\textbackslash b}, \texttt{super\textbackslash sbowl}, \texttt{touchdown}, \texttt{quarterback}, \texttt{patrick\textbackslash smahomes}, \texttt{eagles}, \texttt{cowboys}, \texttt{patriots} \\
 & Sports (NBA) & \texttt{\textbackslash bnba\textbackslash b}, \texttt{basketball}, \texttt{lebron}, \texttt{curry}, \texttt{lakers}, \texttt{warriors}, \texttt{playoffs?}, \texttt{finals?}, \texttt{dunk} \\
 & Movies \& TV & \texttt{movie}, \texttt{film}, \texttt{hollywood}, \texttt{box\textbackslash soffice}, \texttt{netflix}, \texttt{disney\textbackslash +?}, \texttt{marvel}, \texttt{season\textbackslash s\textbackslash d+}, \texttt{episode}, \texttt{series} \\
\addlinespace
\textbf{Education, Food \& Space} & Education & \texttt{student\textbackslash sloan}, \texttt{tuition}, \texttt{college}, \texttt{university}, \texttt{k[-]?12}, \texttt{school\textbackslash sboard} \\
 & Food \& Nutrition & \texttt{\textbackslash bdiet\textbackslash b}, \texttt{calorie}, \texttt{vegan}, \texttt{vegetarian}, \texttt{keto}, \texttt{gluten[-\textbackslash s]?free}, \texttt{nutrition}, \texttt{protein}, \texttt{sugar} \\
 & Space \& Astronomy & \texttt{nasa}, \texttt{spacex}, \texttt{rocket}, \texttt{launch}, \texttt{mars}, \texttt{lunar}, \texttt{satellite}, \texttt{iss}, \texttt{telescope} \\
\addlinespace
\textbf{Science, Health \& Environment} & Health, COVID-19 & \texttt{covid}, \texttt{coronavirus}, \texttt{pandemic}, \texttt{vaccine}, \texttt{booster}, \texttt{omicron}, \texttt{mask}, \texttt{cdc}, \texttt{\textbackslash bwho\textbackslash b}, \texttt{pfizer}, \texttt{moderna} \\
 & Climate \& Environment & \texttt{climate\textbackslash schange}, \texttt{global\textbackslash swarm}, \texttt{greenhouse}, \texttt{carbon}, \texttt{emission}, \texttt{renewable\textbackslash senergy}, \texttt{solar}, \texttt{wind\textbackslash s(?:energy|power)}, \texttt{climate\textbackslash scrisis} \\
 & Weather \& Disasters & \texttt{hurricane}, \texttt{tornado}, \texttt{earthquake}, \texttt{wildfire}, \texttt{heatwave}, \texttt{flood}, \texttt{blizzard}, \texttt{typhoon}, \texttt{storm} \\
\addlinespace
\textbf{Technology \& AI} & Artificial Intelligence & \texttt{\textbackslash bai\textbackslash b}, \texttt{artificial\textbackslash sintelligence}, \texttt{machine\textbackslash slearning}, \texttt{chatgpt}, \texttt{openai}, \texttt{gpt[-\_]?4}, \texttt{deepmind}, \texttt{neural\textbackslash snetwork} \\
 & Tech Companies & \texttt{apple}, \texttt{iphone}, \texttt{google}, \texttt{android}, \texttt{microsoft}, \texttt{windows}, \texttt{amazon}, \texttt{aws}, \texttt{meta}, \texttt{facebook}, \texttt{instagram}, \texttt{threads} \\
\addlinespace
\textbf{Society \& Governance} & U.S. Politics & \texttt{biden}, \texttt{trump}, \texttt{democrat}, \texttt{republican}, \texttt{gop}, \texttt{congress}, \texttt{senate}, \texttt{house\textbackslash sof\textbackslash srepresentatives}, \texttt{white\textbackslash shouse}, \texttt{supreme\textbackslash scourt}, \texttt{2024\textbackslash selection} \\
 & Crime \& Legal & \texttt{indict}, \texttt{lawsuit}, \texttt{legal}, \texttt{court}, \texttt{judge}, \texttt{jury}, \texttt{police}, \texttt{arrest}, \texttt{felony}, \texttt{trial} \\
\addlinespace
\textbf{Economics \& Finance} & Economy \& Finance & \texttt{inflation}, \texttt{recession}, \texttt{interest\textbackslash srate}, \texttt{stock\textbackslash smarket}, \texttt{dow\textbackslash sjones}, \texttt{nasdaq}, \texttt{gdp}, \texttt{unemployment}, \texttt{cpi} \\
\addlinespace
\textbf{Scams \& Platform Integrity} & Scams & \texttt{scam}, \texttt{undisclosed\textbackslash sad}, \texttt{terms\textbackslash sof\textbackslash sservice}, \texttt{help\textbackslash .x\textbackslash .com}, \texttt{x\textbackslash .com/tos}, \texttt{engagement\textbackslash sfarm}, \texttt{spam}, \texttt{gambling}, \texttt{apostas}, \texttt{apuestas}, \texttt{dropship}, \texttt{drop\textbackslash sship}, \texttt{promotion} \\
\bottomrule
\end{tabularx}
\caption{Our expanded set of topics and seed terms. The regex syntax follows Python \texttt{re} conventions.}
\label{tab:topic-seed-terms}
\end{table}

\begin{table}
\centering
\begin{tabularx}{\textwidth}{l X}
\toprule
\textbf{Category} & \textbf{Topics} \\
\midrule
Controversial &
\texttt{USPolitics}, \texttt{UkraineConflict}, \texttt{GazaConflict}, \texttt{CrimeLegal}, \texttt{HealthCovid}, \texttt{EconomyFinance}, \texttt{MessiRonaldo}, \texttt{politics}, \texttt{economy}, \texttt{health} \\
\addlinespace
Non-controversial &
\texttt{SpaceAstronomy}, \texttt{ClimateEnvironment}, \texttt{EntertainmentMoviesTV}, \texttt{WeatherDisasters}, \texttt{TechCompanies}, \texttt{SportsNFL}, \texttt{SportsNBA}, \texttt{FoodNutrition}, \texttt{Education}, \texttt{Scams}, \texttt{ArtificialIntelligence}, \texttt{science}, \texttt{other} \\
\bottomrule
\end{tabularx}
\caption{Manual categorization of topics as controversial or non-controversial.}
\label{tab:topic-controversy}
\end{table}

\begin{figure}
    \centering
    \includegraphics[width=.6\linewidth]{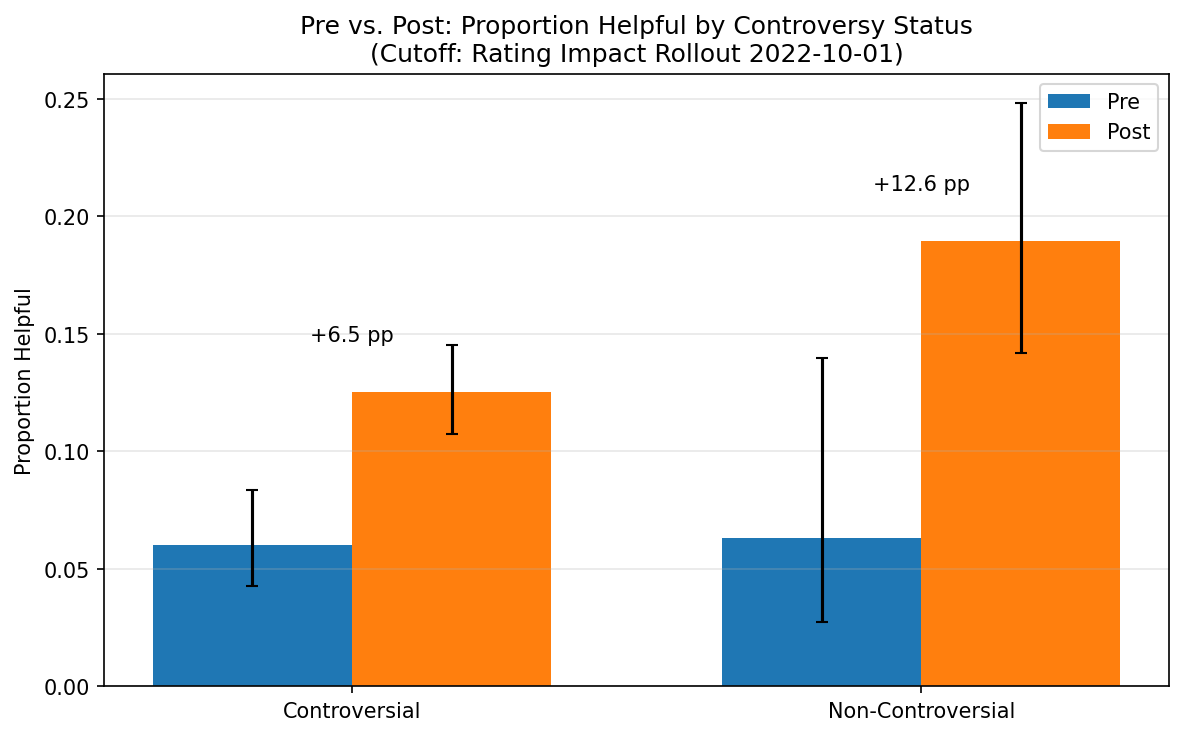}
    \caption{Pre-post change in the share of notes with final status Helpful by controversy category using topic-based controversy with topics classified using the prompted LLM approach around the Rating Impact rollout (cutoff: 2022-10-01). Bars show the mean proportion in the roughly 20 weeks before (Pre, blue) and after (Post, orange) the cutoff; error bars are 95\% CIs. The difference between controversial notes and non-controversial notes is stark; controversial notes face a dramatic decline in helpful ratings post-rollout, whereas non-controversial notes have increased helpful ratings.}
    \label{fig:topic-prompted}
\end{figure}

\begin{figure}
    \centering
    \includegraphics[width=.6\linewidth]{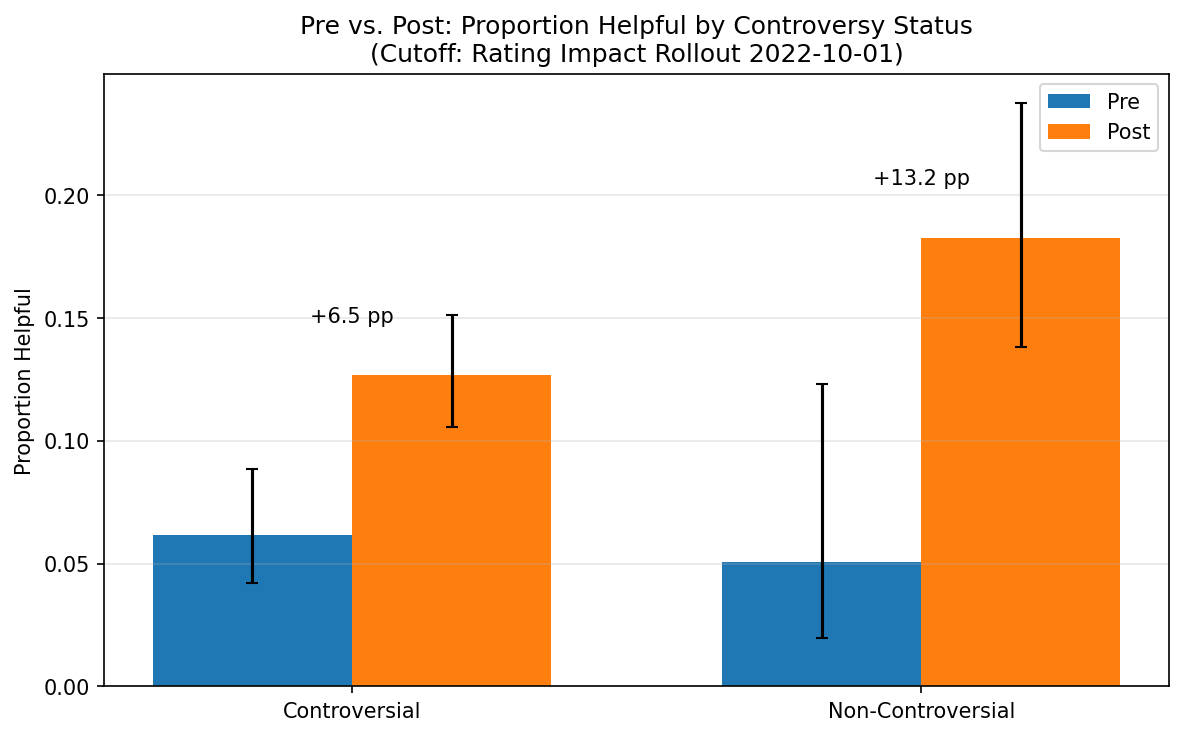}
    \caption{Pre-post change in the share of notes with final status Helpful by controversy category using topic-based controversy with topics classied using the managed LLM approach around the Rating Impact rollout (cutoff: 2022-10-01). Bars show the mean proportion in the roughly 20 weeks before (Pre, blue) and after (Post, orange) the cutoff; error bars are 95\% CIs. The difference between controversial notes and non-controversial notes is stark; controversial notes face a dramatic decline in helpful ratings post-rollout, whereas non-controversial notes have increased helpful ratings.}
    \label{fig:topic-managed}
\end{figure}

\begin{figure}
    \centering
    \includegraphics[width=.6\linewidth]{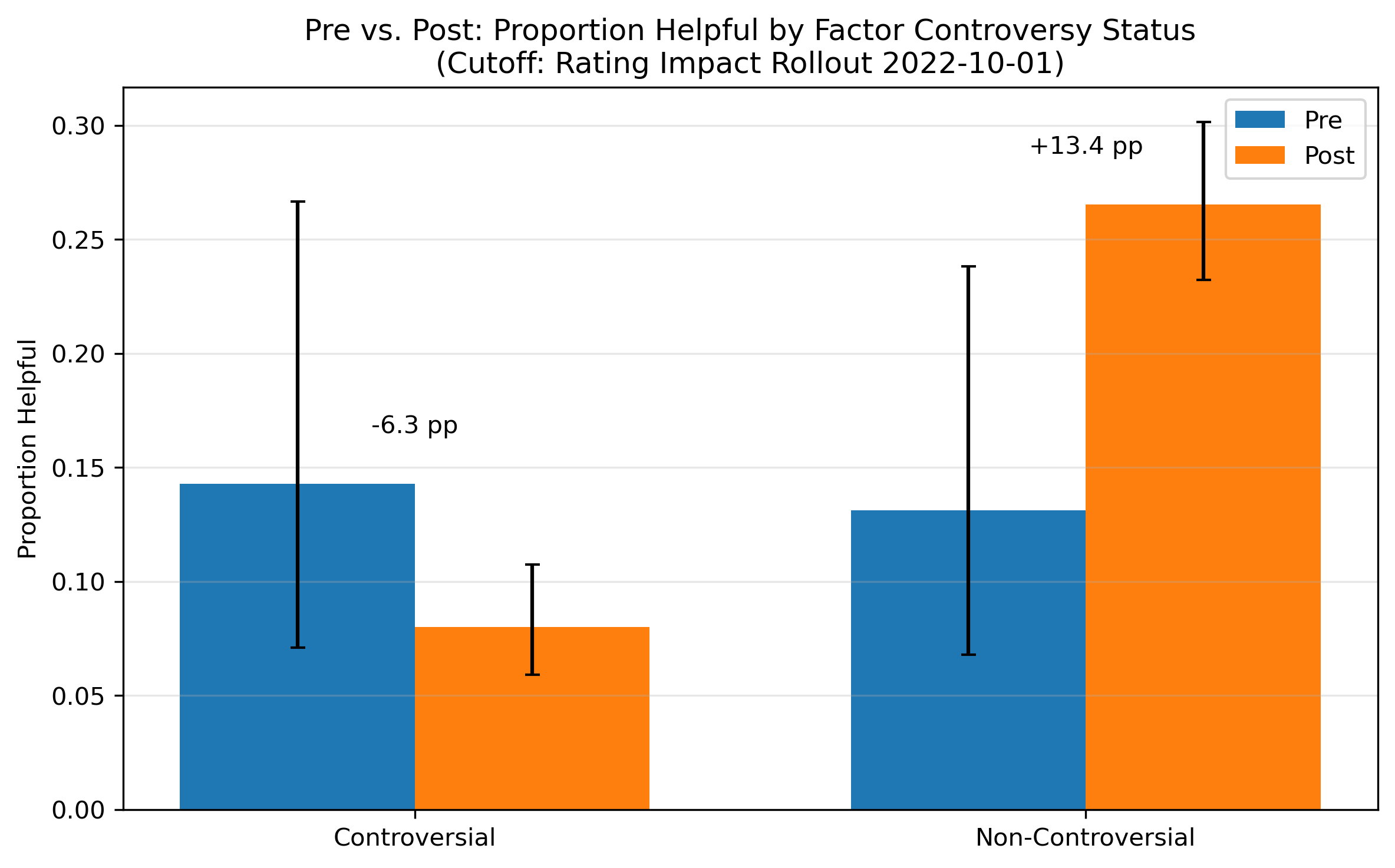}
    \caption{Pre-post change in the share of notes with final status Helpful by controversy category using \textit{factor-based} controversy around the Rating Impact rollout (cutoff: 2022-10-01). Bars show the mean proportion in the roughly 20 weeks before (Pre, blue) and after (Post, orange) the cutoff; error bars are 95\% CIs. The difference between controversial notes and non-controversial notes is stark; controversial notes face a dramatic decline in helpful ratings post-rollout, whereas non-controversial notes have increased helpful ratings.}
    \label{fig:factor-controversy-plot}
\end{figure}

\section{Two-Stage Weighted Matrix Factorization: Empirical Implementation}\label{sec:twostage imp}
In this section we give the implementation details for the two-stage matrix factorization approach documented in the empirical section (Fig. 5) of the main text. As proof of concept that our algorithm improves upon the predictive performance of the current algorithm, we test our algorithm on ratings data from Jan. 1, 2023 to June 1, 2024, dates during which there is a larger volume of ratings.

\subsection{Algorithm and Implementation Details}
The empirical comparison is designed to be conservative. Both the baseline algorithm and the two-stage algorithm are run on the same weekly filtered datasets, with the same underlying regularized matrix factorization objective and the same sign convention for the latent factor. The only modification is that the second-stage procedure reweights contributors using the inverse of their first-stage residual variance.

Consistent with the Community Notes implementation, we restrict attention each week to notes that have received at least five ratings and raters who have rated at least ten notes. We index weeks by \(t\), let \(D_t^{\mathrm{all}}\) denote all ratings observed up to and including week \(t\), and let \(D_t\) denote the filtered version of this cumulative dataset. We write \(\Theta_t\) for the weekly latent-factor estimates from the platform's baseline matrix factorization algorithm and \(\Theta_t^{\mathrm{ts}}\) for the corresponding estimates from the two-stage procedure.

The algorithm proceeds as follows.

\paragraph{Warm start.}
Before the weekly evaluation loop begins, we fit the existing matrix factorization algorithm to all ratings observed prior to June~1,~2023 and use the resulting latent quantities as a warm start. This reduces transient instability in the factor estimates early in the evaluation period.

\paragraph{Stage 1: baseline matrix factorization.}
For each week \(t\) from June~1,~2023 to June~1,~2024, we fit the platform's regularized matrix factorization algorithm to the filtered cumulative dataset \(D_t\), initialized at the previous week's fit. Denote the resulting first-stage estimates by
\[
\Theta_t^{(1)} = (\hat\mu_t,\hat h_t,\hat i_t,\hat f_t,\hat g_t).
\]

For each observed user--note pair \((u,n)\in D_t\), we compute the first-stage residual
\[
e_{un}^{(1)} := r_{un} - \hat r_{un}\!\left(\Theta_t^{(1)}\right),
\]
where \(\hat r_{un}(\Theta_t^{(1)})\) is the fitted rating estimate under the first-stage model. For each contributor \(u\), let \(N_t(u)\) denote the set of notes rated by \(u\) in the filtered cumulative dataset at week \(t\). We then estimate the contributor-specific residual variance by
\[
\hat\sigma_{u,t}^2 := \var\!\bigl(\{e_{un}^{(1)}: n\in N_t(u)\}\bigr).
\]
The corresponding second-stage weight is
\[
\hat w_{u,t} := \bigl(\max\{\hat\sigma_{u,t}^2,10^{-4}\}\bigr)^{-1}.
\]
The floor at \(10^{-4}\) prevents extremely large weights when a contributor's empirical residual variance is near zero.

\paragraph{Stage 2: weighted matrix factorization.}
We then refit the same regularized matrix factorization model on the same filtered cumulative dataset \(D_t\), but now using contributor-specific weights:
\[
\Theta_t^{(2)}
=
\arg\min_{\Theta}
\sum_{(u,n)\in D_t}
\hat w_{u,t}\bigl(r_{un}-\hat r_{un}(\Theta)\bigr)^2 + \lambda \|\Theta\|_2^2.
\]
We initialize this second-stage optimization at the first-stage solution \(\Theta_t^{(1)}\). The resulting fitted quantities are stored as
\[
\Theta_t^{\mathrm{ts}} := \Theta_t^{(2)}.
\]

Algorithm \ref{alg:two_stage_mf} outputs two dataframes, one with latent factors from X's existing matrix factorization algorithm, and one with latent factors from the two-stage algorithm.

\paragraph{Sign convention.}
Because matrix factorization identifies the latent factor only up to a global sign, we orient the factors after each weekly fit so that the majority of users have negative latent factor. This keeps the interpretation of majority- and minority-aligned contributors consistent across weeks. The sign-fixed second-stage estimates are also used as the warm start for the next week's fit.

\begin{algorithm2e}[H]
\SetAlgoLined
\DontPrintSemicolon
\caption{Two-Stage Matrix Factorization}
\label{alg:two_stage_mf}

\KwIn{\texttt{ratings\_df}: a dataframe containing rater, note pairs with the date of the rating, the rating (\textit{Helpful, Somewhat Helpful, Not Helpful}) and the estimated helpfulness level based on previous latent factors}
\KwOut{$\Theta_t$ and $\Theta_t^\mathrm{ts}$: dataframes containing weekly note and rater latent factor estimates from X's matrix factorization algorithm and from the two-stage algorithm, respectively}
% \KwParam{Min ratings per note $K_n=5$, Min ratings per user $K_u=10$, Regularization $\lambda$}

\BlankLine
\tcc{Warm start initialization}
Let $\mathcal{D}_{0} \leftarrow \{(u, n, r_{un}) \mid t < T_{\mathrm{start}}\}$\;

$\Theta_{\mathrm{prev}} \leftarrow \text{TrainExistingMF}(\mathcal{D}_{0})$ 
% \tcp*{Initial warm start}

\BlankLine
\tcc{Weekly latent factor estimation loop}
\For{week $t \leftarrow T_{\mathrm{start}}$ \KwTo $T_{\mathrm{end}}$}{
    Let $\mathcal{D}_t^{\mathrm{all}}\leftarrow \{(u, n, r_{un}) \mid \text{timestamp} \le t\}$\;
    
    \BlankLine
    % \tcp{Step 1: Iterative Density Filtering}
    $\mathcal{D}_t \leftarrow \text{Filter}(\mathcal{D}_t^{\mathrm{all}})$ s.t. within $\mathcal D_t$ each note has at least 5 ratings and each user has rated as least 10 notes \;
    
    \BlankLine
    \tcp{Stage 1: X's Matrix Factorization Algorithm}
    Initialize $\Theta^{(1)}$ with $\Theta_{\mathrm{prev}}$\;
    
    $\Theta^{(1)} \leftarrow \operatorname*{argmin}_{\Theta} \sum_{(u,n) \in \mathcal{D}_t} (r_{un} - \hat{r}_{un}(\Theta))^2 + \lambda \|\Theta\|^2$\;

    $\Theta_t \leftarrow \text{append}(\Theta_t, \Theta^{(1)})$
    
    \BlankLine
    % \tcp{Weekly factor estimation loop}
    \For{each user $u \in \mathcal{D}_t$}{
        Calculate residuals: $e_{un} \leftarrow r_{un} - \hat{r}_{un}(\Theta^{(1)})$\;
        
        Estimate variance: $\sigma_u^2 \leftarrow \text{Var}(\{e_{un} \mid n \in \text{set of notes rated by user $u$}\})$\;
        
        Compute weight: $w_u \leftarrow (\max(\sigma_u^2, 10^{-4}))^{-1}$\;
    }
    
    \BlankLine
    % \tcp{Step 4: Weighted Refinement (Stage 2)}
    Initialize $\Theta^{(2)}$ with $\Theta^{(1)}$\;
    
    $\Theta^{(2)} \leftarrow \operatorname*{argmin}_{\Theta} \sum_{(u,n) \in \mathcal{D}_t} w_u (r_{un} - \hat{r}_{un}(\Theta))^2 + \lambda \|\Theta\|^2$\;

    $\Theta_t^{\mathrm{ts}} \leftarrow \text{append}(\Theta_t^{\mathrm{ts}}, \Theta^{(2)})$
    
    \BlankLine
    % \tcp{Step 5: Alignment and Forecast}
    $\Theta_{\mathrm{prev}} \leftarrow \text{FixFactorSigns}(\Theta^{(2)})$ \tcp*{Flip signs if needed}
}
\end{algorithm2e}

\subsection{Prediction Targets and Evaluation Protocol}
We compare the baseline and two-stage procedures using out-of-sample one-week-ahead predictions. For each week \(t\), the baseline fit \(\Theta_t\) and the two-stage fit \(\Theta_t^{\mathrm{ts}}\) are computed using ratings observed up to and including week \(t\). These week-\(t\) estimates are then used to predict ratings observed in the following week.

Formally, for every rating \(r_{un}\) observed in week \(t+1\), we compute two fitted values:
\[
\hat r_{un}^{\mathrm{MF},t} := \hat r_{un}(\Theta_t),
\qquad
\hat r_{un}^{\mathrm{TS},t} := \hat r_{un}(\Theta_t^{\mathrm{ts}}).
\]
The corresponding one-week-ahead residuals are
\[
e_{un}^{\mathrm{MF},t+1} := r_{un} - \hat r_{un}^{\mathrm{MF},t},
\qquad
e_{un}^{\mathrm{TS},t+1} := r_{un} - \hat r_{un}^{\mathrm{TS},t}.
\]

For each evaluation week, we summarize predictive performance using several standard statistics derived from these out-of-sample residuals. In particular, the main text reports the weekly mean absolute residual
\[
\mathrm{MAR}_t
:=
\frac{1}{|\mathcal R_{t+1}|}
\sum_{(u,n)\in \mathcal R_{t+1}}
|e_{un}^{t+1}|,
\]
and the weekly median absolute residual
\[
\mathrm{MedAR}_t
:=
\mathrm{median}\bigl\{|e_{un}^{t+1}|:(u,n)\in\mathcal R_{t+1}\bigr\},
\]
computed separately for the baseline and two-stage procedures, where \(\mathcal R_{t+1}\) denotes the set of ratings observed in week \(t+1\). These are the quantities shown in the main-text figure comparing one-week-ahead predictive performance of the two methods.

We also report the corresponding one-week-ahead MSE,
\[
\mathrm{MSE}_t
:=
\frac{1}{|\mathcal R_{t+1}|}
\sum_{(u,n)\in \mathcal R_{t+1}}
\bigl(e_{un}^{t+1}\bigr)^2,
\]
but our main empirical comparison emphasizes the absolute-residual metrics because they are easier to interpret on the scale of the platform's fitted rating scores.

\newpage

\section{Theory Appendix}\label{sec:proofs}

This appendix collects the technical results underlying the theoretical section of the main text. Throughout, we study a stylized regularized rank-1 matrix factorization model that isolates the effect of conformity incentives on estimation. This stylized estimator is intentionally simpler than the full Community Notes production system used in the empirical sections. The distinction is important: the empirical reconstruction follows the open-source Community Notes pipeline, whereas the theory abstracts to a regularized rank-1 factorization in order to characterize what object is recovered when contributors strategically anticipate the platform's eventual consensus.

The appendix is organized to mirror the logic of the main text. We first state the estimation problem and regularity conditions. We then prove consistency of note helpfulness under private-signal reporting, show that conformity makes estimators biased, characterize the resulting distortion in user factors and minority compression, and finally establish a statistical guarantee for the two-stage weighted estimator.

\subsection{Estimation Model and Assumptions}

We begin by restating the estimator analyzed in the proofs.

 Let \(\Omega\subseteq [U]\times[N]\) denote the set of observed user--note ratings. The platform fits the regularized rank-1 model
\begin{equation}
\label{eq:ridge-opt}
\arg\min_{\mu,h,i,f,g}
\sum_{(u,n)\in\Omega}
\bigl(r_{un}-(\mu+h_u+i_n+f_u g_n)\bigr)^2
+
\lambda_h\|h\|_2^2
+
\lambda_i\|i\|_2^2
+
\lambda_f\|f\|_2^2
+
\lambda_g\|g\|_2^2.
\end{equation}

For simplicity of notation, for a parameter vector $\theta = (\mu, h, i, f, g)$ and a set of regularization parameters $\lambda = (\lambda_h, \lambda_i, \lambda_f, \lambda_g)$, we write
\[
    L(\theta)=\sum_{(u,n)\in\Omega}\bigl(r_{un}-(\mu+h_u+i_n+f_u g_n)\bigr)^2, \qquad G_\lambda(\theta) = \lambda_h\|h\|_2^2 + \lambda_i\|i\|_2^2 + \lambda_f\|f\|_2^2 + \lambda_g\|g\|_2^2.
\]
Here and below, for a matrix \(A\in\mathbb{R}^{U\times N}\), we write $\|A\|_*$ for the nuclear norm of a matrix, $\|A\|_F$ for the Frobenius norm, and $\|A\|_\infty$ for the entrywise maximum norm.

We use two sets of regularity conditions. Assumption~\ref{asmp:random-latents} governs the latent components, noise, and sampling scheme under private-signal reporting. Assumption~\ref{ass:conformity} adds conditions on the anticipated consensus \(m_n\) and conformity weights $\rho_n = \rho(c_n)$ for the conformity regime.

% \paragraph{Assumptions used in the proof (move to theoerem statment)}
\begin{assumption}[Random latent components, boundedness, and MCAR sampling]
\label{asmp:random-latents}
Assume:
\begin{enumerate}
    \item \(\mu \in \mathbb R\) is fixed.
    \item \(\{(h_u,f_u)\}_{u=1}^U\) are i.i.d., with
    \[
    \mathbb E[h_u]=0,\qquad \mathbb E[f_u]=c,
    \]
    \[
    \var(h_u)=\sigma_h^2>0,\qquad \var(f_u)=\sigma_f^2>0,
    \]
    and \(h_u\) is independent of \(f_u\). Moreover, \(|h_u|\le B_h\) and \(|f_u|\le B_f\) almost surely.
    \item \(\{(i_n,g_n)\}_{n=1}^N\) are i.i.d., with
    \[
    \mathbb E[i_n]=0,\qquad \mathbb E[g_n]=0,
    \]
    \[
    \var(i_n)=\sigma_i^2>0,\qquad \var(g_n)=\sigma_g^2>0,
    \]
    and \(i_n\) is independent of \(g_n\). Moreover, \(|i_n|\le B_i\) and \(|g_n|\le B_g\) almost surely.
    \item The user-side variables \(\{(h_u,f_u)\}_{u=1}^U\) are independent of the note-side variables \(\{(i_n,g_n)\}_{n=1}^N\).
    \item The noise variables \(\epsilon_{un}\) are independent, mean-zero, \(\sigma_\epsilon\)-sub-Gaussian, and independent of all other latent variables (see e.g., \cite{vershynin2012introduction} Definition 5.7).
    \item Each entry \((u,n)\) is observed independently with probability \(p\in(0,1)\), and \(U \asymp N\).
\end{enumerate}
\end{assumption}

\begin{assumption}[Conformity parameters]
\label{ass:conformity}
In addition to Assumption~\ref{asmp:random-latents}, assume:
\begin{enumerate}
    \item \(m_n\) are i.i.d.\ draws from a bounded distribution on a compact interval with positive finite variance.
    \item \(\rho_n:=\rho(c_n)\) is a weakly decreasing function in $c_n$. 
    \item Let $\bar \rho_N = N^{-1}\sum_n \rho_n$. Assume that $\bar \rho_N \to \bar \rho \in [0, 1]$ a deterministic constant.
    \item For each note \(n\), the random variables \(m_n\) and \(g_n\) are independent.
\end{enumerate}
\end{assumption}

\subsection{Private-Signal Reporting: Consistency of Note Helpfulness (Proof of Theorem 1)}\label{proof:thm1}
We first study the benchmark case in which contributors report their private signals, so the platform observes a noisy version of the latent signal matrix
\[
S=\mu \mathbf{1}_U \mathbf{1}_N^\top + h\mathbf{1}_N^\top + \mathbf{1}_U i^\top + f g^\top.
\]
The proof proceeds in three steps. First, we show that the note intercept can be uniquely decoded from the fitted matrix by means of a canonical centered decomposition. Second, we show that a global solution to \eqref{eq:ridge-opt} is consistent for the true intercept and factor terms under the canonical decomposition. Finally, we connect the estimator of \eqref{eq:ridge-opt} to a corresponding nuclear-norm matrix completion problem.

\paragraph{Centered decomposition}
Here, we identify the intercept terms from a complete signal matrix. For any tuple \(\theta=(\mu,h,i,f,g)\), let
\[
S(\theta) = \mu \mathbf{1}_U \mathbf{1}_N^\top + h\mathbf{1}_N^\top + \mathbf{1}_U i^\top + f g^\top.
\]
We use the canonical centered decomposition of \(S(\theta)\) to recover centered intercept and factor components directly from the matrix itself.

% The next lemma is the key identification step. 

\begin{lemma}\label{lem:canonical-centering}
    Every matrix \(S\) with entries of the form
    \[
        s_{un} = \mu + h_u + i_n + f_u g_n
    \]
    admits a canonical representation with
    \[
        s_{un} = \mu^{c} + h_u^{c} + i_n^{c} + f_u^{c} g_n^{c}
    \]
    such that 
    \[
        \mathbf 1_U^{\top} h^{c}=0, \qquad 
        \mathbf 1_N^{\top} i^{c}=0, \qquad
        \mathbf 1_U^{\top} f^{c}=0, \qquad
        \mathbf 1_N^{\top} g^{c}=0.
    \]
    Since the factors \(f, g\) are only identifiable up to sign and scale, we may also choose a canonical sign and scaling representation:
    \[
        \frac{1}{U}\|f^{c}\|_2^2 = 1, \qquad \langle f^{c}, f\rangle \geq 0.
    \]
\end{lemma}

\begin{proof}
    Let
    \[
        \bar h = \frac{1}{U}\mathbf 1_U^{\top}h, \qquad
        \bar i = \frac{1}{N}\mathbf 1_N^{\top}i,\qquad
        \bar f = \frac{1}{U}\mathbf 1_U^{\top}f, \qquad
        \bar g = \frac{1}{N}\mathbf 1_N^{\top}g,
    \]
    and define
    \begin{align*}
        \mu^c &= \mu + \bar h + \bar i + \bar f\,\bar g,\\
        h^c &= h - \bar h \mathbf 1_U + \bar g\bigl(f-\bar f\mathbf 1_U\bigr),\\
        i^c &= i - \bar i \mathbf 1_N + \bar f\bigl(g-\bar g\mathbf 1_N\bigr),\\
        f^c &= f-\bar f\mathbf 1_U, \qquad g^c = g-\bar g\mathbf 1_N.
    \end{align*}
    Writing $ f = f^c + \bar f \mathbf 1_U, g = g^c + \bar g \mathbf 1_N$, we have
    \[
        f g^{\top} = f^c (g^c)^{\top} + \bar f \mathbf 1_U (g^c)^{\top} + \bar g f^c \mathbf 1_N^{\top} + \bar f \bar g \mathbf 1_U \mathbf 1_N^{\top}.
    \]
    Substituting and collecting terms gives
    \[
        S(\theta) = \mu^c \mathbf 1_U \mathbf 1_N^{\top} +
        h^c \mathbf 1_N^{\top} + \mathbf 1_U (i^c)^{\top} + f^c (g^c)^{\top}.
    \]
    The centering conditions
    \[
        \mathbf 1_U^{\top}h^c=0, \qquad 
        \mathbf 1_N^{\top}i^c=0, \qquad
        \mathbf 1_U^{\top}f^c=0, \qquad
        \mathbf 1_N^{\top}g^c=0
    \]
    are immediate from the definitions. Multiplying the centered decomposition on the left and right by vectors of ones then yields
    \[
        \mu^c = \frac{1}{UN}\mathbf 1_U^{\top}S\mathbf 1_N,
    \]
    \[
        i^c = \frac{1}{U}S^{\top}\mathbf 1_U - \mu^c \mathbf 1_N, \qquad
        h^c = \frac{1}{N}S\mathbf 1_N - \mu^c \mathbf 1_U.
    \]

    Finally, the rank-one term is unchanged under reciprocal rescaling:
    \[
        f^c g^{\top}
        =
        \Bigl(\frac{f^c}{c}\Bigr)\Bigl(c \, g^c\Bigr)^{\top},
        \qquad c>0.
    \]
    For \(f^c \neq 0\), take $c = \sqrt{\frac{1}{U}\|f^c\|_2^2}$ and redefine $f^c \leftarrow f^c/c$,  $g^c \leftarrow c\,g^c.$ Then $ \frac{1}{U}\|f^c\|_2^2 = 1.$ Since simultaneously replacing \((f^c,g^c)\) by \((-f^c,-g^c)\) leaves the product unchanged, we may also choose the sign so that $\langle f^c, f\rangle \geq 0.$ These operations preserve both the representation and the centering conditions. If \((f^c) = 0\), then the rank-one term vanishes.
\end{proof}

In the following we use $\theta^0 = (\mu^0, h^0, i^0, f^0, g^0)$ to denote the true parameter values, $\theta^c = (\mu^c, h^c, i^c, f^c, g^c)$ to denote the canonically centered representation of the true parameter values, and $\hat \theta, \hat \theta^c$ to denote the global minimizer of \eqref{eq:ridge-opt} and its canonically centered representation, respectively.

% unless explicitly otherwise stated, we work in the setting where all parameters, true and estimated, are in their canonical centered decomposition. Thus, for notation we use $\mu, h_u, i_n, f_u, g_n$ to denote the canonical decomposition of the true parameters. At the end of each section, we remark on how our results transfer back to the original factors and intercepts.

\begin{remark}\label{rmk:canonical-centered}
    Under the canonical centered decomposition of the matrix $S(\theta^0)$, it is straightforward to see that under Assumption \ref{asmp:random-latents} the centered intercepts and factors $\mu^c, h_u^c, i_n^c, f_u^c, g_n^c$ are still bounded random variables with strictly positive, finite variance. We use the notation $\sigma_h^c, \sigma_i^c, \sigma_f^c, \sigma_g^c$ to denote the standard deviations of the canonical centered factors in the following sections. Note that according to the decomposition given in Lemma \ref{lem:canonical-centering}, $\sigma_f^c = 1$.
\end{remark}

\begin{corollary}\label{cor:asmp-center}
    Under Assumption \ref{asmp:random-latents}, 
    \[
        \bar h^0 = o_p(1), \qquad \bar i^0 = o_p(1), \qquad \bar f^0 \pto c, \qquad \bar g^0 = o_p(1).
    \]
    Thus, as a corollary of Lemma \ref{lem:canonical-centering} we have that, 
    \[
        \frac{1}{N}\sum_n \omega_{un} g_n^c = o_p(1), \qquad \frac 1N \sum_n \omega_{un} (g_n^c)^2 \pto p\sigma_g^2, \qquad \frac{1}{N}\sum_n \omega_{un} g_n^c \epsilon_{un} = o_p(1).
    \]
\end{corollary}
\begin{proof}
    This is immediate by Lemma \ref{lem:canonical-centering} and the weak law of large numbers.
\end{proof}

\paragraph{Consistency of solution under missing data}

Recall that when $\rho \equiv 1$, users report their true latent signals $r_{un} = s_{un} + \epsilon_{un}$. \eqref{eq:ridge-opt} is equivalent to the following optimization problem where $\omega_{un} \sim \Bern(p)$ with constant probability $p$:
\begin{equation*}
% \label{eq:ridge-opt-2}
    \arg\min_{\mu,h,i,f,g} \sum_{u, n} \omega_{un}\left(r_{un}-(\mu+h_u+i_n+f_u g_n)\right)^2 + \lambda_h\|h\|_2^2 + \lambda_i\|i\|_2^2 + \lambda_f\|f\|_2^2 + \lambda_g\|g\|_2^2.
\end{equation*}

% Let $\mu, h, i, f, g$ be the canonical representation of the true parameters. Let $\hat \mu, \hat h, \hat i, \hat f, \hat g$ denote the canonically identified observed optimizers for \eqref{eq:ridge-opt}. Then, we can prove the following theorem.

\begin{theorem}\label{thm:ridge-recovery}
    Let $\theta^c, \hat \theta^c$ be defined as above. Under Assumption \ref{asmp:random-latents}, as $U, N\to\infty$, for each fixed user $u$ and note $n$ we have 
    \[
        \hat \mu^c \pto \mu^c, \qquad \hat h_u^c \pto h_u^c, \qquad \hat f_u^c \pto f_u^c, \qquad \hat i_n^c \pto i_n^c, \qquad \hat g_n^c \pto g_n^c.
    \]
\end{theorem}

The proof of Theorem \ref{thm:ridge-recovery} requires the following lemmas.

\begin{lemma}\label{lem:regularization-to-zero}
    Suppose Assumption \ref{asmp:random-latents} holds. Assume that $\lambda_h, \lambda_f, \lambda_i, \lambda_g = o(1)$ are non-negative. Then, 
    \[
        \frac{\lambda_h}{N} |\hat h_u| \pto 0, \qquad \frac{\lambda_f}{N} |\hat f_u| \pto 0, \qquad \frac{\lambda_i}{U} |\hat i_n| \pto 0, \qquad \frac{\lambda_g}{U}|\hat g_n| \pto 0.
    \]
\end{lemma}
\begin{proof}
    Since $\hat \theta$ is a global optimum for \eqref{eq:ridge-opt}, we have that
    \begin{align*}
        L(\hat \theta) + G_\lambda(\hat\theta) &\leq L(\theta^0) + G_\lambda(\theta^0) \\
        &= \sum_{u, n}\omega_{un}\epsilon_{un}^2 + G_\lambda(\theta^0) = O_p(UN)
    \end{align*}
    where the last equality is true by weak law of large numbers and Assumption \ref{asmp:random-latents}. This implies that
    \[
        G_\lambda(\hat\theta) = \lambda_h \|\hat h\|_2^2 + \lambda_i \|\hat i\|_2^2 + \lambda_f \|\hat f\|_2^2 + \lambda_g \|\hat g\|_2^2 = O_p(UN).
    \]
    For any fixed user $u$, we then have that
    \[
        \frac{\lambda_h}{N}|\hat h_u| \leq \frac{\lambda_h}{N}\|\hat h\|_2 = O_p\left(\sqrt{\frac{U\lambda_h}{N}}\right) = o_p(1).
    \]
    The rest of the statements can be proven analogously.
\end{proof}

\begin{lemma}\label{lem:frobenius-norm-bound}
    Suppose that Assumption \ref{asmp:random-latents} holds. Let $\hat \theta$ be a global optimizer of \eqref{eq:ridge-opt} and $\hat \theta^c$ be its canonically centered representation. Let $S(\hat\theta)$ be the matrix with entries $s_{un}(\hat\theta) = s_{un}(\hat\theta^c) = \hat \mu^c + \hat h_u^c + \hat i_n^c + \hat f_u^c \hat g_n^c$ and $S(\theta)$ be the matrix with entries $s_{un}(\theta) = s_{un}(\theta^c) =\mu^c + h_u^c + i_n^c + f_u^c g_n^c$. Further, assume that $\lambda_h, \lambda_f = o(N^{-1/2})$, and $\lambda_i, \lambda_g = o(U^{-1/2})$. Then,
    \[
        \|S(\hat\theta) - S(\theta)\|_F = O_p\left(\sqrt{U\wedge N}\right).
    \] 
    In addition, the canonically identified solutions satisfy,
    \[
        |\hat \mu^c - \mu^c| = O_p\left((U\wedge N)^{-1/2}\right)
    \]
    \[
        \frac{1}{\sqrt{U}}\|\hat h^c - h^c\|_2 = O_p\left((U\wedge N)^{-1/2}\right), \qquad \frac{1}{\sqrt U}\|\hat f^c - f^c\|_2 = O_p\left((U\wedge N)^{-1/2}\right)
    \]
    \[
        \frac{1}{\sqrt{N}}\|\hat i^c - i^c\|_2 = O_p\left((U\wedge N)^{-1/2}\right), \qquad \frac{1}{\sqrt N}\|\hat g^c - g^c\|_2 = O_p\left((U\wedge N)^{-1/2}\right).
    \]
\end{lemma}
\begin{proof}
    Let $\hat \theta$ be a global optimizer of \eqref{eq:ridge-opt} and let $\hat \theta^c$ be its canonical representation as given by Lemma \ref{lem:canonical-centering}. Then, letting
    \[
        \delta \mu = \hat \mu^c - \mu^c, \quad \delta h = \hat h^c - h^c, \quad \delta i = \hat i^c - i^c, \quad \delta f = \hat f^c - f^c, \quad \delta g = \hat g^c - g^c,
    \]
    we can write
    \[
        \hat s_{un} = \mu^c + \delta \mu + h_u^c + \delta h_u + i_n^c + \delta i_n + (f_u^c + \delta f_u)(g_n^c + \delta g_n).
    \]
    Let $\delta \hat \theta = (\delta \mu, \delta h, \delta i, \delta f, \delta g)$. Since $\hat \theta$ is a global optimizer for \eqref{eq:ridge-opt}, we furthermore have that
    \[
        \sum_{u, n} \omega_{un}(r_{un} - s_{un})^2 - \sum_{u, n}\omega_{un}(r_{un} - \hat s_{un})^2 + G_\lambda(\theta) - G_\lambda(\hat\theta) > 0,
    \]
    where
    \[
        G_\lambda(\theta) = \lambda_h\|h\|_2^2 + \lambda_i\|i\|_2^2 + \lambda_f\|f\|_2^2 + \lambda_g\|g\|_2^2.
    \]
    Some rearranging of the above implies that
    \begin{align}
        \sum_{u, n} \omega_{un}(\hat s_{un} - s_{un})^2 &\leq 2\sum_{u, n}\omega_{un} \epsilon_{un}(\hat s_{un} - s_{un}) + G_\lambda(\theta) - G_\lambda(\hat\theta) \\
        &\leq 2\sum_{u, n}\omega_{un} \epsilon_{un}(\hat s_{un} - s_{un}) + G_\lambda(\theta).\label{eq:bounded-error}
    \end{align}
    Let $\delta_{un} = \hat s_{un} - s_{un}$. Then, we have that
    \begin{align*}
        \delta_{un} &= x_{un} + y_{un}
    \end{align*}
    where 
    \[
        x_{un} = \delta \mu + \delta h_u + \delta i_n + f_u^c \delta g_n + \delta f_u g_n^c, \qquad y_{un} = \delta f_u \delta g_n.
    \]
    Rewriting the error bound from above and using the assumptions on the regularization terms, we then have
    \begin{align*}
        \sum_{u, n}\omega_{un}(x_{un} + y_{un})^2 \leq 2 \sum_{u, n}\omega_{un} \epsilon_{un}(x_{un} + y_{un}) + O_p(\sqrt{U \wedge N}).
    \end{align*}
    This implies that
    \[
        \frac{1}{2}\sum_{u, n}\omega_{un} x_{un}^2 \leq \sum_{u, n}\omega_{un} y_{un}^2 + 2 \sum_{u, n}\omega_{un} \epsilon_{un}(x_{un} + y_{un}) + O_p(\sqrt{U \wedge N}).
    \]

    First, we can lower bound the left-hand side. Under the canonical choice of centering constraints:
    \[
        \delta h^\top \one_U =  \delta i^\top\one_N = \delta f^\top \one_U =  \delta g^\top \one_N= 0,
    \]
    we have
    \begin{align*}
        \sum_{u, n}x_{un}^2 &= UN (\delta \mu)^2 + N\|\delta h\|_2^2 + U\|\delta i\|_2^2 + \|f^c \delta g^\top + \delta f (g^c)^\top\|_F^2\\
        &=UN (\delta \mu)^2 + N\|\delta h\|_2^2 + U\|\delta i\|_2^2 + \|f^c\|_2^2 \|\delta g\|_2^2 + \|g^c\|_2^2 \|\delta f\|_2^2 + 2 \langle f^c, \delta f\rangle\langle g^c, \delta g\rangle.
    \end{align*}
    By Lemma \ref{lem:canonical-centering}, we have that
    \[
        \langle f^c, \delta f\rangle = - \frac{1}{2}\|\delta f\|_2^2.
    \]
    Continuing from above we have that there exists constants $c_1, c_2$ with high probability such that
    \begin{align*}
         \|f^c\|_2^2 \|\delta g\|_2^2 + \|g^c\|_2^2 \|\delta f\|_2^2 + 2 \langle f^c, \delta f\rangle\langle g^c, \delta g\rangle &= U\|\delta g\|_2^2 + \|g^c\|_2^2 \|\delta f\|_2^2 - \|\delta f\|_2^2 \langle g^c, \delta g\rangle \\
         &\geq c_1\left(U\|\delta g\|_2^2 + N \|\delta f\|_2^2\right) -c_2 \|\delta f\|_2^2 \|\delta g\|_2^2,
    \end{align*}
    where in the last line we bounded $\langle g^c, \delta g\rangle$ by Cauchy-Schwarz and then used Young's inequality. Note that by Assumption \ref{asmp:random-latents} and Remark \ref{rmk:canonical-centered}, $f^c, g^c$ have component-wise variances bounded and away from zero, so continuing from above we have for some constant $c$ that 
    \[
        \sum_{u, n} x_{un}^2 \geq c\left(UN (\delta \mu)^2 + N\|\delta h\|_2^2 + U\|\delta i\|_2^2 + N \|\delta f\|_2^2 + U \|\delta g\|_2^2\right) - c_2 \|\delta f\|_2^2 \|\delta g\|_2^2.
    \]
   
    Now, it suffices to show that the observed errors do not deviate too much from the population errors. Define the following orthonormal basis matrices. Let $Q_h$ be an orthonormal basis for the space of vectors $\{v \in \mathbb R^U : \one_U^\top v = 0\}$, $Q_i, Q_g$ be an orthonormal basis for the space of vectors $\{v \in \mathbb R^N : \one_N^\top v = 0\}$. Finally, let $Q_f$ be an orthonormal basis for the tangent space of vectors $\{v \in \mathbb R^U: \one_U^\top v = 0 , (f^c)^\top v =0\}$. Define
    \[
        z_{un} = \begin{pmatrix}
            1 \\
            Q_h(u, :) \\
            Q_i(n, :) \\
            g_n^c Q_f(u, :) \\
            f_u^c Q_g(n, :)
        \end{pmatrix}.
    \]
    Let 
    \[
        \hat G = \sum_{u, n}\omega_{un} z_{un} z_{un}^\top, \qquad G = \sum_{u, n}z_{un}z_{un}^\top,
    \]
    and $D$ be the block diagonal matrix
    \[
        D = \mathrm{diag}\begin{pmatrix}
            UN & NI_{U-1} & UI_{N-1} & \|g\|_2^2 I_{U-2} & \|f\|_2^2 I_{N-1}
        \end{pmatrix}.
    \]
    where $I_N$ denotes the $N \times N$ identity matrix. Write $\delta f = Q_f \alpha + r f^c$, where
    \[
        r = \frac{\langle f^c, \delta f\rangle}{\|f^c\|_2^2}.
    \]
    By the normalization imposed on the norm in the canonically centered representation in Lemma \ref{lem:canonical-centering}, $\|\hat f^c\|_2^2 = \|f^c\|_2^2 = U$, which implies that $2\langle f^c, \delta f\rangle + \|\delta f\|_2^2 = 0,$ and therefore
    \[
        r = -\frac{\|\delta f\|_2^2}{2U}.
    \]
    
    Define
    \[
        x_{un}^{\mathrm{lin}} = \delta \mu + \delta h_u + \delta i_n + f_u^c \delta g_n + g_n^c (Q_f \alpha)_u, \qquad x_{un}^{\mathrm{rem}} = r f_u^c g_n^c,
    \]
    so $x_{un} = x_{un}^{\mathrm{lin}} + x_{un}^{\mathrm{rem}}.$ Then we can write
    \[
        \sum_{u, n} \omega_{un} (x_{un}^{\mathrm{lin}})^2 = \delta \vartheta^\top \hat G \delta \vartheta, \qquad \sum_{u, n} (x_{un}^{\mathrm{lin}})^2 = \delta \vartheta^\top G \delta \vartheta,
    \]
    where
    \[
        \delta \vartheta = (\delta \mu, a, b, \alpha, d), \qquad a = Q_h\delta h, \quad b = Q_i \delta i, \quad d = Q_g \delta g.
    \]
    In addition, notice that $\tilde G = D^{-1/2} G D^{-1/2} = I$. Write $\hat{\tilde G} = D^{-1/2} \hat G D^{-1/2}$. Then,
    \[
        v^\top (\hat{\tilde G} - pI) v = \sum_{u, n}(\omega_{un} -p)(v^\top \tilde z_{un})^2, \qquad \tilde z_{un} = D^{-1/2}z_{un}.
    \]
    We have that
    \[
        \|\tilde z_{un}\|_2^2 \le \frac{1}{UN} + \frac{\|Q_h(u,:)\|_2^2}{N} + \frac{\|Q_i(n,:)\|_2^2}{U} + \frac{|g_n^c|^2\,\|Q_f(u,:)\|_2^2}{\|g^c\|_2^2} + \frac{|f_u^c|^2\,\|Q_g(n,:)\|_2^2}{\|f^c\|_2^2}.
    \]
    Using the fact that
    \[
        \|Q_h(u,:)\|_2,\ \|Q_i(n,:)\|_2,\ \|Q_f(u,:)\|_2,\ \|Q_g(n,:)\|_2 \le 1,
    \]
    together with Assumption \ref{asmp:random-latents} and Remark \ref{rmk:canonical-centered}, this gives
    \[
    \|\tilde z_{un}\|_2^2 \le \frac{C}{U \wedge N}.
    \]
     
    Therefore each summand $X_{un}:=(\omega_{un}-p)\tilde z_{un}\tilde z_{un}^\top$ satisfies
    \[
        \|X_{un}\|_{\op}\le \frac{C}{U \wedge N}.
    \]

    We now use a matrix Bernstein inequality to bound $\|\hat{\tilde G} - pI\|_{\mathrm{op}}$. Let
    \[
        \mathcal F = \sigma\!\left(\{f_u^c\}_{u=1}^U,\{g_n^c\}_{n=1}^N,Q_h,Q_i,Q_f,Q_g\right).
    \]
    Conditional on $\mathcal F$, the vectors $\tilde z_{un}$ are fixed, and the matrices $X_{un}$ are independent, self-adjoint, and mean-zero.
    From the bound above on $\|\tilde z_{un}\|_2^2$, there exists a deterministic constant $C_0$ such that
    \[
        \|X_{un}\|_{\op}\le \frac{C_0}{U\wedge N}.
    \]
    For the variance proxy, notice that
    \[
        X_{un}^2 = (\omega_{un}-p)^2\,\|\tilde z_{un}\|_2^2\, \tilde z_{un}\tilde z_{un}^\top,
    \]
    which implies that
    \[
        \E[X_{un}^2\mid \mathcal F] = p(1-p)\,\|\tilde z_{un}\|_2^2\,\tilde z_{un}\tilde z_{un}^\top.
    \]
    Therefore,
    \[
        \sum_{u, n}\E[X_{un}^2\mid \mathcal F] \preceq p(1-p) \left(\max_{u, n}\|\tilde z_{un}\|_2^2\right)\sum_{u, n}\tilde z_{un} \tilde z_{un}^\top = p(1-p)\left(\max_{u, n}\|\tilde z_{un}\|_2^2\right)I,
    \]
    so
    \[
        \left\|\sum_{u,n} \E[X_{un}^2\mid \mathcal F]\right\|_{\op}
        \le \frac{C_1}{U\wedge N}
    \]
    for a deterministic constant $C_1$.
    Let
    \[
        d:=\dim(\hat{\tilde G})=1+(U-1)+(N-1)+(U-2)+(N-1)=2U+2N-4.
    \]
    Then, conditional on $\mathcal F$, matrix Bernstein (Tropp, 2012) gives for any $t>0$,
    \[
        \Pr\left(
            \|\hat{\tilde G} - pI\|_{\mathrm{op}} \ge t \,\middle|\, \mathcal F
        \right)
        \le 2d \exp\left(
            \frac{-t^2/2}{v_{U,N} + R_{U,N}t/3}
        \right).
    \]
   Choosing
    \[
        t = C'\sqrt{\frac{\log d}{U\wedge N}}
    \]
    and taking $C'$ large enough yields
    \[
        \Pr\left(
            \|\hat{\tilde G} - pI\|_{\mathrm{op}} \ge C'\sqrt{\frac{\log d}{U\wedge N}} \,\middle|\, \mathcal F
        \right)
        \le 2d^{-2}.
    \]
    Taking expectation over $\mathcal F$ gives the same unconditional bound, hence
    \[
        \|\hat{\tilde G} - pI\|_{\mathrm{op}}
        = O_p\left(\sqrt{\frac{\log d}{U\wedge N}}\right)=o_p(1).
    \]
    For any vector $v$ we then have 
    \begin{align*}
        \left|v^\top (\hat G - p G) v\right| &=  \left|D^{1/2}v^\top (\hat{\tilde G} - p I) D^{1/2}v\right|\\
        &\leq \|\hat{\tilde G} - pI\|_\op v^\top D v \\
        &= o_p(1) v^\top G v.
    \end{align*}
    This implies that
    \[
        \sum_{u, n}\omega_{un} (x_{un}^{\mathrm{lin}})^2 = (p+o_p(1)) \sum_{u, n} (x_{un}^{\mathrm{lin}})^2.
    \]
    Moreover,
    \[
        \sum_{u, n}(x_{un}^{\mathrm{rem}})^2 = r^2 \|f^c\|_2^2 \|g^c\|_2^2 \le C\frac{N}{U}\|\delta f\|_2^4,
    \]
    and hence, by Cauchy-Schwarz,
    \[
        \left|\sum_{u, n} x_{un}^{\mathrm{lin}} x_{un}^{\mathrm{rem}}\right| \le \|X^{\mathrm{lin}}\|_F \|X^{\mathrm{rem}}\|_F.
    \]
    Therefore
    \[
        \sum_{u, n} x_{un}^2 \ge \frac{1}{2}\sum_{u, n}(x_{un}^{\mathrm{lin}})^2 - \sum_{u, n}(x_{un}^{\mathrm{rem}})^2,
    \]
    and similarly
    \[
        \sum_{u, n}\omega_{un} x_{un}^2 \ge \frac{1}{2}\sum_{u, n}\omega_{un}(x_{un}^{\mathrm{lin}})^2 - \sum_{u, n}\omega_{un}(x_{un}^{\mathrm{rem}})^2.
    \]
    Since $U \asymp N$ with probability tending to $1$, the remainder terms are bounded by $C\|\delta f\|_2^4$. Combining this with the lower bound on $\sum_{u,n} x_{un}^2$ established above, we obtain with high probability
    \[
        \sum_{u, n}\omega_{un} x_{un}^2 \geq c'\left(UN (\delta \mu)^2 + N\|\delta h\|_2^2 + U\|\delta i\|_2^2 + N \|\delta f\|_2^2 + U \|\delta g\|_2^2\right) - C\|\delta f\|_2^2\|\delta g\|_2^2 - C\|\delta f\|_2^4
    \]
    for some constants $c', C > 0$.

    Now, we turn to proving upper bounds. Recall that it remains to upper bound the terms
    \[
        \sum_{u, n}\omega_{un} y_{un}^2, \qquad \sum_{u, n}\omega_{un} \epsilon_{un} x_{un}, \qquad \sum_{u, n}\omega_{un} \epsilon_{un} y_{un}.
    \]
    Define $M(\delta) = UN (\delta \mu)^2 + N\|\delta h\|_2^2 + U\|\delta i\|_2^2 + N \|\delta f\|_2^2 + U \|\delta g\|_2^2$. Then, 
    \[
        \sum_{u, n}\omega_{un} y_{un}^2 \leq \sum_{u, n}y_{un}^2 = \sum_{u, n}(\delta f_u \delta g_n)^2 \leq \|\delta f\|_2^2\|\delta g\|_2^2 \leq \frac{M(\delta)^2}{UN}.
    \]
    Since $X$ is at most rank $5$, we have
    \[
        \left|\sum_{u, n}\omega_{un}\epsilon_{un}x_{un}\right| \leq \|\Omega\circ E\|_\op \cdot \|X\|_* \leq \sqrt{5}\|\Omega \circ E\|_\op\cdot \|X\|_F.
    \]
    Here, we slightly abuse notation and let $\Omega$ denote the matrix with entries $\omega_{un}$. Using the fact that $\|X\|_F^2 \lesssim M(\delta)$, Young's inequality, and sub-Gaussian errors from Assumption \ref{asmp:random-latents}, we have that for some constants $\eta$ and $C_\eta$,
    \[
        \sqrt{5}\|\Omega \circ E\|_\op \cdot \|X\|_F \leq \eta M(\delta) + C_\eta \|\Omega \circ E\|_\op^2 = \eta M(\delta) + O_p(U\vee N).
    \]
    Finally,
    \[
       \left| \sum_{u, n}\omega_{un}\epsilon_{un} y_{un} \right| \leq \|\Omega\circ E\|_\op \cdot \|\delta f\|_2 \cdot \|\delta g\|_2 \leq \frac{\|\Omega\circ E\|_\op}{2\sqrt{UN}}\left(N\|\delta f\|_2^2 + U\|\delta g\|_2^2\right) \leq \frac{\|\Omega\circ E\|_\op}{2\sqrt{UN}}M(\delta).
    \]

    Combining this all gives, for some constants $c', C'$
    \[
        \left(c' -\eta - \frac{\|\Omega\circ E\|_\op}{2\sqrt{UN}}\right) M(\delta) \leq C'\frac{M(\delta)^2}{UN} + O_p(U \vee N).
    \]
    By sub-Gaussian concentration, we have that $\|\Omega\circ E\|_\op = O_p(\sqrt{U \vee N})$. Thus, the bound simplifies to
    \begin{equation}\label{eq:m-delta-bound}
        (c'-\eta - o_p(1)) M(\delta) \leq C'\frac{M(\delta)^2}{UN} + O_p(U \vee N).
    \end{equation}
    By a simple argument, we can argue that $M(\delta) = o_p(UN)$. Fix $\epsilon > 0$. Take $\eta < c'/4$ so that $c' - \eta > c'/2$. Since $U\asymp N$ with probability tending to $1$, on the event $\mathcal A = \{M(\delta) \geq \epsilon N^2\}$ \eqref{eq:m-delta-bound} gives
    \[
        \frac{c'}{2}\epsilon N^2 \leq C'\epsilon^2 N^2 + O_p(N),
    \]
    which implies that
    \[
        \frac{c'}{2}\epsilon \leq C'\epsilon^2 + O_p(N^{-1}).
    \]
    If $\epsilon < c'/(4C')$, then $\epsilon^2 < c'/2 \epsilon$, implying that $M(\delta) = o_p(UN)$.

    Putting this all together, we have that
    \begin{align*}
        &UN (\delta \mu)^2 + N\|\delta h\|_2^2 + U\|\delta i\|_2^2 + N \|\delta f\|_2^2 + U \|\delta g\|_2^2 = O_p(U \vee N) \\ 
        &\implies (\delta \mu)^2 + \frac{1}{U}\|\delta h\|_2^2 + \frac{1}{N}\|\delta i\|_2^2 + \frac{1}{U}\|\delta f\|_2^2 + \frac{1}{N}\|\delta g\|_2^2 = O_p\left(\frac{1}{U\wedge N}\right)
    \end{align*}
    immediately implying the Frobenius bounds on the recovered matrix $\hat S$, as well as the bounds for each of the intercept terms.
\end{proof}

\begin{proof}[Proof of Theorem \ref{thm:ridge-recovery}]
    Let $\hat \theta = (\hat \mu, \hat h, \hat i, \hat f, \hat g)$ be a global optimizer of \eqref{eq:ridge-opt}, and $\tilde \theta = (\tilde \mu, \tilde h, \tilde i, \tilde f, \tilde g)$ be the centered but unscaled representation of the true parameters $\theta^0$ obtained from Lemma \ref{lem:canonical-centering} before the final normalization $U^{-1}\|f\|_2^2 = 1$. Let $\hat{\tilde \theta}$ be the corresponding quantity derived by centering $\hat \theta$. Notice that we can recover $\theta^c$ from $\tilde \theta$ since
    \[
        f^c = \frac{\tilde f}{a_U}, \qquad g^c = a_u \tilde g, \qquad a_u^2 = \frac{1}{U}\|\tilde f\|_2^2,
    \]
    and similarly for $\hat\theta^c$. Thus, it suffices to show the statement of the theorem replacing $\hat \theta^c$ and $\theta^c$ with $\hat{\tilde \theta}$ and $\tilde \theta$.

    Fix $u$ and let $\tilde{\alpha}_u = (\tilde{h}_u, \tilde{f}_u)^\top$, $\hat{\tilde{\alpha}}_u = (\hat{\tilde{h}}_u, \hat{\tilde{f}}_u)^\top$ and let $\tilde \nu = (\tilde \mu, \tilde i, \tilde g)$, $\hat{\tilde{\nu}} = (\hat{\tilde \mu}, \hat{\tilde i}, \hat{\tilde g})$. Denote the loss function in the objective \eqref{eq:ridge-opt} by 
    \[
        L(\theta) + G_\lambda(\theta), 
    \]
    where $G_\lambda(\theta)$ is the ridge penalty term as in the proof of Lemma \ref{lem:frobenius-norm-bound} and 
    \[
        L(\theta) = \sum_{(u, n)}\omega_{un}(r_{un} - \mu - h_u - i_n - f_ug_n)^2.
    \]
    For a parameter vector $\theta$, define the row-normalized score function to be 
    \[
        \Psi_{u, N}(\alpha_u; \nu) = -\frac{1}{2N}\begin{pmatrix}
            \frac{d}{dh_u} L(\theta) \\
            \frac{d}{df_u}L(\theta)
        \end{pmatrix} = \frac{1}{N}\sum_{n}\omega_{un} \begin{pmatrix}
            1\\ g_n
        \end{pmatrix}e_{un},
    \]
    where $\nu = (\mu, i, g)$, $e_{un}(\theta) = r_{un} - \mu - h_u - i_n - f_u g_n$. Also, define
    \[
        \hat A_u = \frac{1}{N}\sum_n \omega_{un} \begin{pmatrix}
            1 & \hat{\tilde{g}}_n \\
            \hat{\tilde g}_n & \hat{\tilde g}_n^2
        \end{pmatrix}.
    \]
    Then,
    \[
        \Psi_{u, N}(\hat{\tilde \alpha}_u; \hat{\tilde \nu}) - \Psi_{u, N}(\tilde \alpha_u; \hat{\tilde \nu}) = -\hat A_u(\hat{\tilde \alpha}_u - \tilde \alpha_u).
    \]
    In order show pointwise convergence in probability for $\hat{\tilde \alpha}_u$ to $\tilde \alpha_u$, it suffices to show that the terms on the left-hand side of the above are $o_p(1)$ and that $\hat A_u$ is asymptotically invertible.

    First, consider $\Psi_{u, N}(\hat{\tilde \alpha}_u; \hat{\tilde \nu})$. Since the fitted matrix given by $\hat \theta$ and $\tilde \theta$ are the same, we have
    \begin{align*}
        \Psi_{u, N}(\hat{\tilde \alpha}_u; \hat{\tilde \nu}) &= \frac{1}{N}\sum_n \omega_{un}\begin{pmatrix}
            1 \\ \hat{\tilde g}_n 
        \end{pmatrix}e_{un}(\hat\theta).
    \end{align*}
    The first order conditions on $\hat \theta$ imply that,
    \[
        \frac{1}{N}\sum_n \omega_{un} e_{un}(\hat\theta) = \frac{\lambda_h}{N} \hat h_u
    \]
    and by Lemma \ref{lem:regularization-to-zero} the right-hand side goes to $0$. For the second term, we have
    \begin{align*}
        \frac{1}{N}\sum_n \omega_{un} \hat{\tilde g}_n e_{un}(\hat\theta) &= \frac{1}{N}\sum_n \omega_{un} \hat g_n e_{un}(\hat\theta) - \bar{\hat g} \frac{1}{N}\sum_n \omega_{un} e_{un}(\hat\theta) \\
        &= \frac{\lambda_f}{N}\hat f_u - \bar{\hat g} \frac{\lambda_h}{N}\hat h_u \pto 0,
    \end{align*}
    where the convergence is again implied by Lemma \ref{lem:regularization-to-zero}.

    Next, consider $\Psi_{u, N}(\tilde \alpha_u; \tilde \nu)$. We have 
    \begin{equation}\label{eq:psi-tilde-nu}
        \Psi_{u, N}(\tilde \alpha_u; \hat{\tilde \nu}) = \frac{1}{N}\sum_n \omega_{un} \begin{pmatrix}
            1 \\ \hat{\tilde g}_n
        \end{pmatrix}\left(\epsilon_{un} + (\tilde \mu - \hat{\tilde \mu}) + (\tilde i_n - \hat{\tilde i}_n) + \tilde f_u(\tilde g_n - \hat{\tilde g}_n) \right).
    \end{equation}
    We show that the last term is $o_p(1)$; the rest can be proven analogously. First, by Cauchy-Schwarz
    \begin{align*}
        \left|\tilde f_u \frac{1}{N}\sum_n \omega_{un} (\tilde g_n - \hat{\tilde g}_n) \right| &\leq |\tilde f_u|\left(\frac{1}{N}\sum_n (\tilde g_n - \hat{\tilde g}_n)^2\right)^{1/2} = o_p(1)
    \end{align*}
    by Lemma \ref{lem:frobenius-norm-bound} and the fact that Assumption \ref{asmp:random-latents} implies $f_u$ is bounded so $\tilde f_u = O_p(1)$. Then, again by Cauchy-Schwarz
    \begin{align*}
        \left|\tilde f_u \frac{1}{N}\sum_n \omega_{un}\hat{\tilde g}_n (\tilde g_n - \hat{\tilde g}_n) \right| &\leq |\tilde f_u| \left(\frac{1}{N}\sum_n \hat{\tilde g}_n^2\right)^{1/2} \left(\frac{1}{N}\sum_n (\tilde g_n - \hat{\tilde g}_n)^2\right)^{1/2} = o_p(1).
    \end{align*}
    The last equality comes from Lemma \ref{lem:frobenius-norm-bound} and the fact that Lemma \ref{lem:frobenius-norm-bound} and Corollary \ref{cor:asmp-center} together imply 
    \[
        \frac{1}{N}\sum_n \hat{\tilde g}_n^2 \leq 2\frac{1}{N}\sum_n (\hat{\tilde g}_n - \tilde g_n)^2 + 2\frac{1}{N}\sum_n \tilde g_n^2 = O_p(1).
    \]
    The other terms in \eqref{eq:psi-tilde-nu} can be shown to be $o_p(1)$ analogously.

    Finally, we show that $\hat A_u$ is asymptotically invertible. Let
    \[
        \tilde A_u = \frac{1}{N}\sum_n \omega_{un}\begin{pmatrix}
            1 & \tilde g_n \\
            \tilde g_n & \tilde g_n^2
        \end{pmatrix}.
    \]
    Then, 
    \[
        \hat A_u - \tilde A_u = \frac{1}{N}\sum_n \omega_{un}\begin{pmatrix}
            0 & \hat{\tilde g}_n - \tilde g_n \\
            \hat{\tilde g}_n - \tilde g_n & \hat{\tilde g}_n^2 - \tilde g_n^2
        \end{pmatrix}.
    \]
    By similar reasoning as before, the off-diagonal terms of $\hat A_u - \tilde A_u$ are $o_p(1)$. Also,
    \begin{align*}
        \frac{1}{N}\sum_n |\hat{\tilde g}_n^2 - \tilde g_n^2| \leq \left(\frac{1}{N}\sum_n (\hat{\tilde g}_n - \tilde g_n)^2\right)^{1/2}\left(\frac{1}{N}\sum_n (\hat{\tilde g}_n + \tilde g_n)^2\right)^{1/2}
    \end{align*}
    The first term in the product is $o_p(1)$ and the second term can be bounded as follows:
    \[
        \frac{1}{N}\sum_n (\hat{\tilde g}_n + \tilde g_n)^2 \leq 2\frac{1}{N}\sum_n (\hat{\tilde g}_n - \tilde g_n)^2 + 8\frac{1}{N}\sum_n \tilde g_n^2 = O_p(1)
    \]
    by Lemma \ref{lem:frobenius-norm-bound} and Corollary \ref{cor:asmp-center}. Therefore, each entry in $\hat A_u - \tilde A_u = o_p(1)$. By Corollary \ref{cor:asmp-center}, 
    \[
        \tilde A_u \pto p \begin{pmatrix}
            1 & 0\\
            0 & \sigma_g^2
        \end{pmatrix},
    \]
    and since $\sigma_g^2 > 0$, the limit matrix is positive definite. Thus, with probability tending to one $\hat A_u$ is invertible and its inverse has bounded operator norm.

    Putting this all together, we have that
    \[
        \hat A_u(\hat{\tilde \alpha}_u - \tilde \alpha_u) = o_p(1), 
    \]
    and invertibility of $\hat A_u$ implies that $(\hat{\tilde h}_u, \hat{\tilde f}_u)\pto (\tilde h_u, \tilde f_u)$. The proof of the pointwise convergence for the note side factors follows analogously.
\end{proof}

To prove Theorem 1 from the main text, it now suffices to show that the solution to \eqref{eq:ridge-opt} with no missing data is consistent for the true parameters under the canonical decomposition. We will then be able to recover the original factors by de-centering the estimates using Lemma \ref{lem:canonical-centering}.

% The proof is similar to the framework in \cite{bai2009panel}, but we specialize to our case and write a proof here for completion.

\begin{proof}[Proof of Theorem 1]\label{pf:proof-thm-1}
    Theorem \ref{thm:ridge-recovery} implies that the canonical representation of the solution to \eqref{eq:ridge-opt} is consistent for the canonical representation of the true parameter values. We now show that $\hat i_n \pto i_n^0$. 

    First, Theorem \ref{thm:ridge-recovery} implies that for a fixed $n$, $\hat g_n^c\pto g_n^c = g_n^0 - \bar g_n^0 \one_N$. By weak law of large numbers, $\bar g_n^0 \pto 0$, so $\hat g_n^c \pto g_n^0$. 
    In particular, we have that $\hat i_n^c \pto i_n^c$. By Lemma \ref{lem:canonical-centering}, 
     \[
        i_n^c = i_n^0 - \bar i^0\one_N + \bar f^0(g_n^0 - \bar g^0 \one_N).
    \]
    with
    \[
        \bar i^0 = \frac{1}{N}\one_N^\top i^0, \qquad \bar f^0 = \frac{1}{U}\one_U^\top f^0.
    \]
    By Theorem \ref{thm:ridge-recovery}, for fixed $n$, $\hat i_n^c \pto i_n^c$. By weak law of large numbers and Assumption \ref{asmp:random-latents}, $i_n^c \pto i_n^0 + cg_n^0$.

    $\hat g_n^c \pto g_n$ coordinate-wise. Furthermore, by Lemma \ref{lem:canonical-centering} and weak law of large numbers, $g_n^c = g_n^0 + \bar g^0 \one_N \pto g_n^0$ since $\E[g_n] = 0$. Therefore, $\hat g_n \pto g_n^0$. This in turn implies, by weak law of large numbers, that 
    \[
        \hat i_n \pto i_n^0 + cg_n^0,
    \]
    where $c = \E[f_u]$. If $c$ is known, then we can recover a consistent estimate of $i_n^0$ by $\hat i_n^0 = \hat i_n + c\hat g_n$.
\end{proof}

% \begin{proof}[Proof of Theorem 1]\label{pf:proof-thm-1}
% Let \((\mu^c,h^c,i^c,f^c,g^c)\) denote the canonical representation of the true parameter vector, with normalization
% \[
%     \frac1U\|f^c\|_2^2=1.
% \]
% By Theorem \ref{thm:ridge-recovery}, for each fixed \(n\),
% \[
%     \hat i_n^c \pto i_n^c, \qquad \hat g_n^c \pto g_n^c.
% \]

% By Lemma \ref{lem:canonical-centering}, if we write
% \[
%     \bar i^0 = \frac1N\sum_{m=1}^N i_m^0 \qquad \bar f^0 = \frac1U\sum_{u=1}^U f_u^0, \qquad a_U^2 = \frac1U\sum_{u=1}^U (f_u^0-\bar f^0)^2,
% \]
% then the canonical and original parameters are related by
% \[
%     g_n^c = a_U(g_n^0-\bar g^0), \qquad i_n^c = i_n^0-\bar i^0+\frac{\bar f^0}{a_U}g_n^c.
% \]
% Equivalently,
% \[
%     i_n^0 = i_n^c+\bar i^0-\frac{\bar f^0}{a_U}g_n^c.
% \]

% Now let \(\hat{\bar i}\), \(\hat{\bar f}\), and \(\hat a_U\) be estimators such that
% \[
%     \hat{\bar i}\pto \bar i^0, \qquad \hat{\bar f}\pto \bar f^0, \qquad \hat a_U\pto a_U.
% \]
% Define the recovered original note intercept by
% \[
%     \hat i_n^0 = \hat i_n^c+\hat{\bar i}-\frac{\hat{\bar f}}{\hat a_U}\hat g_n^c.
% \]
% Then, by \((1)\), \((2)\), Theorem \ref{thm:ridge-recovery}, and Slutsky's theorem,
% \[
%     \hat i_n^0 = \hat i_n^c+\hat{\bar i}-\frac{\hat{\bar f}}{\hat a_U}\hat g_n^c \pto i_n^c+\bar i^0-\frac{\bar f^0}{a_U}g_n^c = i_n^0.
% \]
% Hence \(\hat i_n^0\pto i_n^0\).
% \end{proof}

\paragraph{Matrix-completion formulation and incoherence}
In this section, we describe a related approach for finding estimators for latent factors and intercept terms. We can first estimate the missing entries of the observed matrix using matrix completion techniques, and then use the canonical decomposition of the resulting matrix as an estimator for the latent factors. 

Define the nuclear-norm completion problem
\begin{equation}
\label{eq:si-nuclear}
\hat S_{\mathrm{mc}}
\in
\arg\min_{X\in\mathbb{R}^{U\times N}}
\frac12 \sum_{(u,n)\in\Omega} (X_{un}-R_{un})^2 + \lambda \|X\|_\ast.
\end{equation}
The following definition and lemmas verify the conditions needed to apply the matrix-completion result. As before, in the truthful regime users report
\[
r_{un} = s_{un} + \epsilon_{un}, \qquad s_{un} = \mu + h_u + i_n + f_u g_n,
\]
where \(\epsilon_{un}\) is an additive noise term. Let $S$ be the true signal matrix with entries $s_{un}$. We apply the framework of \cite{chen2020noisy} (see also \cite{farias2022uncertainty}) to get bounds on the entry-wise norm of the error of an estimator of $S$, $\| \widehat S_{\mathrm{mc}} - S\|_\infty$, and then show how this translates to pointwise consistency of the intercept terms. This framework requires bounds on the condition number, $\kappa(S)$, and that $S$ is $\nu$-incoherent, which we define next.

\begin{definition}[Incoherence]
    A rank-$r$ matrix $A \in \mathbb R^{U\times N}$ with Singular Value Decomposition $A = U\Sigma V^{\top}$ is said to be $\nu$-incoherent if 
    \[
        \|U\|_{2, \infty} \leq \sqrt{\frac{\nu}{U}}\|U\|_F = \sqrt{\frac{\nu r}{U}}, \qquad \text{and} \qquad \|V\|_{2, \infty} \leq \sqrt{\frac{\nu}{N}}\|V\|_F = \sqrt{\frac{\nu r}{N}}.
    \]
    Here $\|X\|_{2, \infty}$ denotes the max row  $\ell_2$-norm of all rows in the matrix $X$.
\end{definition}

\begin{lemma}\label{lem:s-incoherent}
Under Assumption~\ref{asmp:random-latents}, the signal matrix $S$ has rank at most $3$. Furthermore, with probability tending to one, $S$ is $\nu$-incoherent for some deterministic, finite $\nu$, and its singular values satisfy
\[
    \sigma_{\min}^+(S) \asymp \sqrt{UN}, \qquad \kappa(S) = O_p(1).
\]
Here $\sigma_{\min}^+(S)$ denotes the smallest nonzero singular value of $S$.
\end{lemma}
\begin{proof}
Let
\[
    L =
    \begin{bmatrix}
    \mathbf 1_U & h^c & f^c
    \end{bmatrix}
    \in \mathbb R^{U\times 3},
    \qquad
    R =
    \begin{bmatrix}
        \mu^c \mathbf 1_N + i^c & \mathbf 1_N & g^c
    \end{bmatrix}
    \in \mathbb R^{N\times 3}.
\]
Then $S = L R^{\top}$, so $\rank(S)\le 3$. Let $\ell_u^{\top}$ denote the $u$-th row of $L$. By boundedness and weak law of large numbers,
\[
    \frac{1}{U}L^{\top}L = \frac{1}{U}\sum_u \ell_u \ell_u^{\top} \xrightarrow{p} \Sigma_L := \mathbb E[\ell_1\ell_1^{\top}] =
    \begin{pmatrix}
    1 & 0 & 0 \\
    0 & (\sigma_h^c)^2 & 0 \\
    0 & 0 & (\sigma_f^c)^2
    \end{pmatrix},
\]
which is positive definite. Therefore, there exist constants \(0<c_L<C_L<\infty\) such that, with probability tending to one,
\[
c_L U I_3 \preceq L^{\top}L \preceq C_L U I_3.
\]
Let \(L=Q_L T_L\) be the thin QR decomposition. Since \(T_L^{\top}T_L=L^{\top}L\), $\sigma_{\min}(T_L)\ge \sqrt{c_L U}$ with probability tending to one. Since
% \[
%     \tilde h_u = h_u-\bar h+\bar g(f_u-\bar f),\qquad \tilde f_u=f_u-\bar f,
% \]
% and the 
the latent variables are uniformly bounded under Assumption~\ref{asmp:random-latents} there exists a deterministic constant $B_L<\infty$ such that with probability tending to one $\max_{1\le u\le U}\|\ell_u\|_2 \le B_L.$ 
Hence
\[
\|(Q_L)_{u\cdot}\|_2
\le
\|\ell_u\|_2 \|T_L^{-1}\|_{\mathrm{op}}
\le
B_L \,\sigma_{\min}(T_L)^{-1}
=
O_p(U^{-1/2}),
\]
uniformly in \(u\). Therefore $\|Q_L\|_{2,\infty}=O_p(U^{-1/2})$.

Now let \(r_n^{\top}\) denote the \(n\)-th row of \(R\). Again by Assumption \ref{asmp:random-latents} and the weak law of large numbers,
\[
    \frac{1}{N}R^{\top}R = \frac{1}{N}\sum_{n=1}^N r_n r_n^{\top} \xrightarrow{p} \Sigma_R := \mathbb E[r_1r_1^{\top}] =
    \begin{pmatrix}
    \mu^2+(\sigma_i^c)^2 & \mu & 0 \\
    \mu & 1 & 0 \\
    0 & 0 & (\sigma_g^c)^2
    \end{pmatrix}.
\]
The matrix \(\Sigma_R\) is positive definite because its upper-left \(2\times 2\) principal minor equals \((\sigma_i^c)^2>0\), and \((\sigma_g^c)^2>0\). Thus, for some constants \(0<c_R<C_R<\infty\),
\[
c_R N I_3 \preceq R^{\top}R \preceq C_R N I_3
\]
with probability tending to one.

If \(R=Q_R T_R\) is the thin QR decomposition, the same argument gives $\|Q_R\|_{2,\infty}=O_p(N^{-1/2})$. Finally, define \(K:=T_L T_R^{\top}\). Then $S = L R^{\top} = Q_L K Q_R^{\top}$. Let \(K=\widetilde U \Sigma \widetilde V^{\top}\) be the singular value decomposition of \(K\). Then $S = (Q_L\widetilde U)\Sigma(Q_R\widetilde V)^{\top}$ is a singular value decomposition of \(S\). Since right multiplication by an orthogonal matrix preserves row-wise Euclidean norms,
\[
    \|U_S\|_{2,\infty} \le \|Q_L\|_{2,\infty} = O_p(U^{-1/2}) \qquad
    \|V_S\|_{2,\infty} \le \|Q_R\|_{2,\infty} = O_p(N^{-1/2}),
\]
where \(U_S\) and \(V_S\) denote the left and right singular vector matrices of \(S\). Since \(\rank(S)\le 3\), this is exactly the desired incoherence bound.

Finally, the nonzero singular values of $S$ are exactly the singular values of $K$ so $\kappa(S)=\kappa(K)$. By submultiplicativity of the spectral condition number,
\[
    \kappa(K) = \kappa(T_LT_R^\top) \le \kappa(T_L)\kappa(T_R).
\]
On the high-probability event above,
\[
\kappa(T_L)
\le
\sqrt{\frac{C_L}{c_L}},
\qquad
\kappa(T_R)
\le
\sqrt{\frac{C_R}{c_R}},
\]
so
\[
\kappa(S)\le \sqrt{\frac{C_LC_R}{c_Lc_R}}=O(1)
\]
with probability tending to one. Moreover, the nonzero singular values of \(S\) are exactly the singular values of \(K=T_LT_R^\top\). On the same high-probability event, $\sigma_{\min}(T_L)\ge \sqrt{c_LU}, \sigma_{\min}(T_R)\ge \sqrt{c_RN}.$ Since \(T_L,T_R\in\mathbb R^{3\times 3}\), we may apply the bound $\sigma_{\min}(AB)\ge \sigma_{\min}(A)\sigma_{\min}(B)$ to obtain
\[
    \sigma_{\min}(K) \ge \sigma_{\min}(T_L)\sigma_{\min}(T_R) \ge \sqrt{c_Lc_R}\,\sqrt{UN}.
\]
Therefore the smallest nonzero singular value of \(S\) satisfies
\[
    \sigma_{\min}^+(S)=\sigma_{\min}(K)=\Omega(\sqrt{UN})
\]
with probability tending to one. Combined with the bound \(\kappa(S)=O(1)\), this also implies
\[
    \sigma_{\min}^+(S)=\Theta(\sqrt{UN}), \qquad \sigma_{\max}(S)=\Theta(\sqrt{UN})
\]
with probability tending to one. This completes the proof.
\end{proof}

We now apply Theorem~1 from \cite{chen2020noisy}, which we bring next for completion:
\begin{theorem}[Theorem 1 from \cite{chen2020noisy}]\label{thm:mat-completion}
    Assume that $S$ is $\nu$-incoherent with condition number $\kappa$, such that $\kappa = O(1)$. Assume that $U\asymp N$. Let $\lambda = C_\lambda \sigma \sqrt{(U \vee N)p}$ be the regularizer in \eqref{eq:si-nuclear} for some large constant $C_\lambda$. Then, with probability at least $1 - O((U\vee N)^{-3})$, any minimizer $\hat S$ of \eqref{eq:si-nuclear} satisfies
    \[
        \|\hat S - S\|_\infty \lesssim \sqrt{\kappa^3 \nu r}\frac{\sigma}{\sigma_{\min}} \sqrt{\frac{\nu\log(U \vee N)}{p}}\|S\|_\infty.
    \]
\end{theorem}
Since \(\rank(S)\le 3\), \(S\) is \(\nu\)-incoherent with \(\nu=O(1)\) and  \(\kappa(S)=O_p(1)\)  by Lemma~\ref{lem:s-incoherent}, and
\[
    \|S\|_\infty \le |\mu| + B_h + B_i + B_f B_g <\infty,
\]
Theorem~\ref{thm:mat-completion} yields
\begin{equation}\label{eq:rate-bound}
\|\widehat S_{\mathrm{mc}}-S\|_\infty
=
O_p\!\left(
\sqrt{\frac{\log(U\vee N)}{U\vee N}}
\right).
\end{equation}

\paragraph{Proof of Theorem 1 (Matrix Completion)}
Now, we can prove an analogous version of Theorem 1 from the main text under Assumption~\ref{asmp:random-latents}, using the nuclear norm estimator.

\begin{proof}[Proof of Theorem 1 (Matrix Completion)]
Define the matrix-decoded note estimator from \(\widehat S_{\mathrm{mc}}\) by
\[
    \hat\mu_{\mathrm{mc}} = \frac{1}{UN}\mathbf 1_U^{\top}\widehat S_{\mathrm{mc}}\mathbf 1_N,
    \qquad
    \hat i_{\mathrm{mc}} := \frac{1}{U}\widehat S_{\mathrm{mc}}^{\top}\mathbf 1_U - \widehat\mu_{\mathrm{mc}}\mathbf 1_N.
\]
Comparing \(\hat i_{\mathrm{mc}}\) to \(i^c\) we have,
\begin{align*}
    \|\hat i_{\mathrm{mc}} - i^c\|_\infty &= \left\| \frac{1}{U}(\widehat S_{\mathrm{mc}}-S)^{\top}\mathbf 1_U -\left(\frac{1}{UN}\mathbf 1_U^{\top}(\widehat S_{\mathrm{mc}}-S)\mathbf 1_N\right)\mathbf 1_N\right\|_\infty\\
    &\le
    \left\| \frac{1}{U}(\widehat S_{\mathrm{mc}}-S)^{\top}\mathbf 1_U\right\|_\infty+\left|\frac{1}{UN}\mathbf 1_U^{\top}(\widehat S_{\mathrm{mc}}-S)\mathbf 1_N\right|\\
    &\le 2\|\widehat S_{\mathrm{mc}}-S\|_\infty.
\end{align*}
Using \eqref{eq:rate-bound}, this gives
\[
    \|\hat i_{\mathrm{mc}} - i^c \|_\infty = O_p(\alpha_{U,N}), \qquad \alpha_{U,N} = \sqrt{\frac{\log(U\vee N)}{U\vee N}}.
\]
Therefore, for each fixed \(n\),
\[
    |\hat i_{\mathrm{mc}, n}- i_n^c| = O_p(\alpha_{U,N}) \to 0.
\]

\begin{remark}
    In the same way as in the proof of Theorem 1, by de-centering the estimated $\hat i_{\mathrm{mc}, n}$, we can recover an estimate for the original note helpfulness, $i_n^0$. 
\end{remark}

\end{proof}

\subsection{Proof of the Conformity Regime (Theorem 2)} 

We now turn to the case \(\rho(\cdot) < 1\), where contributors partially align their reports with the anticipated platform consensus. In this regime, the platform no longer observes the latent signal matrix \(S\), but rather a conformity-distorted matrix. The proof parallels the private-signal case, but with a different target matrix and a different limiting parameter.  
Here, we restate the Theorem 2 from the main text to be proven.

\begin{theorem}[Conformity Regime]\label{thm:conformity}
    Assume that $\rho \not \equiv 1$, i.e. there exists at least one $n$ such that $\rho_n < 1$. Let $\hat \theta = (\hat \mu, \hat h, \hat i, \hat f, \hat g)$ denote the canonically chosen estimator from either \eqref{eq:ridge-opt} and let $\theta = (\mu, h, i, f, g)$ denote the canonically centered true parameters. Under Assumptions \ref{asmp:random-latents}-\ref{ass:conformity}, 
    \[
        \hat i_n \pto i_n^\infty
    \]
    with $i_n^\infty \neq i_n$ for at least one $n$.
\end{theorem}

In the case that $\rho\not \equiv 1$, a users'  best action is given by
\[
    a_{un} + \epsilon_{un} = \rho(c_n) s_{un} + (1-\rho(c_n))m_n + \epsilon_{un},
\]
where $\epsilon_{un}$ satisfies the condition in Assumption \ref{ass:conformity}. For notational simplicity, in this section we denote $\rho_n = \rho(c_n)$. We can use the techniques from Section \ref{proof:thm1} to show that in this case, the platform's estimate for the note helpfulness does not allow for the recovery of the true note helpfulness.

\begin{proof}[Proof of Theorem \ref{thm:conformity}]
    The natural target under conformity is not the original note effect \(i_n\), but the note intercept induced by the approximation to the distorted matrix \(A\) with entries $a_{un}$. 

    Under misspecification, the arguments in Section \ref{proof:thm1} continue to apply with the true parameter replaced by the canonical centered pseudo-true parameters, i.e., the projections of $A$ recovered from Lemma \ref{lem:canonical-centering}. The only change is that the residual is $\zeta_{un}=\epsilon_{un}+(a_{un}-s^*_{un})$ rather than $\epsilon_{un}$; the additional approximation term is absorbed by the projection property of the pseudo-true optimizer, so all rates in Lemma 2 and the subsequent score expansion argument goes through with minor modification.
    
    Thus, it suffices to understand the limiting behavior of $\hat \mu, \hat h_u, \hat i_n, \hat f_u, \hat g_n$ when they solve \eqref{eq:ridge-opt} with full data. We show that $\hat i$ is a biased estimate of $i$. Define $\delta_n = m_n - (\mu + i_n)$. We have that
    \begin{align*}
        \hat \mu
        &= \frac{1}{UN} \sum_{u, n}\rho_n (\mu + h_u + i_n + f_u g_n) + (1-\rho_n)m_n \\
        &= \frac{1}{UN} \sum_{u, n} \mu + i_n + \rho_n(h_u + f_ug_n) + (1-\rho_n)\delta_n \\
        &= \mu + \frac{1}{N}\sum_n i_n + \left(\frac1U\sum_u h_u\right)\left(\frac1N\sum_n \rho_n\right)  + \left(\frac1U\sum_u f_u\right)\left(\frac1N\sum_n \rho_n g_n\right) + \frac1N\sum_n (1-\rho_n)\delta_n.
    \end{align*}
    Using Assumption \ref{asmp:random-latents}-\ref{ass:conformity} and weak law of large numbers, we have that 
    \[
        \hat \mu \pto \mu + \E[\delta_1](1-\bar\rho).
    \]
    Then, for fixed \(n\), we have
    \begin{align*}
        \hat i_n
        &= \frac{1}{U}\sum_u \left(\mu + i_n + \rho_n(h_u + f_u g_n) + (1-\rho_n)\delta_n \right) - \hat \mu \\
        &= \mu + i_n
        + \rho_n\Bigl(\frac1U\sum_u h_u\Bigr)
        + \rho_n g_n\Bigl(\frac1U\sum_u f_u\Bigr)
        + (1-\rho_n)\delta_n
        - \hat \mu.
    \end{align*}
    By the weak law of large numbers and Assumption \ref{asmp:random-latents},
    \[
        \frac1U\sum_u h_u \xrightarrow{p} \E[h_u]=0,
        \qquad
        \frac1U\sum_u f_u \xrightarrow{p} \E[f_u]=c.
    \]
    Since \(\rho_n\) and \(g_n\) do not depend on \(u\), it follows by Slutsky's theorem that
    \[
        \rho_n\Bigl(\frac1U\sum_u h_u\Bigr)\xrightarrow{p}0,
        \qquad
        \rho_n g_n\Bigl(\frac1U\sum_u f_u\Bigr)\xrightarrow{p} c\rho_n g_n.
    \]
    Therefore,
    % \[
    %     \hat i_n \pto \rho_n(i_n + c g_n) + (\rho_n - \E[\rho_1])\mu + (1-\rho_n)m_n - \E[1-\rho_1]\E[m_1].
    % \]
    \[
        \hat i_n \pto i_n + c\rho_n g_n + (1-\rho_n)\delta_n - (1-\bar \rho)\E[\delta_1].
    \]
    In this setting, the platform cannot recover an estimate of the true $i_n$ without additionally knowing $\delta_n$ and $\rho_n$, hence generally $i_n^\infty \neq i_n$.
\end{proof}

\subsection{User-Factor Distortion and Minority Compression}
Now, we turn to understanding when the factor estimates become biased in the case where users do not report truthfully. In this section, we again assume that all estimated factors are chosen under the canonical decomposition, and at the end of the section we discuss what the implications are for the true factors. By Theorem \ref{thm:ridge-recovery} we may proceed as if we observe full data. The first order conditions on $f_u$ and $g_n$ in the full data case give

\begin{equation}
\hat f_u \;=\; \frac{\sum_n y_{un}\,g_n}{\sum_n g_n^2},
\qquad
\hat g_n \;=\; \frac{\sum_u y_{un}\,f_u}{\sum_u f_u^2},
\label{eq:normal-eq}
\end{equation}
where $y_{un} = a_{un} - (\hat \mu + \hat h_u+ \hat i_n)$ is the estimated rank-1 residual.

\begin{theorem}
\label{thm:user-normal}
Let $\rho \not \equiv 1$. Consider the setting from Theorem \ref{thm:conformity}, and suppose that the true note factors $\{g_n\}_n$ are known. 
Let $\hat \mu, \hat h_u, \hat i_n, \hat f_u$ denote the canonically centered solution to \eqref{eq:ridge-opt}. Then
\[
    \E[\hat f_u \mid f_u, g_n, c_n] = w_1 f_u - w_1 c + o_p(1), \qquad w_1 = \frac{\sum_n \rho_n g_n^2}{\sum_n g_n^2}.
\]
\end{theorem}

\begin{proof}
    From \eqref{eq:ridge-opt}, the first order conditions for $f_u$ gives
    \begin{equation}
        \hat f_u \;=\; \frac{\sum_ny_{un}\,g_n}{\sum_n g_n^2 + \lambda_f},
    \label{eq:normal-eq}
    \end{equation}
    where $y_{un} = a_{un} - (\hat \mu + \hat h_u+ \hat i_n)$ is the estimated residual.

    First, we compute $\E[y_{un} \mid f_u, g_n, c_n]$. Recall that the first order conditions on $\hat h_u$ and $\hat i_n$ imply that 
    \[
        \hat h_u = \frac{1}{N}\sum_n a_{un} - \hat \mu, \qquad \hat i_n = \frac{1}{U}\sum_u a_{un} - \hat \mu,
    \]
    with $\hat \mu$ defined as in the proof of Theorem \ref{thm:conformity}.  Then,
    \begin{align*}
        \E[y_{un} \mid f_u, g_n, c_n] &= \E[a_{un} - (\hat \mu + \hat h_u + \hat i_n) \mid f_u, g_n, c_n] \\
        &= \E\left[a_{un} - \left(\frac{1}{N}\sum_n a_{un} + \frac{1}{U}\sum_u a_{un} - \hat \mu\right)\mid f_u, g_n, c_n\right].
    \end{align*}
    We can compute each term separately.
    \begin{align*}
        \E[a_{un} \mid f_u, g_n, c_n] &= \E[\rho_n(\mu + h_u + i_n + f_u g_n) + (1-\rho_n) m_n \mid f_u, g_n, c_n] \\
        &= \mu \rho_n + \E[m_n](1-\rho_n) + \rho_n f_u g_n.
    \end{align*}
    Then,
    \begin{align*}
        \frac{1}{N}\sum_n \E[a_{un} \mid f_u, g_n, c_n] &= \frac{1}{N}\sum_n \mu \rho_n + \E[m_n](1-\rho_n) + \rho_n f_u g_n \\
        &= \mu \bar \rho + \E[m_1](1-\bar \rho) + f_u \frac{1}{N}\sum_n g_n \rho_n,
    \end{align*}
    \begin{align*}
        \frac{1}{U}\sum_u \E[a_{un} \mid f_u, g_n, c_n] &= \frac{1}{U}\sum_u \mu \rho_n + \E[m_n](1-\rho_n) + \rho_nf_u g_n \\
        &= \mu \rho_n  + \E[m_1](1-\rho_n) + g_n \rho_n\frac{1}{U}\sum_u f_u,
    \end{align*}
    \begin{align*}
        \frac{1}{UN}\sum_u \E[a_{un} \mid f_u, g_n, c_n] &= \frac{1}{UN}\sum_{u, n} \mu \rho_n + \E[m_n](1-\rho_n) + \rho_n f_u g_n \\
        &= \mu \bar \rho + \E[m_1](1-\bar \rho) + \frac{1}{U}\sum_u f_u \frac{1}{N}\sum_n g_n\rho_n.
    \end{align*}
    Therefore, 
    \begin{align*}
        \E[y_{un} \mid f_u, g_n, c_n] &= \mu \rho_n + \E[m_n](1-\rho_n)+ \rho_nf_u g_n \\
        &\qquad - \left(\mu \bar \rho + \E[m_1](1-\bar \rho) + f_u \frac{1}{N}\sum_n g_n \rho_n\right) \\
        &\qquad - \left(\mu \rho_n + \E[m_1](1-\rho_n) + g_n \rho_n\frac{1}{U}\sum_u f_u\right) \\
        &\qquad + \left(\mu \bar \rho + \E[m_1](1-\bar \rho) + \frac{1}{U}\sum_u f_u \frac{1}{N}\sum_n g_n \rho_n\right) \\
        &= \rho_nf_u g_n - f_u \frac{1}{N}\sum_n g_n \rho_n - g_n \rho_n\frac{1}{U}\sum_u f_u + \frac{1}{U}\sum_u f_u \frac{1}{N}\sum_n g_n \rho_n.
    \end{align*}
    By weak law of large numbers and Assumptions \ref{asmp:random-latents}-\ref{ass:conformity}, 
    \[
         \E[y_{un} \mid f_u, g_n, c_n] = \rho_nf_u g_n - cg_n \rho_n + o_p(1).
    \]
    Since $\lambda_f = o(N^{-1/2})$, we have that 
    \begin{align*}
        \E\left[\hat f_u \mid f_u, g_n, c_n\right] &= \E\left[\frac{\sum_n y_{un} g_n}{\sum_n g_n^2 + \lambda_f}\mid f_u, g_n, c_n\right] \\
        &= \frac{\sum_n f_u \rho_n g_n^2 - cg_n^2 \rho_n}{\sum_n g_n^2} + o_p(1) \\
        &= f_u \frac{\sum_n \rho_n g_n^2}{\sum_n g_n^2} - c\frac{\sum_n \rho_n g_n^2}{\sum_n g_n^2} + o_p(1)
    \end{align*}

\end{proof}

Notice that Theorem \ref{thm:user-normal} is a statement about $\hat f_u$, which represent the centered user factor estimates. To recover estimates for the original factors, the platform can de-center the distribution of $\hat f_u$ according to Lemma \ref{lem:canonical-centering}. In particular, given knowledge of the true mean $\E[f_u] =c $, the estimate for the original factor can be recovered as $\hat f_u^0 = \hat f_u + c$. Theorem \ref{thm:user-normal} implies the following two propositions about $\hat f_u^0$.
\begin{proposition}
\label{prop:flip}
Consider the setting of Theorem \ref{thm:user-normal}. Then, 
\[
    \E[\hat f_u^0 \mid f_u, g_n, c_n] > 0 \qquad \text{if and only if} \qquad f_u > -\frac{c(1-w_1)}{w_1}.
\]
In particular, the minority slice $(-\frac{c(1-w_1)}{w_1}, 0)$ is mapped to positive estimates.
\end{proposition}
\begin{proof}
Since $\hat f_u^0 = \hat f_u + c$, Theorem \ref{thm:user-normal} implies
\[
    \E[\hat f_u^0 \mid f_u, \{g_n,\rho_n\}]
    = w_1 f_u - w_2 + c + o_p(1)
    = w_1 f_u + c(1-w_1) + o_p(1).
\]
This immediately implies the proposition.
\end{proof}

\begin{proposition}
\label{prop:shrink}
Let $F$ be the CDF of $f_u$. Assume that $F$ is continuous with no atom at $0$. Then, the estimated minority share
\[
    \pi^-_{\mathrm{est}} := \Pr(\hat f_u^0<0 \mid f_u, \{g_n,\rho_n\})
\]
satisfies
\[
    \pi^-_{\mathrm{est}}
    = F\left(\frac{w_2-c}{w_1}\right) + o(1)
    = F\left(-\frac{c(1-w_1)}{w_1}\right) + o(1)
    \le F(0)+o(1)
    = \pi^-_{\mathrm{true}}+o(1).
\]
If $0<w_1<1$, then the inequality is strict. Moreover, $\pi^-_{\mathrm{est}}$ is weakly increasing in $w_1$.
\end{proposition}

\begin{proof}
From Theorem \ref{thm:user-normal},
\[
    \E[\hat f_u^0 \mid f_u,\{g_n,\rho_n\}]
    = w_1 f_u - w_2 + c + o(1).
\]
Hence
\[
    \pi^-_{\mathrm{est}}
    = \Pr(\hat f_u^0<0 \mid \{g_n,\rho_n\})
    = \Pr\!\left(w_1 f_u - w_2 + c < 0\right) + o(1)
    = \Pr\!\left(f_u < \frac{w_2-c}{w_1}\right) + o(1)
    = F\!\left(\frac{w_2-c}{w_1}\right) + o(1).
\]
Since $w_2=cw_1$ and $c>0$,
\[
    \frac{w_2-c}{w_1}
    = -\frac{c(1-w_1)}{w_1}
    \le 0,
\]
with strict inequality when $0<w_1<1$. Therefore
\[
    \pi^-_{\mathrm{est}}
    \le F(0)+o(1)
    = \pi^-_{\mathrm{true}}+o(1),
\]
and strict inequality holds when $0<w_1<1$ because $F$ has no atom at $0$. Finally,
\[
    \frac{d}{dw_1}\left(\frac{w_2-c}{w_1}\right)
    = \frac{d}{dw_1}\left(-\frac{c(1-w_1)}{w_1}\right)
    = \frac{c}{w_1^2} > 0,
\]
so $\pi^-_{\mathrm{est}}$ is weakly increasing in $w_1$.
\end{proof}

\section{Statistical Guarantee for the Two-Stage Estimator}
We finally analyze the two-stage weighted estimator introduced in the main text. The argument proceeds in the same order as feasible generalized least squares: first analyze the inverse-variance rule, then show that estimated residual variances are consistent, and finally conclude that the feasible two-stage estimator attains the same asymptotic variance.

For bounded positive weights \(w=\{w_u\}\), recall the weighted regularized rank-1 problem
\begin{equation}\label{eq:weighted-opt}
    \arg\min_{\mu, h, i, f, g} \sum_{(u, n)\in \Omega}w_u \left(r_{un} - (\mu + h_u +i_n + f_ug_n)\right)^2 + \lambda_h\|h\|_2^2 + \lambda_i\|i\|_2^2 + \lambda_g\|g\|_2^2 + \lambda_f\|f\|_2^2.
\end{equation}
In this section, we use the nuclear norm problem to complete the matrix in both the first and second stages of the analysis before proceeding with ridge regression. Although we have shown that the nuclear norm and the direct solution to the ridge regularized problem are both consistent for the true canonical parameters, the sharper bounds on the nuclear norm problem allow for a more refined analysis. We first show that our estimates $\hat\sigma_u^2$ are consistent for the true variance of user residuals $\sigma_u^2$. 

\begin{lemma}\label{lem:sigma-consistency}
    Let $\hat\mu, \hat h_u, \hat i_n, \hat f_u, \hat g_n$ be the first stage estimates for the latent factors. Define
    \begin{equation}
    \label{eq:first-stage-variance}
    \hat\sigma_u^2 = \frac{1}{N_u} \sum_{n\in S_u} \left(r_{un}-\hat\mu-\hat h_u-\hat i_n-\hat f_u\hat g_n\right)^2,
    \end{equation}
    where $S_u = |\{n: (u, n) \in \Omega\}|$. Then, 
    \[
        \max_u|\hat\sigma_u^2 - \sigma_u^2| = O_p\left(\sqrt{\frac{\log (U \vee N)}{U \vee N}}\right).
    \]
\end{lemma}
\begin{proof}
     As before, $s_{un} = \mu + h_u + i_n + f_ug_n$ be the true signal for user $u$ on note $n$. Let the matrix $\hat S$ with entries $\hat s_{un}$ denote the matrix estimator for $S$ given by solving \eqref{eq:si-nuclear}. Then,
    \begin{align*}
        \hat \sigma_u^2 &= \frac{1}{N_u}\sum_{n \in \Omega_u}(r_{un} - \hat s_{un})^2 \\
        &= \frac{1}{N_u}\sum_{n \in \Omega_u}(\epsilon_{un} - (\hat s_{un} - s_{un}))^2 \\
        &= \frac{1}{N_u}\sum_{n \in \Omega_u}\epsilon_{un}^2 - 2\epsilon_{un}(\hat s_{un} - s_{un}) + (\hat s_{un} - s_{un})^2.
    \end{align*}
    Note that the first term above converges to $\sigma_u^2$ at rate $O_p(N^{-1/2})$ by law of large numbers, the MCAR sampling scheme, sub-Gaussian concentration and Slutsky's theorem. For the second term, by Cauchy-Schwarz we have
    \[
        \sum_{n \in \Omega_u} \epsilon_{un}(\hat s_{un} - s_{un}) \leq \left(\sum_{n \in \Omega_u}\epsilon_{un}^2\right)^{1/2}\left(\sum_{n \in \Omega_u}(\hat s_{un} - s_{un})^2\right)^{1/2}.
    \]
    We have
    \begin{equation}\label{eq:s-sq-bound}
        \sum_{n \in \Omega_u}(\hat s_{un} - s_{un})^2 \leq N_u \max_n (\hat s_{un} - s_{un})^2 = N_u O_p\left(\frac{\log (U \vee N)}{U \vee N}\right).
    \end{equation}
    By sub-Gaussian concentration,
    \[
        \max_{1\leq u \leq U}\left|\frac{1}{N_u}\sum_{n \in S_u} \epsilon_{un}^2 - \sigma_u^2\right| = O_p\left(\sqrt{\frac{\log U}{N}}\right).
    \]
    By \eqref{eq:s-sq-bound}, the last term is $O_p\left(\frac{\log (U \vee N)}{U \vee N}\right)$. Since all bounds are uniform bounds over $u$, 
    \begin{equation}\label{eq:sigma-bound}
        \max_u|\hat\sigma_u^2 - \sigma_u^2| = O_p\left(\sqrt{\frac{\log (U \vee N)}{U \vee N}}\right).
    \end{equation}

\end{proof}

\begin{theorem}
    Assume that $\rho \equiv 1$ and that $\E[f_u] = c$ is known. Let $\hat\mu^\ts, \hat h^\ts, \hat i^\ts, \hat f^\ts, \hat g^\ts$ be the solutions to \eqref{eq:weighted-opt} with weights $w_{u} = 1/\sigma_u^2$, where $\sigma_u^2 = \E[\epsilon_{un}^2]$. Then $\hat i^\ts$ is consistent for the canonical centered representation $i$.
    Moreover, among all solutions of \eqref{eq:weighted-opt} obtained using weights positive, finite weights $w_u \in (0, \infty)$, the estimator $\hat i^\ts$ has the lowest asymptotic variance; that is, for any other solution $\hat i$ of \eqref{eq:weighted-opt}$,$
    \begin{equation}\label{eq:efficiency}
        \mathrm{avar}(\hat i^\ts)\preceq \mathrm{avar}(\hat i).
    \end{equation}
    Furthermore, let $\tilde i$ be the solution to \eqref{eq:weighted-opt} with feasible weights $\hat w_u = 1/\hat\sigma_u^2$. Then,
    \begin{equation}\label{eq:feasible-efficiency}
        \avar(\tilde i) = \avar(\hat i^\ts).
    \end{equation}
\end{theorem}

\begin{proof}
Let \(T = W^{1/2}S\), where \(S\) is the true signal matrix and \(W=\mathrm{diag}(w_1,\dots,w_U)\). Let \(\widehat T = W^{1/2}\widehat S\) be the fitted matrix obtained by solving \eqref{eq:weighted-opt}. Since the weights are known, they can be factored out, so it suffices to study the transformed signal \(T\).

We first show that \(T\) satisfies the conditions needed to apply \eqref{eq:rate-bound}. By standard singular value inequalities,
\[
    \sigma_{\min}(W^{1/2})\sigma_j(S)
    \le \sigma_j(W^{1/2}S)
    \le \sigma_{\max}(W^{1/2})\sigma_j(S).
\]
Since the weights are bounded away from \(0\) and \(\infty\), Lemma \ref{lem:s-incoherent} implies that
\[
    \kappa(T)=\kappa(W^{1/2}S)\le \kappa(W^{1/2})\kappa(S)=O_p(1),
\]
and that $\sigma_{\min}^+(T)=\sigma_{\min}^+(W^{1/2}S)\asymp_p \sqrt{UN}.$ Let \(T=U_T\Sigma_TV_T^\top\) and \(S=U\Sigma V^\top\) be singular value decompositions. Since $T=W^{1/2}S=W^{1/2}U\Sigma V^\top$ and \(W^{1/2}\) is invertible, \(T\) and \(S\) have the same right singular subspace. Hence there exists an orthogonal matrix \(Q\) such that \(V_T=VQ\), and therefore
\[
    \|e_n^\top V_T\|_2=\|e_n^\top VQ\|_2=\|e_n^\top V\|_2=O_p(N^{-1/2})
\]
by Lemma \ref{lem:s-incoherent}.

For the left singular subspace, note that \(\mathrm{span}(U_T)=\mathrm{span}(W^{1/2}U)\). An orthonormal basis is given by
\[
    U_T=W^{1/2}U(U^\top WU)^{-1/2}.
\]
Since the weights are bounded away from \(0\) and \(\infty\), $U^\top WU \succeq w_{\min}I$ so $(U^\top WU)^{-1}\preceq w_{\min}^{-1}I.$ Therefore
\[
    \|e_u^\top U_T\|_2
    = \|e_u^\top \sqrt{w_u}\,U(U^\top WU)^{-1/2}\|_2
    \le \frac{\sqrt{w_u}}{\sqrt{w_{\min}}}\|e_u^\top U\|_2
    = O_p(U^{-1/2}),
\]
where the last step again uses Lemma \ref{lem:s-incoherent}. Thus \(T\) is \(\nu\)-incoherent with high probability. Combining incoherence, bounded condition number, and the growth of \(\sigma_{\min}^+(T)\), we obtain
\begin{equation}\label{eq:weighted-rate-bound-proof}
    \|\widehat T-T\|_\infty
    = O_p\!\left(\sqrt{\frac{\log(U\vee N)}{U\vee N}}\right).
\end{equation}

Since \(W^{1/2}\) is invertible, $ \widehat S = W^{-1/2}\widehat T.$ Define the weighted decoded estimators
\[
    \hat\mu
    = \frac{1}{N\sum_u w_u}\one_U^\top W\widehat S\one_N,
    \qquad
    \hat i
    = \frac{1}{\sum_u w_u}\widehat S^\top W\one_U - \hat\mu\,\one_N.
\]
By the same argument as in the proof of Theorem \ref{thm:mat-completion}, using \eqref{eq:weighted-rate-bound-proof} and the fact the weights are bounded and bounded away from $0$, the decoded estimators are consistent, and asymptotically coincide with the minimizers of \eqref{eq:weighted-opt}. In particular, when \(\rho\equiv 1\), \(\hat i\) is consistent for \(i\). Just as in the proof of Theorem 1 from the main text in Section \ref{proof:thm1}, we can de-center the estimates $\hat i_n$ to get consistent estimates of the original note parameter $i_n^0$.

Next, we compare asymptotic variances across admissible weights. Let $R=S+\Xi,$ where \(\Xi=(\epsilon_{un})\), and let
\[
    N_w:=\sum_u w_u,
    \qquad
    M_N:=I_N-\frac{1}{N}\one_N\one_N^\top.
\]
Then
\begin{align*}
    \hat i
    &= \frac{1}{N_w}\widehat S^\top W\one_U - \hat\mu\,\one_N \\
    &= \frac{1}{N_w}\widehat T^\top W^{1/2}\one_U - \hat\mu\,\one_N  \\
    &= \frac{1}{N_w}M_N\widehat T^\top W^{1/2}\one_U.
\end{align*}
Using \eqref{eq:weighted-rate-bound-proof}, we may replace \(\widehat T\) by \(T=W^{1/2}R\) in a first-order expansion, yielding
\[
    \hat i
    = \frac{1}{N_w}M_N T^\top W^{1/2}\one_U + o_p(1).
\]
Substituting \(T=W^{1/2}R=W^{1/2}S+W^{1/2}\Xi\), we obtain
\begin{align*}
    \hat i
    &= \frac{1}{N_w}M_N(W^{1/2}S)^\top W^{1/2}\one_U
    + \frac{1}{N_w}M_N(W^{1/2}\Xi)^\top W^{1/2}\one_U
    + o_p(1) \\
    &= \frac{1}{N_w}M_N S^\top W\one_U
    + \frac{1}{N_w}M_N \Xi^\top W\one_U
    + o_p(1).
\end{align*}
The first term is deterministic conditional on the latent variables, so the asymptotic variance is determined by the second term. Since the noise variables are independent across \(u\) and \(n\), with \(\var(\epsilon_{un})=\sigma_u^2\), we have
\[
    \var(\Xi^\top W\one_U)
    = \sum_u w_u^2\sigma_u^2\,I_N.
\]
Therefore
\begin{align*}
    \var(\hat i)
    &= \frac{1}{N_w^2}
    M_N\left(\sum_u w_u^2\sigma_u^2\,I_N\right)M_N \\
    &= \frac{\sum_u w_u^2\sigma_u^2}{\left(\sum_u w_u\right)^2}M_N.
\end{align*}
Thus each coordinate has variance
\[
    \frac{\sum_u w_u^2\sigma_u^2}{\left(\sum_u w_u\right)^2},
\]
and any two distinct coordinates have covariance
\[
    -\frac{1}{N}\frac{\sum_u w_u^2\sigma_u^2}{\left(\sum_u w_u\right)^2}.
\]
To minimize the asymptotic variance, it suffices to minimize
\[
    \frac{\sum_u w_u^2\sigma_u^2}{\left(\sum_u w_u\right)^2}.
\]
Since the objective is scale-invariant, we may impose the normalization \(\sum_u w_u=1\), reducing the problem to
\[
    \min_{\{w_u\}_u} \sum_u w_u^2\sigma_u^2
    \qquad\text{subject to}\qquad
    \sum_u w_u=1.
\]
The first-order conditions give $ 2w_u\sigma_u^2=\lambda,$ so $w_u\propto 1/\sigma_u^2.$ Hence the variance is minimized by inverse-variance weights, proving $\avar(\hat i^\ts)\preceq \avar(\hat i)$. 

Finally, define the feasible weights $\hat w_u=1/\hat\sigma_u^2$, and let \(\tilde i\) denote the solution to \eqref{eq:weighted-opt} with weights \(\hat w_u\). By Lemma \ref{lem:sigma-consistency}, there exist constants \(c_1,c_2>0\) such that with probability tending to one, $\hat\sigma_u^2\in[c_1,c_2]$ for all $u$. Since \(x\mapsto 1/x\) is Lipschitz on \([c_1,c_2]\),
\[
    |\hat w_u-w_u|
    = \frac{|\hat\sigma_u^2-\sigma_u^2|}{\hat\sigma_u^2\sigma_u^2}
    \le \frac{1}{c_1^2}|\hat\sigma_u^2-\sigma_u^2|.
\]
Hence, by \eqref{eq:sigma-bound} and the fact that \(\sigma_u^2\) are bounded away from \(0\) and \(\infty\),
\[
    \max_u |\hat w_u-w_u| = o_p(1).
\]
Moreover, by construction,
\[
    \sum_u \hat w_u \asymp U,
    \qquad
    \sum_u \hat w_u^2\sigma_u^2 = O_p(U).
\]
It therefore suffices to analyze
\[
    \left|
    \frac{\sum_u \hat w_u^2\sigma_u^2}{\left(\sum_u \hat w_u\right)^2}
    -
    \frac{\sum_u w_u^2\sigma_u^2}{\left(\sum_u w_u\right)^2}
    \right|.
\]
Adding and subtracting an intermediate term yields
\begin{align*}
    \left|
    \frac{\sum_u \hat w_u^2\sigma_u^2}{\left(\sum_u \hat w_u\right)^2}
    -
    \frac{\sum_u w_u^2\sigma_u^2}{\left(\sum_u w_u\right)^2}
    \right|
    &\le
    \left|
    \frac{\sum_u \hat w_u^2\sigma_u^2}{\left(\sum_u \hat w_u\right)^2}
    -
    \frac{\sum_u \hat w_u^2\sigma_u^2}{\left(\sum_u w_u\right)^2}
    \right| \\
    &\qquad
    +
    \left|
    \frac{\sum_u \hat w_u^2\sigma_u^2}{\left(\sum_u w_u\right)^2}
    -
    \frac{\sum_u w_u^2\sigma_u^2}{\left(\sum_u w_u\right)^2}
    \right|.
\end{align*}
For the second term,
\begin{align*}
    \left|
    \frac{\sum_u \hat w_u^2\sigma_u^2}{\left(\sum_u w_u\right)^2}
    -
    \frac{\sum_u w_u^2\sigma_u^2}{\left(\sum_u w_u\right)^2}
    \right|
    &=
    \frac{1}{\left(\sum_u w_u\right)^2}
    \left|
    \sum_u (\hat w_u^2-w_u^2)\sigma_u^2
    \right| \\
    &\le
    \frac{1}{\left(\sum_u w_u\right)^2}
    \sum_u |\hat w_u-w_u|\,|\hat w_u+w_u|\,\sigma_u^2 \\
    &\le
    \frac{C\,\max_u|\hat w_u-w_u|}{U}
    = o_p(1),
\end{align*}
for some deterministic constant \(C<\infty\), since \(\hat w_u\), \(w_u\), and \(\sigma_u^2\) are uniformly bounded.

For the first term,
\begin{align*}
    \left|
    \frac{\sum_u \hat w_u^2\sigma_u^2}{\left(\sum_u \hat w_u\right)^2}
    -
    \frac{\sum_u \hat w_u^2\sigma_u^2}{\left(\sum_u w_u\right)^2}
    \right|
    &=
    \left(\sum_u \hat w_u^2\sigma_u^2\right)
    \left|
    \frac{1}{\left(\sum_u \hat w_u\right)^2}
    -
    \frac{1}{\left(\sum_u w_u\right)^2}
    \right| \\
    &=
    \left(\sum_u \hat w_u^2\sigma_u^2\right)
    \frac{\left|\sum_u (w_u-\hat w_u)\right|\left|\sum_u (w_u+\hat w_u)\right|}
    {\left(\sum_u \hat w_u\right)^2\left(\sum_u w_u\right)^2} \\
    &= o_p(1).
\end{align*}
Combining the two bounds yields
\[
    \frac{\sum_u \hat w_u^2\sigma_u^2}{\left(\sum_u \hat w_u\right)^2}
    =
    \frac{\sum_u w_u^2\sigma_u^2}{\left(\sum_u w_u\right)^2}
    + o_p(1).
\]
Since \(w_u=1/\sigma_u^2\) minimizes the asymptotic variance among all admissible weights, it follows that \(\tilde i\) has the same asymptotic variance as the oracle inverse-variance weighted estimator. Therefore
\[
    \avar(\tilde i)=\avar(\hat i^\ts)\preceq \avar(\hat i).
\]
\end{proof}

\subsection{Discussion of the Two-Stage Result}

The two-stage theorem is a statistical result, not a full incentive-compatibility result. It shows that under private-signal reporting, inverse residual variance weighting yields a consistent estimator of note helpfulness and improves efficiency relative to alternative bounded weight choices. 

What the theorem does \emph{not} show is that residual-variance weighting uniquely induces truthful reporting in equilibrium. The argument that such a rule may weaken conformity incentives is heuristic: unlike consensus-based auditing, it does not directly reward agreement with the eventual majority outcome. This makes it a plausible alternative auditing principle, but a full strategic analysis of that induced game lies beyond the scope of the present theorem.

For this reason, we treat the two-stage rule in the main text as an alternative estimator and auditing principle motivated by predictive stability. The empirical results show that it improves out-of-sample prediction, and the theorem shows that it is statistically well behaved under informative reports. The stronger incentive question is left open.

\end{document}